\documentclass[14pt, a4paper, onecolumn ]{IEEEtran}
\usepackage{times}
\usepackage{graphicx}
\usepackage{amssymb}
\usepackage{amsmath}
\usepackage{wrapfig}
\usepackage{multirow}
\usepackage{epsfig}
\usepackage{natbib}
\usepackage{algorithmic}
\usepackage{algorithm}

\usepackage[english]{babel}
\usepackage[font=small,labelfont=bf]{caption}

\usepackage{helvet}
\usepackage{courier}

\usepackage{authblk}
\usepackage{url}
\usepackage{etoolbox}

\newtheorem{conjecture}{\bf Conjecture}

\begin{document}


\title {\Large{Topologies and Price of Stability of Complex Strategic Networks with Localized Payoffs : Analytical and Simulation Studies}}


\author[$\dagger$,*]{Rohith Dwarakanath Vallam,}
\author[$\dagger$]{\hspace{-0.1in},\hspace{0.1in}C. A. Subramanian}
\author[$\ddagger$]{\hspace{-0.1in},\hspace{0.1in}Ramasuri Narayanam}
\author[$\dagger$]{\hspace{-0.1in},\hspace{2.7in} Y Narahari}
\author[$\dagger$]{\hspace{-0.1in},\hspace{0.1in}Srinath Narasimha}
\affil[$\dagger$]{Indian Institute of Science, Bangalore, India}
\affil[$\ddagger$]{IBM India Research Lab, Bangalore, India}

\maketitle

\let\oldthefootnote\thefootnote
\renewcommand{\thefootnote}{\fnsymbol{footnote}}
\footnotetext[1]{Any correspondence can be addressed to \url{rohithdv@gmail.com} }
\let\thefootnote\oldthefootnote

\begin{abstract}
Several real-world networks exhibit a complex structure and are formed due to strategic interactions among rational and intelligent individuals. 
In this paper, we analyze a network formation game in a strategic setting where  payoffs of individuals depend only on their immediate neighbourhood. We call these payoffs as localized payoffs. In this network formation game, the payoff of each individual captures (1) the gain from immediate neighbors, (2) the bridging benefits, and (3) the cost to form links. This implies that the payoff of each individual can be computed using only its single-hop neighbourhood information. Based on this simple and appealing model of network formation, our study explores the structure of networks that form, satisfying one or both of the properties, namely, pairwise stability and efficiency. 
We analytically prove the pairwise stability of several interesting network structures, notably, the complete bi-partite network, complete equi-k-partite network, complete network and cycle network, under various configurations of the model. We validate and further extend these results through extensive simulations. We then characterize topologies of efficient networks by drawing upon classical results from extremal graph theory and discover that the Turan graph (or the complete equi-bi-partite network) is the unique efficient network under many configurations of parameters. We next examine the tradeoffs between topologies of pairwise stable networks and efficient networks using the notion of price of stability, which is the ratio of the sum of payoffs of the players in an optimal pairwise stable network to 
that of an efficient network. Interestingly, we find that price of stability is equal to
$1$ for almost all
configurations of parameters in the proposed model; and for the rest of the configurations of the parameters, we obtain a lower bound of $0.5$ on the price of stability. This leads to another key insight of this paper: under mild conditions, efficient networks will form when  strategic individuals choose to add or delete links based on only localized payoffs. 
\end{abstract}




\newtheorem{theorem}{Theorem}
\newtheorem{lemma}{Lemma}
\newtheorem{proposition}{Proposition}
\newtheorem{corollary}{Corollary}
\newtheorem{example}{Example}
\newtheorem{definition}{Definition}


\section{Introduction}
\label{introduction}
Several real world networks such as the Internet, social networks, organizational networks,
biological networks, food webs, co-authorship networks, citation networks, and many more
exhibit complex network structures.
Complex networks, generally modeled as graphs in most of the mathematical literature,
have been extensively studied in recent years and they are pervasive in today's science and technology \cite{barrat:08, newman:06, strogatz:01,newman:03}. Studying the properties of the complex network structures helps to understand the underlying phenomena and developing new insights into the system such as small-world phenomena, scale-free topology, and structural holes \cite{watts:98,albert:02,newman:03,song:05, burt1}.

Complex networks have also been studied extensively in the social
sciences \cite{ newman:03, easley:10, brandes:05, wasserman:94} (and the references therein). These studies reveal that 
complex social networks play an important role in spreading information
\cite{boorman:75, schelling:78, rogers:95, cooper:82, valente:95, strang:98}.
Individuals that participate in the process of information dissemination in such networks
receive various kinds of social and economic incentives and at the same time they also
incur costs in forming and maintaining the contacts (i.e. links) with other individuals
in terms of time, money, and effort. For this reason, individuals do act strategically while
selecting their neighbors. Thus, in several contexts, the behavior of the system is driven by
the strategic actions of a large number of individuals, each motivated by self-interest
and optimizing an individual objective function. Thus, it is important to study the dynamics of strategic
interaction among the individuals in complex social networks in order to understand
how such networks form and this is the primary motivation for this paper.


Many recent studies on network formation have used game theoretic approaches \cite{myerson:91, jackson:08, goyal:07, demange:05, slikker:01,dutta:00,dvt:98,borgs:11,brautbar:11} based on the observation that individuals are strategic and are
interested in maximizing their payoffs from
the social interactions. These models capture the strategic interactions among
individuals and the analysis of these models satisfactorily deduces the topologies of equilibrium networks. In this domain, networks that are enforced by a central authority
are known as efficient networks. Understanding the compatibility between
equilibrium networks and efficient networks has been the primary focus of research
in network formation \cite{elias:11, jackson:08, goyal:07, hummon00,
doreian06, corbo:05, galeotti:06, jackson:02}.

The crux of most of the models for network formation in the
literature \cite{elias:11, jackson-wolinsky:96, anshelevich:03, anshelevich:08,
fabrikant:03, corbo:05, galeotti:06} is the underlying
strategic form game where the players, strategies, and utilities (also termed as payoffs) are defined as follows: (i) the
individual agents in the complex network are the players, (ii) the strategy of each agent
is a subset of other agents with which it wishes to form links, and
(iii) the utility of each agent depends on the structure of the
network.

Another key aspect of most of the existing work in the literature
is that the process of network formation is  modeled in a decentralized fashion
where the individuals in the network take autonomous decisions regarding
whether to form or delete links with other agents.
However, most of these models require the agents to know the complete global
structure (that is, information about all nodes as well as all the
links between the nodes) of the network to compute their respective payoffs.
In many practical scenarios, this will be a very demanding requirement
making the utility computation a cumbersome and often intractable task.
Moreover, empirical evidence \cite{burt1, burt2} has clearly shown that
a significant fraction of the perceived social and economic
benefits for the individuals is derived from their $1$-hop or $2$-hop neighborhood.
Motivated by this, a few models of network formation
have been investigated that use local information (such as information about $1$-hop or $2$-hop neighborhood). For instance, Kleinberg and co-authors \cite{kstw:08} propose a
network formation model where the utility function of each node is based
on $2$-hop neighborhood information. However, in several real-world examples,
we observe that complete knowledge about $2$-hop information may be infeasible and nodes
may need to get a reasonably accurate estimate of their payoffs by using
just their immediate neighborhood (or $1$-hop) information.
In fact, we can observe such constraints in several real-world examples
like distributed sensor networks and real-life social networks.
In distributed sensor networks, coalitions of sensors can work together to
track targets of interest and each sensor knows only its immediate neighborhood.
In real-life social networks, it may not be possible for an individual to
know all the friends of his/her immediate friends. Note that individuals can know partial information about their $2$-hop neighborhood (i.e. friends of friends); however, this partial information is inadequate to accurately compute the payoffs of the individuals. Hence, in such settings, it becomes important to
study the network formation process using only single hop neighborhood information and this is the primary motivation behind our work in this paper.

In this paper, we explore a novel model of network formation process from
an economic perspective in which individuals derive payoffs
(consisting of benefits from immediate neighbors as well as structural holes and the costs to form links)
using purely local neighbourhood information and we refer to this setting as {\em network formation with localized payoffs}. The primary contribution of
our work is to come up with a game theoretic model in the above setting
and study the topologies of the equilibrium networks and efficient networks
that emerge in such a network formation process.
We next examine the tradeoffs between topologies of equilibrium networks and
efficient networks using the notion of price of stability \cite{anshelevich:08}.
Informally, price of stability is the ratio of the sum of payoffs of the players
in an optimal (in terms of sum of payoffs of the players) pairwise stable network to
that of an efficient network. Interestingly, we find that price of stability is $1$ for almost all
configurations of the parameters in the proposed model; and for the rest of
the configurations of the parameters in the proposed model, we obtain a lower bound of $0.5$ on price of stability. This indicates that, when some
mild conditions are satisfied, efficient networks will form when strategic individuals
choose to add or delete links based on localized payoffs.

We note that our model assumes that a link forms with the consent
of both the individuals (refer to Section~\ref{utilitymodel}), as
social contacts usually emerge in this manner. This assumption is widely considered in several models of network formation in the literature \cite{doreian06, jackson-wolinsky:96, hummon00, jackson:03, xie:08a,xie:08b}. In such situations, an appropriate choice for the notion of equilibrium is \textit{pairwise
stability} \cite{jackson-wolinsky:96}. Informally, we call a
network pairwise stable if no agent can improve its utility by
deleting any link and no two unconnected individuals can form a link to
improve their respective payoffs. We call a network {\em efficient} if
the sum of payoffs of the individuals is maximal. In this framework,
our objective is to investigate the tradeoff between topologies
of pairwise stable and efficient networks. In the rest of the
paper, we use the terms {\em graph\/} and {\em network\/}
interchangeably. We thus use the terms nodes and individuals interchangeably throughout the paper. As a game-theoretic approach is used, we sometimes use the terms players and individuals interchangeably throughout the paper.


\subsection{Relevant Work}
\label{relevant-work}



The field of network formation has been extensively studied in diverse fields such as sociology, physics, computer science,
economics, mathematics and biology 
\cite{jackson:08, goyal:07, demange:05, slikker:01,hummon00,doreian06, buskens-vanderijt:07,goyal-vegaredondo:07, kstw:08,johnson:00, bloch:07, calvo:00, dvt:98, dutta:00, galeotti:06, jackson-wolinsky:96, jackson:02, jackson:05, jackson:05b, jackson:03, doreian08a, doreian08b, borgs:11, brautbar:11}. 
In this section, we have included a discussion of the
models that are most relevant to our work.

The modeling of strategic formation in a general network setting
was first studied in the seminal work of Jackson and Wolinsky \cite{jackson-wolinsky:96}. They basically consider a value function and an allocation
rule model where the value function defines a value to each network
and the allocation rule distributes this value to the nodes
in the network. They investigate whether efficient networks will
form when self-interested individuals can choose to form links
and/or break links. The authors define two stylized models. For
these models, the authors observe that for high and low costs the
efficient networks are pairwise stable, but not always for medium
level costs. They also examine the tension between efficiency and
stability and derive various conditions and allocation rules for
which efficiency and pairwise stability are compatible. An important
feature their model does not capture is that of the intermediary
benefits that nodes gain by being intermediaries lying on the paths
between non-neighbor nodes. In particular, they do not capture the
benefits due to structural holes.

Hummon \cite{hummon00} carries out several interesting investigations
to unravel more specific topologies using a specific model proposed
by Jackson and Wolinsky \cite{jackson-wolinsky:96}. Two different agent-based
simulation approaches, the multi-thread model and the discrete
event simulation model, are used in the analysis done by Hummon \cite{hummon00}
to explore the dynamics of network evolution based on a model
proposed in Jackson and Wolinsky \cite{jackson-wolinsky:96}. Hummon identifies certain
pairwise stable structures that are more specific than those
anticipated by the formal analysis of Jackson and Wolinsky \cite{jackson-wolinsky:96}.
Doreian \cite{doreian06} explores the same issue in a systematic manner and
establishes the conditions under which different pairwise structures
are generated. Some gaps in the analysis of Doreian \cite{doreian06}
are addressed by Xie and Cui \cite{xie:08a,xie:08b}.

Jackson \cite{jackson:03} reviews several models of network formation
in the literature with an emphasis on the tradeoffs between efficiency
with stability. This work also studies
the relationship between pairwise stable and efficient networks
in a variety of contexts and under three different definitions of
efficiency. A later paper by Jackson \cite{jackson:05} presents a family of allocation
rules (for example, networkolus) that incorporate information
about alternative network structures when allocating the network
value to the individual nodes. The author provides a general
method of defining allocation rules in network formation games.

Goyal and Vega-Redondo \cite{goyal-vegaredondo:07} propose a non-cooperative
game model in which a node $i$ can benefit from serving as an intermediary
between a pair of nodes $x$ and $y$. In their model, a node $i$
could lie on an arbitrarily long path between $x$ and $y$. The authors
assume, however, that the benefits from farther nodes are not subject
to decay. They also assume that the benefit of communication
between any pair of nodes is always $1$ unit. This $1$ unit is distributed
to the two communicating nodes and only to certain so called essential nodes \cite{goyal-vegaredondo:07} on the paths between
the two communicating nodes. In this setting, the authors show
that a star graph is the only non-empty robust equilibrium graph.
The authors also study the implications of capacity constraints in
the ability of individual nodes to form links to other nodes and show
that a cycle network emerges.

Ramasuri and Narahari \cite{ramasuri:11} propose a generic model of network formation that essentially builds on the model of
Jackson-Wolinsky \cite{jackson-wolinsky:96}. This model simultaneously captures four key determinants of network formation:
(i) benefits from immediate neighbors through links, (ii) costs of maintaining
the links, (iii) benefits from non-neighboring nodes and decay of these benefits with distance, and (iv)
intermediary benefits that arise from multi-step paths. The authors \cite{ramasuri:11} analyze the proposed model to determine
the topologies of stable and efficient networks.

The aforementioned models of network formation have the limitation
that each individual (or node) needs
to know global information about the structure of the network
in order to compute its utility. A few
recent models \cite{buskens-vanderijt:07, arcaute:08, kstw:08} in the literature
make an attempt to overcome the above limitation.
\begin{itemize}

\item Buskens and van de Rijt \cite{buskens-vanderijt:07} propose a model that requires each individual agent to know
just its immediate neighbors (or $1$-hop neighborhood) to optimize its own utility.
However, the model captures only the cost to nodes and ignores various benefits that
nodes can derive from the network such as direct benefits from the neighbors
and the bridging benefits.

\item Arcaute, Johari, and Mannor \cite{arcaute:08} study the myopic dynamics in network
formation games. A key aspect of the dynamics studied in this model
is the local information and the authors show that these dynamics converge to
efficient or near efficient outcomes. However, the model  does not characterize
the topologies of equilibrium and efficient networks. Moreover, the model
works with Pareto efficiency whereas we work with a more natural notion
of efficiency, namely maximizing the sum of payoffs of all the nodes.

\item Kleinberg and co-authors \cite{kstw:08} characterize the
structure of stable networks with {\em Nash equilibrium\/} as the notion of
stability. The authors propose a polynomial time algorithm for a node
to determine its best response in a given graph as nodes can choose to
link to any subset of other nodes. They also show that stable networks have
a rich combinatorial structure. However, the model  needs each individual agent
to know its $2$-hop neighborhood (the set of all individuals that are
reachable within two hops) to compute and optimize its own utility.
The model works with Nash equilibrium while our proposed model works with the more natural
notion of pairwise stability as the notion of equilibrium. Also, our model considers only
single hop neighbourhood which is more appropriate for certain kinds of social networks
as already explained. Moreover, the model \cite{kstw:08}
does not study the tradeoff between the topologies of stable networks and
the topologies of efficient networks.
\end{itemize}


\subsection{Our Contributions}
\label{results}

To the best of our knowledge, our current study is the first one to
comprehensively explore the tradeoff between pairwise stability and efficiency using the notion of price of stability in the context of strategic network formation with localized payoffs, while taking
into account several key factors such as link costs, link benefits, and bridging benefits.
The following are the specific contributions of our paper.
\begin {itemize}
\item \textit{Section~\ref{utilitymodel}: An Elegant Model for Network Formation with Localized Payoffs:}
We propose a strategic form game to model the process of
network formation with localized payoffs and we term the game as {\em network formation (game) with localized payoffs} (NFLP). The utility of each player in the proposed game takes into account not only the benefits ($\delta$) that arise from routing information to and from its neighbors but also the cost ($c$) to maintain a link to each of its neighbors. 

\item \textit{Section~\ref{sec:Stability}: Analytical Characterization of Topologies of Pairwise Stable Networks:}
We first analytically characterize the topologies of the pairwise stable networks using the NFLP model. Some of the networks that we consider for analysis include the cycle, star, complete and null networks. In addition, we also derive pairwise stability conditions for certain classes of k-partite networks namely bipartite complete networks, complete equi-tri-partite networks and complete equi-k-partite networks. We note that our findings extend the possible topologies for pairwise stable networks compared to that of other models in the
literature.

\item \textit{Section~\ref{sec:Simulations}: Simulation of Network Formation Process and Additional Insights:}
Next, we simulate strategic dynamics in NFLP to understand how pairwise stable networks evolve over time. Our simulation results validate our analytical deductions and also reveal additional interesting insights on the topologies of pairwise stable networks.
In addition, we study the emergent pairwise stable topologies during the network formation process and study the evolution of pairwise stable network and its properties like the clustering co-efficient, convergence time, etc. over different configuration parameters.

\item \textit{Section~\ref{sec:Efficiency}: Analytical Characterization of Topologies of Efficient Networks:}
Next, we analytically characterize topologies of efficient networks by drawing upon classical results from extremal graph theory. Our work leads to sharp deductions about the efficient networks in NFLP. A striking discovery of our study here is that the equi-bi-partite graph (popularly known as the Turan graph) emerges as the unique efficient network under many regions of values of $\delta$ and $c$.

\item \textit{Section~\ref{POS}: Price of Stability Investigations:}
The quality of optimal (in terms of the sum of payoffs of the individuals in the network) pairwise stable networks is best understood through the notion of price of stability (PoS). PoS allows us to explore the middle ground between centrally enforced solution and completely unregulated anarchy \cite{anshelevich:08}. In most real-world applications, the nodes are not completely unrestricted in their strategic behavior but rather agree upon a prescribed equilibrium solution. In such scenarios, the prescription can be chosen to be the best equilibrium thus making the price of stability an important issue to study. We study the PoS in NFLP to reveal tradeoffs between pairwise stable networks and efficient networks. Intriguingly, we find that PoS is $1$ for almost all configurations of $\delta$ and $c$. For the remaining configurations of $\delta$ and $c$, we obtain a lower bound of $\frac{1}{2}$ on PoS. This implies, under mild conditions on $\delta$ and $c$, that the proposed NFLP model produces pairwise stable networks that are efficient. 
\end {itemize}

%


\section{A Model for Network Formation with Localized Payoffs}
\label{utilitymodel}

We model network formation using a
strategic form game \cite{myerson:91}. We consider a network setup with
$n$ players denoted by $N=\{1,2,\ldots, n\}$. A strategy $s_i$ of a player $i$ is any subset of players with which
the player would like to establish links. We assume that the
formation of a link requires the consent of both the players.
Assume that $S_i$ is the set of strategies of player $i$. Let
$s=(s_1, s_2, \ldots, s_n)$ be a profile of strategies of the
players. Also let $S$ be the set of all such strategy profiles.
Each strategy profile $s$ leads to an undirected graph and we
represent it by $G(s)$. If there is no confusion, we just use $G$.
If players $x$ and $y$ form a link $(x,y)$ in a graph $g$, then we
represent the new graph by $g + (x,y)$. We assume that players in
the network communicate using shortest paths - this is a standard
assumption used in the literature for ease of modeling. In the
rest the paper, we use the terms players, nodes, and agents
interchangeably.

{\em Degree of Node:} The degree $d_i$ of node $i$ represents the number of neighbors
of node $i$.

{\em Costs:} If nodes $i$ and $j$ are connected by a link, then we
assume that the link incurs a cost $c \in (0,1)$ to each node. That is,
if the degree of node $i$ is $d_i$, then node $i$ incurs a cost of
$cd_i$.

{\em Benefits from Immediate Neighbors:} Assume that $\delta \in (0,1)$. If node $i$ is connected to a node $j$ by a direct link, then we
assume that node $i$ gains a benefit of $\delta$.
That is, if the degree of node $i$ is $d_i$, then node $i$ gains a
benefit of $\delta d_i$ from its immediate neighbors.

{\em Bridging Benefits:} Consider a node $i$. Assume that nodes $j$ and $k$ are two neighbors of node $i$ such that $j$ and $k$ are not connected by a direct link.
Suppose that nodes $j$ and $k$ communicate using the length $2$ path through node $i$, then (i) we assume that a benefit of
$\delta^{2}$ arises due to this communication, and (ii) we also
assume that the benefit $\delta^{2}$ entirely goes to node $i$. We
refer to $\delta^{2}$ as the bridging benefit to node $i$. The
main motivation for this kind of bridging benefits is by
sociological studies suggesting that in practice most of the
bridging benefits arise from bridging the communication between
pairs of non-neighbor nodes in the network \cite{burt:07}.

In this framework, we define the utility of node $i$ such that it
depends on the benefits from immediate neighbors, the costs to
maintain links to these immediate neighbors, and the bridging
benefits. More formally, for any $i \in N$, the utility $u_i$ of
node $i$ in an undirected graph $G$ is defined as follows:
\begin{align}\label{proposedutilitymodel}
u_i(G) &= d_i(\delta - c) + d_i \Biggl(1-\frac{\sigma_i}{{d_i \choose 2}}\Biggr) \delta^2
\end{align}
where $\sigma_i$ is the number of links among the neighbors of
node $i$ in $G$. There are two terms in this utility function. The
first term specifies the net benefit to node $i$ from its
immediate neighbors. The second term specifies the sum of bridging
benefits to node $i$. Here $1-\frac{\sigma_i}{{d_i \choose 2}}$ is
the fraction of pairs of neighbors of node $i$ that are
non-neighbors and $d_i$ normalizes the level of bridging benefits
that node $i$ gains in the network. 
\begin{figure}[!hbtp]
\begin{center}
\includegraphics[scale=0.357]{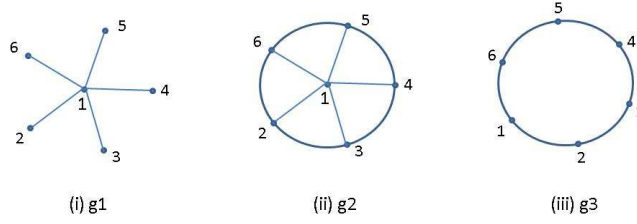}
\end{center}
\caption{An illustrative example
\label{model-illustration}}
\end{figure}
For example, the fraction of pairs of neighbors of node $1$ that are non-neighbors in both $g1$
and $g3$ in Figure~\ref{model-illustration} is $1.0$. However the degree of node $1$ in $g1$ is $d_1 = 5$ and the degree of node $1$
in $g3$ is $d_1 = 2$. The normalization term $d_i$ ensures that the bridging benefit for node $1$ is higher in $g1$ than in $g3$. 
Note that the bridging benefit of our proposed model can also be altered by introducing an arbitrary increasing, real-valued function of $d_i$ (call it $f(d_i)$). In this case, the utility model (Equation~\ref{proposedutilitymodel}) becomes as follows: 
\begin{align}
\nonumber u_i(G) &= d_i(\delta - c) + f(d_i) \Biggl(1-\frac{\sigma_i}{{d_i \choose 2}}\Biggr) \delta^2.
\end{align}
For ease of analysis, we work with $f(d_i)=d_i$ throughout this paper.

Note: Assume that node $i$ bridges the communication between
$j$ and $k$; and a benefit of $\delta^{2}$ is generated. In the
literature, there are three well known ways of distributing the
benefit $\delta^{2}$ to nodes $i$, $j$, and $k$: (i) only node $i$
gets entire $\delta^{2}$, (ii) node $i$ gets $0$, and (iii) nodes
$i$, $j$, and $k$ get equal share of $\delta^{2}$. In this paper,
we work with scenario (i). A similar approach is utilized in
\cite{kstw:08} as well. We note that the analysis that we perform
using scenario (i) can be easily extended to other two scenarios.

\subsection{The Network Formation Game}
The above framework defines a strategic form game $\Gamma =
\Bigl(N, (S_i)_{i \in N}, (u_i)_{i \in N} \Bigr)$ that models
network formation with localized payoffs. We refer to this as
network formation game with localized payoffs (NFLP). The
following example illustrates NFLP.

\begin{example}
\label{nfgl-illustration}
Assume that $N=\{1,2,3,4,5,6\}$ is the
set of $6$ players. If $s_1 = \{2,3,4,5,6\}$, $s_2 = \{1\}$, $s_3
= \{1\}$, $s_4 = \{1\}$, $s_5 = \{1\}$, $s_6 = \{1\}$, then the
resultant graph $g1$ is the star graph as shown in Figure
\ref{model-illustration}.(i). Note that an edge forms with the
consent of both the nodes.

Following the NFLP model, the payoffs of the players in the star graph are as follows:
$u_1(g1)=5(\delta-c) + 5\delta^{2}$ and $u_2(g1) = u_3(g1) =
u_4(g1) = u_5(g1) = u_6(g1) = (\delta-c)$.

If $s_1 = \{2,3,4,5,6\}$, $s_2 = \{1,3,6\}$, $s_3 = \{1,2,4\}$,
$s_4 = \{1,3,5\}$, $s_5 = \{1,4,6\}$, $s_6 = \{1,2,5\}$, then the
resultant graph $g2$ is the wheel graph as shown in Figure
\ref{model-illustration}.(ii). Following the NFLP model, the payoffs of the players in the
wheel graph are as follows: $u_1(g2)=5(\delta-c) +
\frac{5\delta^{2}}{2}$ and $u_2(g2) = u_3(g2) = u_4(g2) = u_5(g2)
= u_6(g2) = 3(\delta-c) + \delta^{2}$.

On similar lines, if $s_1 = \{2,6\}$, $s_2 = \{1,3\}$, $s_3 =
\{2,4\}$, $s_4 = \{3,5\}$, $s_5 = \{4,6\}$, $s_6 = \{1,5\}$, then
the resultant graph $g3$ is the cycle graph as shown in
Figure~\ref{model-illustration}.(iii). Following the NFLP model, the payoffs of the
players in the cycle graph are as follows: $u_1(g3) = u_2(g3) =
u_3(g3) = u_4(g3) = u_5(g3) = u_6(g3) = 2(\delta-c) +
2\delta^{2}$.
\end{example}

\section{Analytical Deductions on Topologies of Pairwise Stable Networks}\label{sec:Stability}

In this section, we first recall the notion of pairwise stability.
Then, we characterize the topologies of pairwise stable networks.
To begin with, we note that the notion of pairwise stability is
defined by Jackson and Wolinsky~\cite{jackson-wolinsky:96}. Formally,
we call an undirected graph \text{$G=(V,E)$} pairwise
stable~\cite{jackson-wolinsky:96} if (i) $\forall (i,j) \in E,
u_i(G) \geq u_i(G-(i,j))$ and $\text{$u_j(G) \geq u_j(G-(i,j))$}$,
(ii) $\forall (i,j) \notin E$, if $u_i(G) < u_i(G+(i,j))$ then
$u_j(G) > u_j(G+(i,j))$.

We now focus on characterizing the topologies of the pairwise
stable networks that may emerge following the framework in NFLP.
Characterizing pairwise stable networks under various network formation models has been addressed in the
literature \cite{jackson:08}, \cite{goyal:07}, \cite{buskens-vanderijt:07},
\cite{goyal-vegaredondo:07}, \cite{kstw:08}, \cite{fabrikant:03},
\cite{corbo:05}, \cite{galeotti:06}, \cite{jackson-wolinsky:96}, \cite{doreian06}, \cite{doreian08a}, \cite{doreian08b}.  
In our approach, we consider the topologies of certain standard
networks (such as complete network, cycle network, star
network, multi-partite networks) and then study whether such
topologies are pairwise stable following the framework of NFLP. We now present few results to establish certain standard networks are pairwise stable in the framework of NFLP.


\begin{proposition}
\label{lem:StabilityConditions1} If $(\delta-c)\leq\delta^{2}$ and
$(c-\delta)\leq\delta^{2}$, then the complete bipartite network is
pairwise stable.
\end{proposition}

\begin{proof}
%
%
%
%

Consider a complete bipartite network, $G$,  with $a_1$ and $a_2$
nodes respectively in the two partitions. The utility of node $i$
in a partition with $a_1$ nodes is
$u_{i}(G)=a_{2}(\delta-c)+a_{2}\delta^{2}$. This proposition can be
proved in two steps.\\
{\em Step 1:} Let us now add the edge $(i,j)$ to $G$ and call the
resultant graph $\overline{G}$. It can be readily checked that
$u_i(\overline{G}) = (a_{2}+1)(\delta-c)+(a_{2}-1)\delta^{2}$.
Since we are given that $\delta^{2} \geq (\delta-c)$, we get that
$u_i(G) = a_{2}(\delta-c)+a_{2}\delta^{2} \geq
(a_{2}+1)(\delta-c)+(a_{2}-1)\delta^{2} = u_i(\overline{G})$. That
is, no pair of non-neighbor nodes is better off by forming a
link in $G$.\\
{\em Step 2:} Assume that node $i$ severs an edge in $G$ and call
the resultant graph $\hat{G}$. It can be shown that
$u_{i}(\hat{G})=(a_{2}-1)(\delta-c)+(a_{2}-1)\delta^{2}$. Since we
are given that $\delta^{2} \geq (\delta-c)$, it is immediately seen
that $u_i(G) \geq u_i(\hat{G})$. Node $i$ is not better off by
severing a link in $G$.

Note that we can apply similar analysis with respect to each node
in the other partition. Hence the complete bipartite network is
pairwise stable.
\end{proof}

\begin{proposition}
\label{lem:StabilityConditions} (a) The complete network is pairwise stable if
$(c-\delta)\leq0$ (b) The cycle network is pairwise stable if
$1\leq(c-\delta)/\delta^{2}\leq2$, (c) The null (empty) network is
pairwise stable if $(\delta-c)\leq0$.
\end{proposition}
The result can be proved easily by using arguments similar to
that in Proposition~\ref{lem:StabilityConditions1}.
\begin{proposition}\label{kpartite-result}
For $k \geq 3$, the complete $k$-partite network is pairwise
stable if (i) $\delta = c$, and (ii) $a_i=a, \forall i \in \{
1,2,...,k\}$ where $a_i$ is the number of nodes in partition $i$
in \text{$k$-partite} network and $a$ is any positive integer.
\end{proposition}
\begin{proof}
We start with a $k$-partite graph, $G$, satisfying condition (ii)
given in the statement of this proposition. Consider a node $i$ in the
$p^{th}$ partition of $G$ where $1\leq p \leq k$. We construct the
proof in two steps.

\textit{Step 1 (edge addition): } We can see that, in $G$,  the
only link that can be added from node $i$ is to a node $j$ in the
$p^{th}$ partition. Let $\overline{G}$ be the network obtained
after a new link $(i,j)$ is added to $G$. For pairwise stability,
we need $u_{i}(\overline{G}) - u_{i}(G) \leq 0 $. This implies
\begin{align}
\nonumber (\delta-c) + (d_i +1) \delta^2 \Biggl( 1 - \frac{\sigma_i^{'}}{\binom{d_i+1}{2}} \Biggr) - d_i\delta^2 \Biggl( 1 - \frac{\sigma_i}{\binom{d_i}{2}} \Biggr) &\leq 0
\end{align}
where $\sigma_i^{'}$ is the number of links among the neighbours
of node $i$ in $\overline{G}$ and $\sigma_i$ is the number of
links among the neighbours of node $i$ in $G$. Note that $d_i =
d_j$ since nodes $i$ and $j$ belong to the same partition in $G$.
Now we get that $\sigma_i^{'} = \sigma_i + d_j = \sigma_i + d_i $.
Simplifying, we get
\begin{align}
\label{lamma-ps-eqn} u_{i}(\overline{G}) - u_{i}(G) = (\delta-c) -
\delta^2 + \delta^2 \Biggl(\frac{2 \sigma_i}{d_i(d_i-1)}\Biggr)
\end{align}
Since the term $\displaystyle\frac{2\sigma_i}{d_i(d_i-1)}$ lies in
the interval $[0,1]$ and the fact that $\delta = c$ (given in the
statement of this proposition), we get that expression
(\ref{lamma-ps-eqn}) is non-positive. This implies that no pair of
nodes can form a link to improve their respective payoffs.

\textit{Step 2 (edge deletion): } In $G$, consider that node $i$
deletes a link to a node $j$ in the $q^{th}$ partition where $1
\leq q \leq k$ and $p\neq q$. Let $\overline{G}$ be the network
obtained after the link $(i,j)$ has been deleted from $G$. For
pairwise stability, we need $u_{i}(\overline{G}) - u_{i}(G) \leq 0
$. This implies
\begin{align}
\nonumber -(\delta-c) + (d_i -1) \delta^2 \Biggl( 1 - \frac{\sigma_i^{'}}{\binom{d_i-1}{2}} \Biggr) - d_i\delta^2 \Biggl( 1 - \frac{\sigma_i}{\binom{d_i}{2}} \Biggr) &\leq 0
\end{align}
where $\sigma_i^{'}$ denotes the number of links among the neighbours of node $i$ in $\overline{G}$. We can see that $\sigma_i^{'} = \sigma_i -d_j + a_i$. Simplifying,
\begin{align}\label{deletionedge}
-(\delta-c) - \delta^2 + \delta^2 \underbrace{\Bigl(\frac{-2 \sigma_i + 2d_j-2a_i}{d_i-2} + \frac{2 \sigma_i}{d_i -1}\Bigr)}_{expr_1} \leq 0
\end{align}
\textit{Claim:} $expr_1 \leq 1$.
\\ \smallskip \noindent {\textit{Proof of the Claim:}}
We know that $d_i = \sum_{j\neq i} a_j$ . Now, we derive an expression for $\sigma_i$.
\begin{align}\label{sigmai}
\sigma_i &= \binom{d_i}{2} - \sum_{j\neq i} \binom{a_j}{2}
&= \frac{d_i(d_i-1)}{2} - \frac{1}{2} \Biggl( \sum_{j\neq i} a_j^2 - \sum_{j\neq i} a_j \Biggr)
&= \frac{d_i^2 - \sum_{j \neq i} a_j^2}{2}
\end{align}
Now, we show that $expr_1 \leq 1$. The proof is by contradiction. Suppose $expr_1  > 1$.
\begin{align}\label{simplify1}
\nonumber \Bigl(\frac{-2 \sigma_i + 2d_j-2a_i}{d_i-2} + \frac{2 \sigma_i}{d_i -1}\Bigr) &> 1 \\
\nonumber 2(d_j-\sigma_i-a_i)(d_i-1) + (2\sigma_i)(d_i-2) &> (d_i-2)(d_i-1)\\
(2d_jd_i-2\sigma_i-2a_id_i-2d_j+2a_i) &> (d_i^2 -3d_i+2)
\end{align}
From condition \textit{(2)} in Proposition~\ref{kpartite-result}, we have $a_i=1, \forall i$ and $d_i=d_j=(k-1)a$. Also, using Equation~(\ref{sigmai}) in Equation~(\ref{simplify1}) and simplifying, we have
\begin{align}\label{simplify2}
(k+1)a - (k-1)a^2 &>2 \\
\nonumber \Rightarrow (k+1)a &> 2+ (k-1)a^2 > (k-1)a^2 \\
\nonumber \Rightarrow a &< \Biggl( \frac{k+1}{k-1}\Biggr)
\end{align}
Let $y(k) = \bigl(\frac{k+1}{k-1}\bigr)$. As we know that the function $y(k)$ is a decreasing function of $k$ (as derivative of $y(k)$ with respect to $k$   is $< 0$), we can write
\begin{align}
\nonumber a &< y(2) \Rightarrow a < 3
\end{align}
So, clearly we can conclude that $expr_1 > 1$ for $0< a < 3$ (i.e., $a=2$ and $a=1$) and $expr_1 \leq 1$ for $a \geq 3$.

Now we will examine what happens when $a=1$ and $a=2$. Substituting $a = 1$ in Equation~(\ref{simplify2}) and simplifying, we get $2>2$ which is absurd. Substituting $a = 2$ in Equation~(\ref{simplify2}) and simplifying, we get $k<2$ which violates the hypothesis that $k \geq 3$. Hence, by the above arguments, \text{$expr_1 \leq 1, \forall a \in \{1,2, ...\}, \forall k \geq 3$}. This completes the proof of the claim.

Note that we are given that $\delta=c$. Thus, from Equation~(\ref{deletionedge}),
\begin{align}\label{deletionedge2}
\nonumber - \delta^2 + \delta^2 \underbrace{\Bigl(\frac{-2 \sigma_i + 2d_j-2a_i}{d_i-2} + \frac{2 \sigma_i}{d_i -1}\Bigr)}_{\leq 1} \leq 0 \;\;
\Rightarrow  \;\; u_{i}(\overline{G}) - u_{i}(G) &\leq 0
\end{align}

Thus, node $i$ does not have any incentive to add an edge to $G$
or delete an edge from $G$ when the conditions given in the
statement of the proposition are satisfied. As node $i$ is chosen
arbitrarily from $G$, we have that $G$ is pairwise stable.
\end{proof}

Using a similar approach, we can prove the stability results for other standard networks. We summarize these results in Table \ref{summarytable2}\footnote{Note that the legends in the figure correspond to the numbering specified in Table~\ref{summarytable2}} and the graphical illustration of these results is depicted in Figure~\ref{summaryfigure3}.

\begin{figure*}[h]
\begin{tabular}{cc}
\begin{minipage}{8cm}
\centering
\scriptsize
\begin {tabular} {||l||l||l||}
\hline
\hline
{\textbf{Parameter }}  &  {\textbf{Additional }} &{\textbf{P.S.}\footnotemark[1]} \\
{\textbf{Region}}  & {\textbf{Conditions}}&{\textbf{networks}} \\
\hline
\hline
\multirow{5}{*}
& $\textbf{(1a) }(\delta - c) \geq \delta^2$ & Complete \tabularnewline
\cline{2-3}
$\textbf{(1) }\delta > c$ & $\textbf{(1b) }(\delta - c) < \delta^2$ & Complete \\
& & C.B.P \footnotemark[4] \\
\cline{2-3}
& $\textbf{(1c) } (\delta - c) < 2/3\delta^2$ & C.E.T.P \footnotemark[6] \\
& & Complete \\
& & C.B.P \\
\cline{2-3}
\hline
\multirow{4}{*}
& & Complete, Null, \\
\textbf{(2) }$\delta = c$ & & C.B.P, \\
 & & C.E.K.P\footnotemark[5]  \\
\hline
\multirow{5}{*}
& $\textbf{(3a) } (c - \delta) > 2\delta^2$& Null \\
\cline{2-3}
& $\textbf{(3b) }(c - \delta) \leq \delta^2$  & C.B.P \\
& & Null \\
\cline{2-3}
\textbf{(3) }$\delta < c$ & $\textbf{(3c) }\delta^2 \leq (c - \delta) \leq 2\delta^2$  & Cycle \\
& & Null \\
\cline{2-3}
& $\textbf{(3d) }(c - \delta) < 2/3\delta^2$ & C.E.T.P \\
& & Null \\
& & C.B.P \\
\hline
\hline
\end {tabular}
\\
\footnotemark[1]{P.S: Pairwise Stable}
\footnotemark[4]{C.B.P: Complete BiPartite }

\footnotemark[5]{C.E.K.P: Complete Equi $K$-Partite }

\footnotemark[6]{C.E.T.P: Complete Equi Tri-Partite }
\captionof{table}{Characterization of pairwise stable \\network topologies in the proposed utility model}\label{summarytable2}
\end{minipage}
&
\begin{minipage}{7.5cm}
\begin{tabular}{l}
\centering
\epsfig{height=8cm, width=8cm, angle=0.0,figure=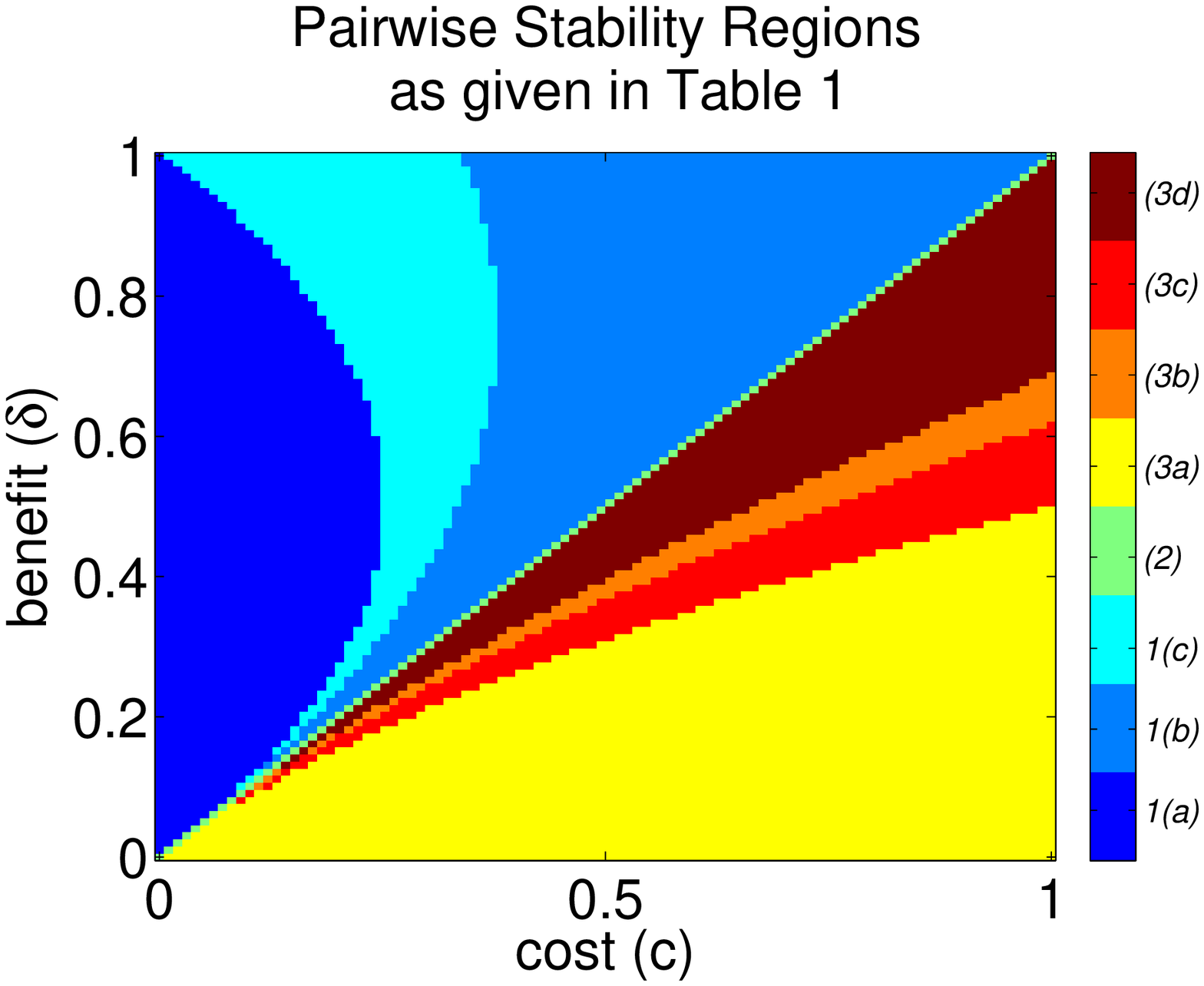, keepaspectratio }
\end{tabular}
\captionof{figure}{Graphical Illustration }\label{summaryfigure3}
\end{minipage}
\end{tabular}
\end{figure*}

\section{Simulation: Validation and Additional Insights on Topologies}\label{sec:Simulations}

In this section, we investigate various aspects of the network formation game through extensive simulations. The main purpose of this exercise is to get a better understanding of the network formation process as theoretical analysis has limited scope in enabling the understanding of the cumulative effects of many of the parameters like the initial network density, cost-benefit values, scheduling order of the nodes, etc that influence the network formation process.

In the network formation process, starting from some initial configuration of a network, the resultant topology of pairwise stable network may not be any of the standard networks considered in the previous section. In other words, these simulation results reveal that there could exist certain other topologies that satisfy pairwise stability apart from these standard networks.

Starting with some initial network (the null network, for example), the network structure changes with time as various nodes in the network add or remove links to their neighbors, so as to maximize their own individual utility from the network. It would be interesting to determine if, in the long run, the network reaches a stable state (an equilibrium or a near-equilibrium state). If the network does reach a stable state, it would be interesting to know the structure (i.e. shape) of the stable network and if this stable network is unique.
One way of approaching this is to start with the initial network and model the dynamics of the system as a function of time (or an analogous parameter) and analytically study the asymptotic network structure in the limit as time tends to infinity. However, the dynamics of the system can become very complex even in a moderately sized network, making such an approach infeasible. Further, such results would only be valid for those particular initial networks.

Another approach is to analyze the stability of some of the standard networks (complete network, cycle network, star network etc.) under our utility model (as presented in Table~\ref{summarytable2}). It would then mean that if the network reaches any of these standard stable networks, it is guaranteed to not deviate from this network. However, one problem with this approach is that starting from some initial network, we may not reach any of these standard networks. That is, some non-standard networks could be stable and the dynamic network could emerge into one of these non-standard networks.

\subsection{Simulation Setup}
We built a custom simulator using the C++ programming language in order to model the network formation process under our proposed network model. To implement the standard graph routines, we used the BOOST C++ libraries~\cite{boost}  which has efficient implementations of fundamental graph data structures and routines. We start with a random initial network consisting of $n$ nodes. The number of edges between these nodes is determined by the parameter $density (\gamma)$. For example, if $\gamma = 0$, we  start with an empty network; if $\gamma = 0.35$, we start with a network that contains $35\%$ of the possible $\binom{n}{2}$ edges. These edges are chosen uniformly at random.
As noted in Section~\ref{utilitymodel}, a node obtains a benefit of $\delta$ $(0\leq\delta\leq1)$ and incurs a cost ($c$ $(0\leq c\leq1)$) for maintaining a direct relationship (represented by an edge) with another node. In addition, each node reaps additional indirect benefit because of its potential to bridge its unconnected neighbors (determined by sparsity of relationships among his neighbors).

\subsection{The Simulation Process}
We run the simulations for each combination of possible values of $\delta$ and $c$ as shown in Table~\ref{tab:Simulation-parameters} given below. A single simulation run refers to a simulation with a particular value of of $\delta$ and $c$. Further, each simulation run is repeated multiple times as per the \textit{Num-Repetitions} parameter. We now describe the details of a single simulation run below. 
 
In a particular simulation run, each node is given an opportunity to act, based on a random schedule. Each node, when scheduled, considers three actions - namely, add an edge to a node that it is not directly connected to, delete an existing edge to a node, or do nothing. Each node chooses the action that maximizes its individual payoff (which is based on the parameters $\delta$ and $c$), breaking ties randomly.
Node~$i$, when adding an edge to node~$j$, may be allowed to do so only if it is beneficial to both or if node~$j$ is at least not worse off (mutual add (MA)). Similarly, node~$i$, when deleting an existing edge to node~$j$, may be allowed to do so unilaterally (unilateral delete). We study pairwise stable network evolution under these conditions.

Table \ref{tab:Simulation-parameters} lists the various simulation parameters. At some stage in the simulation, the network could evolve into a stable state where no node has any incentive to modify the network. One iteration in
which no node modifies the network is an \textit{idle iteration}, and the parameter \textit{Num-Idle-Terminate} indicates the number of idle iterations before we conclude that the network has reached a stable state. This is the case of normal termination of a simulation run. However, there may be cases where the network does not emerge into a stable state and cycles through  previously visited states even after many iterations (the case of \textit{dynamic-equilibrium} as noted in Hummon~\cite{hummon00}). The parameter \textit{Max-Iterations} indicates the number of iterations before we forcibly terminate the simulation run. However, we have observed that all the simulation runs achieved convergence much before the maximum iterations allowed indicating that the formation of dynamic equilibrium is not possible in our utility model. However, we leave the formal proof of this observation as a future work. The parameter \textit{Num-Repetitions} indicates the number of times each simulation run was repeated. The simulations were averaged out over different initial conditions and random schedules.

\begin{table}[h]
 \begin{tabular}{cc}
\begin{minipage}{8cm}
\footnotesize
\begin{center}
\begin{tabular}{|l|l|}
\hline
Parameters & Values\tabularnewline
\hline
\hline
N & 3, 4, 5, 10, 20\tabularnewline
\hline
Cost (c) & 0.05 to 1, in steps of 0.05\tabularnewline
\hline
Benefit ($\delta$) & 0.05 to 1, in steps of 0.05\tabularnewline
\hline
Density ($\gamma$)& 0, 0.35, 0.7\tabularnewline
\hline
Experiment & Mutual-Add, Unilateral-Delete\tabularnewline
\hline
Num-Iterations & 1000\tabularnewline
\hline
Num-Repetitions & 100\tabularnewline
\hline
Num-Idle-Terminate & 30\tabularnewline
\hline
\end{tabular}

\captionof{table}{Simulation parameters and Values\label{tab:Simulation-parameters}}
\end{center}
\end{minipage}
&
\begin{minipage}{8cm}
\centering
\includegraphics[scale=0.4]{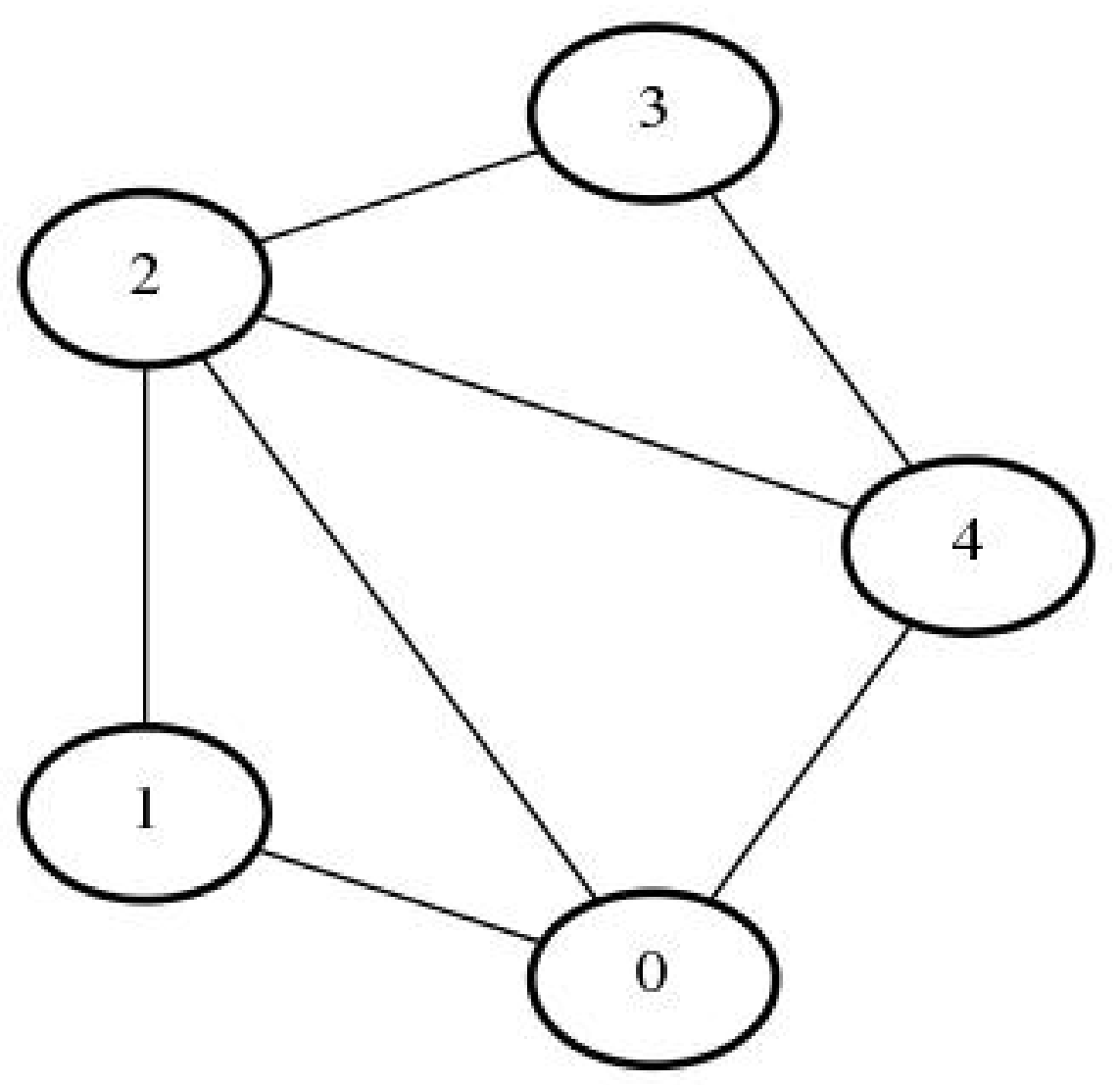}
\captionof{figure}{A stylized 5-node network\label{fig:5-actor-network}}
\end{minipage}
\end{tabular}
\end{table}
\vspace{-0.2in}
\subsection{Metrics Recorded\label{sub:Metrics-Recorded}}

At the end of \textit{Num-Repetitions} number of repetitions, a number of metrics were recorded. The following lists some of the important metrics recorded.
\begin{enumerate}
\item The network structure (shape) for each repetition
\item The frequency with which each of the network structures in Section
\ref{sub:Classification-of-Network-Structures} resulted (across all
repetitions)
\item The mean utility of the final network (across all repetitions)
\item The mean time to reach the final network (across all repetitions)
\item The mean number of acts to reach the final network (across all repetitions)
\end{enumerate}

Before we present the results, we briefly describe the classification criteria used to identify pairwise stable networks.

\subsection{Classification of Pairwise Stable Network Structures\label{sub:Classification-of-Network-Structures}}

Once the network reaches a stable state, we classify the network structure as one of the network structures shown in Table~\ref{networkstructures}.  As in Hummon~\cite{hummon00}, we use the sorted (descending order) degree vector to characterize the structure of the stable network. For example, the Null network has a sorted degree vector of (0, 0, \ldots{}, 0), the Star network (n-1, 1, 1, \ldots{}, 1) and the Complete network (n-1, n-1, \ldots{}, n-1). We refer to a network structure a shared network if it is a regular network (i.e., all nodes have same degree) of some uniform degree. For example, a cycle is a $2$-regular graph and hence is a shared network.

Also as in Hummon~\cite{hummon00}, we use total mean squared deviation (MSD) to classify the resultant stable network as Near-``standard network'' (for example, Near-complete network). Further, if the mean squared deviation is above a certain threshold ($\tau$) then we know its not close to any of the above topologies, we then color the graph using a greedy coloring algorithm ~\cite{boost} and then classify it either as a general k-partite graph (where $k$ equals the number of colors required to color the graph) or any of the other network structures shown in Table~\ref{networkstructures}. In our simulations, we use the maximum deviation ($(n-1)^2$) for calculating the $\tau$, i.e., $\tau=0.1\times(n-1)^2$. 

Note that whenever we classify a network as any type of K-Partite network, we implicitly mean that $K\geq 3$. The case of $K=2$ is the same as bipartite network and is handled as a separately as shown in Table~\ref{networkstructures}. Turan network refers to a complete bipartite network with the sizes of the two partitions to be as equal as possible. If $N$ is even, then the Turan network has equal sized partitions whereas if $N$ is odd, the size of one partition is one less than the other partition. 

For classification of a sorted degree network as a near-shared network, we first need to calculate the order of the regular network with which this degree vector needs to be compared. As in Hummon~\cite{hummon00}, to compute the total mean squared deviation for the shared structure, the ideal order is defined by average number of ties in the in-out degree vector, rounded to the nearest whole tie. In this example, if the degree vector is (3,2,1,1,1), the average is 1.6, and the ideal type shared structure is (2,2,2,2,2).
However, note that a cycle network is necessarily a shared network but a shared network need not always be a cycle network.

\begin{table}[h]
\centering
\scriptsize
\begin{tabular}{|c|c|c|c|}
\hline 
NULL& STAR& SHARED& COMPLETE\\
\hline
NEAR-NULL& NEAR-STAR& NEAR-SHARED& NEAR-COMPLETE\\
\hline 
BI-PARTITITE-COMPLETE & TURAN & EQUI-K-PARTITE-COMPLETE&EQUI-K-PARTITE \\
\hline 
K-PARTITE-COMPLETE & K-PARTITE & &\\
\hline 
\end{tabular}
\caption{Possible Network Structures considered in the simulations}\label{networkstructures}
\end{table}

The following example clarifies this procedure: Consider the $5$-node network as shown in Figure~\ref{fig:5-actor-network}. Suppose that we would like to classify this network as one of the following standard networks : Null, Star, Shared, Complete, Near-Null, Near-Star, Near-Shared or  Near-Complete. This is done as follows. Note that the given network does not classify as any of the first four networks in the list given above. Hence, we try to classify the given network as one of the remaining four networks (i.e., the `near' type networks).

We know that the sorted degree vector is $(4,3,3,2,2)$ for the given network. The ideal order for the shared network comparison is calculated by taking the average degree (which is $2.8$) and rounding to the nearest integer (which gives $3$). This means we have to compare the network to a $3$-regular network. The total MSD from the shared network is thus $((4-3)^{2}+(3-3)^{2}+(3-3)^{2}+(2-3)^{2}+(2-3)^{2}))/5 = 0.6$. The total MSD of this network from Star network is $((4-4)^{2}+(3-1)^{2}+(3-1)^{2}+(2-1)^{2}+(2-1)^{2}))/5 = 2$. Similarly, the total MSD from Null network is $8.4$, and the total MSD from the Complete Network is $2$. The value $0.6$ being the least among these and less than $10\%$ of maximum deviation $16$, we classify the above network structure as Near-Shared.

\subsection{Multiple Classification of Pairwise Stable Structures}
We note that the classification of pairwise stable network structures according to Table~\ref{networkstructures} is not mutually exclusive. There can exist networks which can be classified as more than one of the types described in Table \ref{networkstructures}. We illustrate a couple of interesting network structures that we encountered during our simulations here. Figure~\ref{fig:multipleclassification}(a) refers to a pairwise stable network that emerged when we ran the simulation with \textit{random\_seed }$ = 6875, \delta=0.7, c=0.55$. We observed that this network is both a Near-Shared network as well as a Tri-partite complete network whose parititions are ${(0,6,7,8), (1,2,5), (3,4,9)}$. In such cases, we classify the network structure as a K-Partite Complete network.

\begin{figure*}[h]
\begin{tabular}{cc}
\begin{minipage}{8 cm}
\vspace{0.2in}
\centering
\includegraphics[scale=0.3]{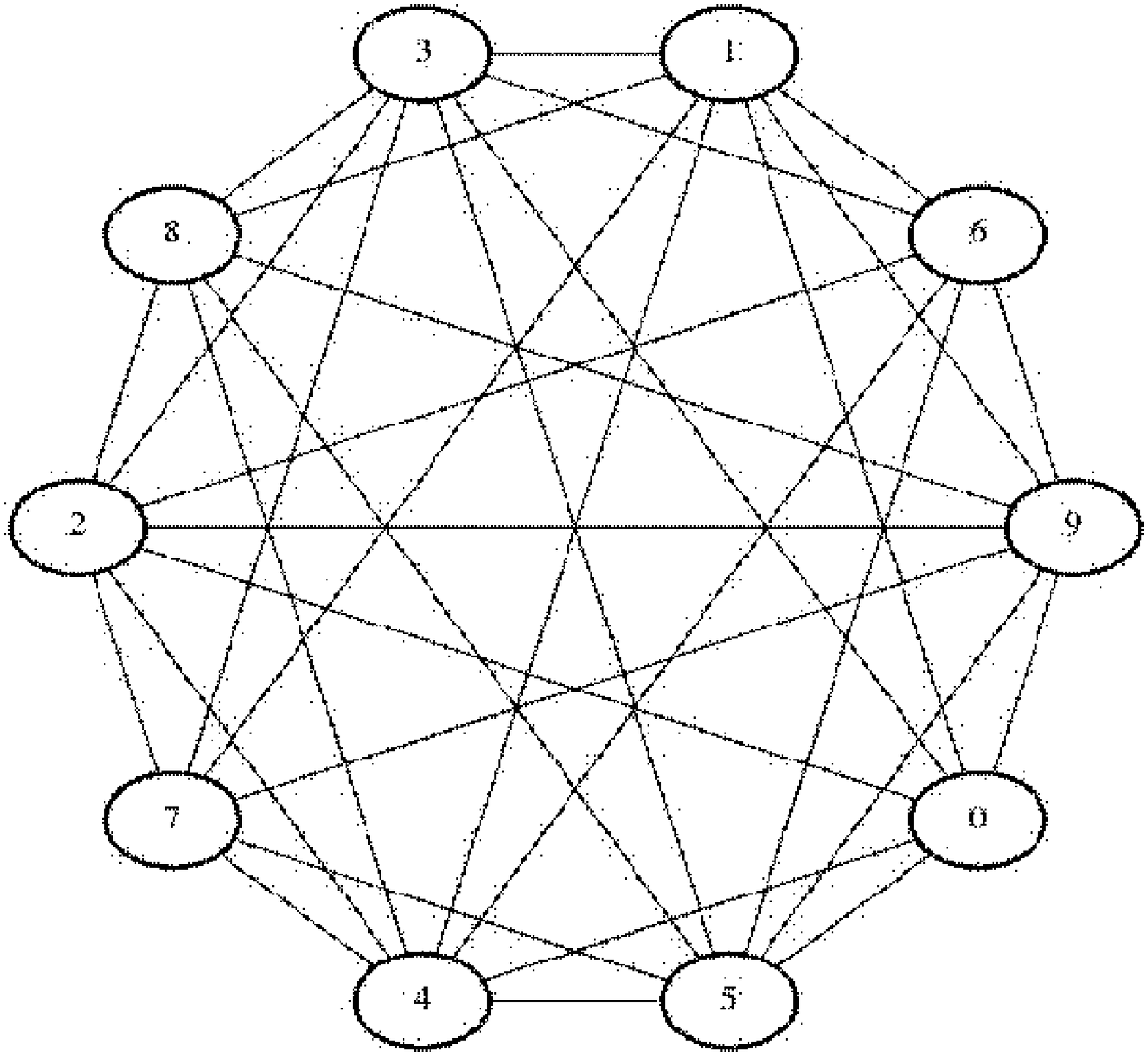}
\end{minipage}
&
\begin{minipage}{8 cm}
\vspace{0.2in}
\centering
\includegraphics[scale=0.3]{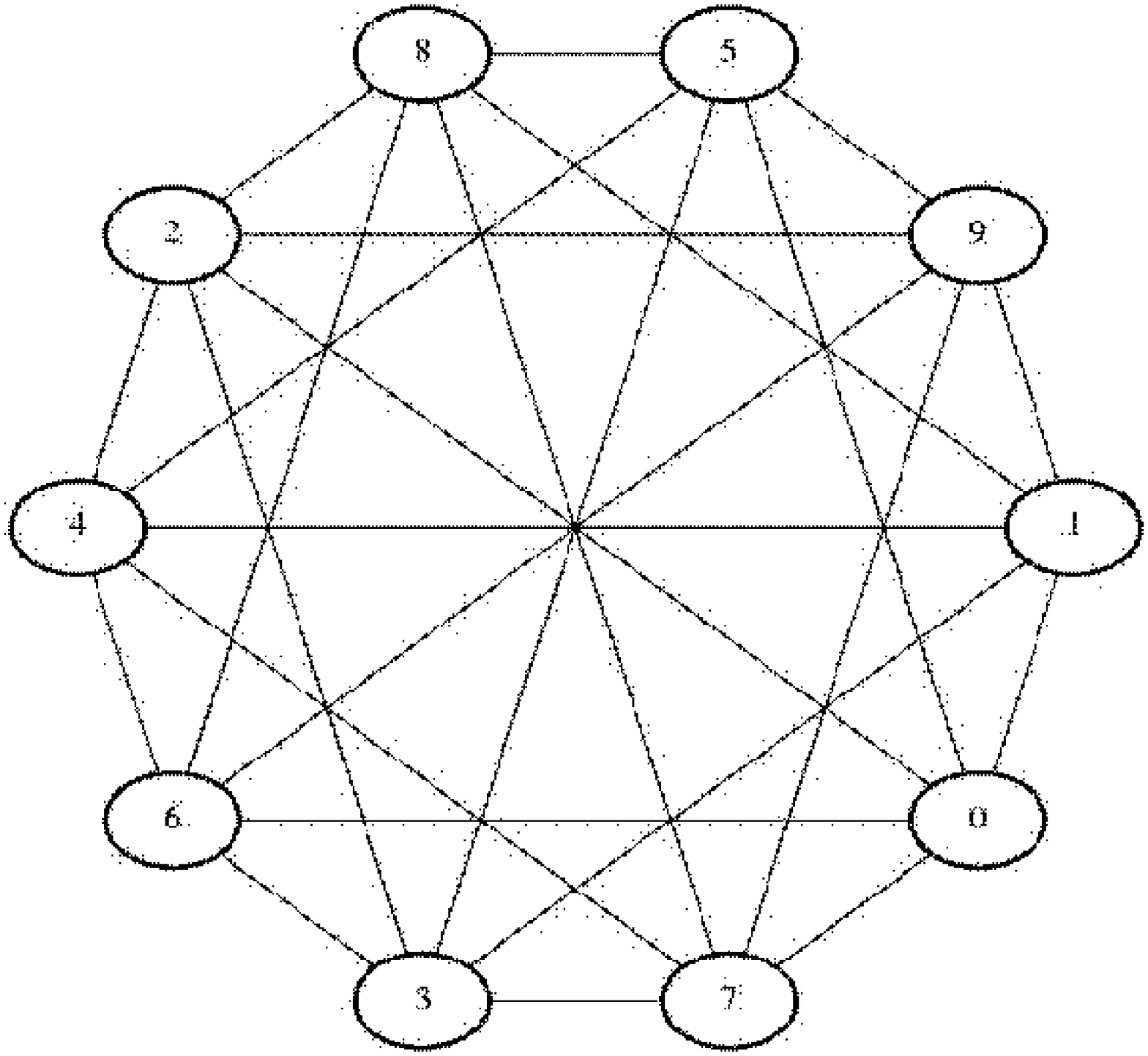}
\end{minipage}
\\
(a) & (b)
\end{tabular}
\caption{Possibility of multiple classifications for a given network structure\label{fig:multipleclassification}}
\end{figure*}
Another example is shown in Figure~\ref{fig:multipleclassification}(b) which is obtained when running simulations with \textit{random\_seed }$=15256, \delta=0.5, c=0.5$. We observe that this graph can be classified as a regular (or Shared) network with degree=$5$. However, it turns out that this graph is also an equi-partitioned bipartite network with partitions $(0, 3, 4, 8, 9), (1, 2, 5, 6, 7)$. In such cases, we classify the graph as equi-bipartite network (or the Turan network).

\subsection{Interpretation of Pairwise Stability}
In a pairwise stable network, if a node adds a link to another node and gains strictly from it, the other node should lose strictly. Hence, the  addition of the link becomes infeasible in this case. However, nodes in a pairwise stable network can still add links if adding these links does not change the payoffs of either of the nodes. In this case, the nodes are indifferent about adding the link. In the case of deletion, a node will delete a link from the current network unilaterally if it strictly benefits from doing so. We use this interpretation of pairwise stability during the course of our simulations.


\subsection{Model Validation}
We now proceed to understand some of the results of our simulations. First, in this section, we focus on the validation of our theoretical results on pairwise stability as shown in Figure~\ref{fig:validation}. We are interested in knowing the following aspects in the simulations.
\begin{itemize}
 \item Do the pairwise stable networks identified in Table~\ref{summarytable2} actually emerge in the simulation process?
 \item If so, under what values of $\delta$ and $c$ do they emerge?
 \item Do the conditions match with the theoretical results?
\end{itemize}

\begin{figure*}[htb!]
\centering
\begin{tabular}{cccc}
\begin{minipage}{3.3cm}
\centering
\epsfig{height=3cm, width=3cm, angle=0.0,figure=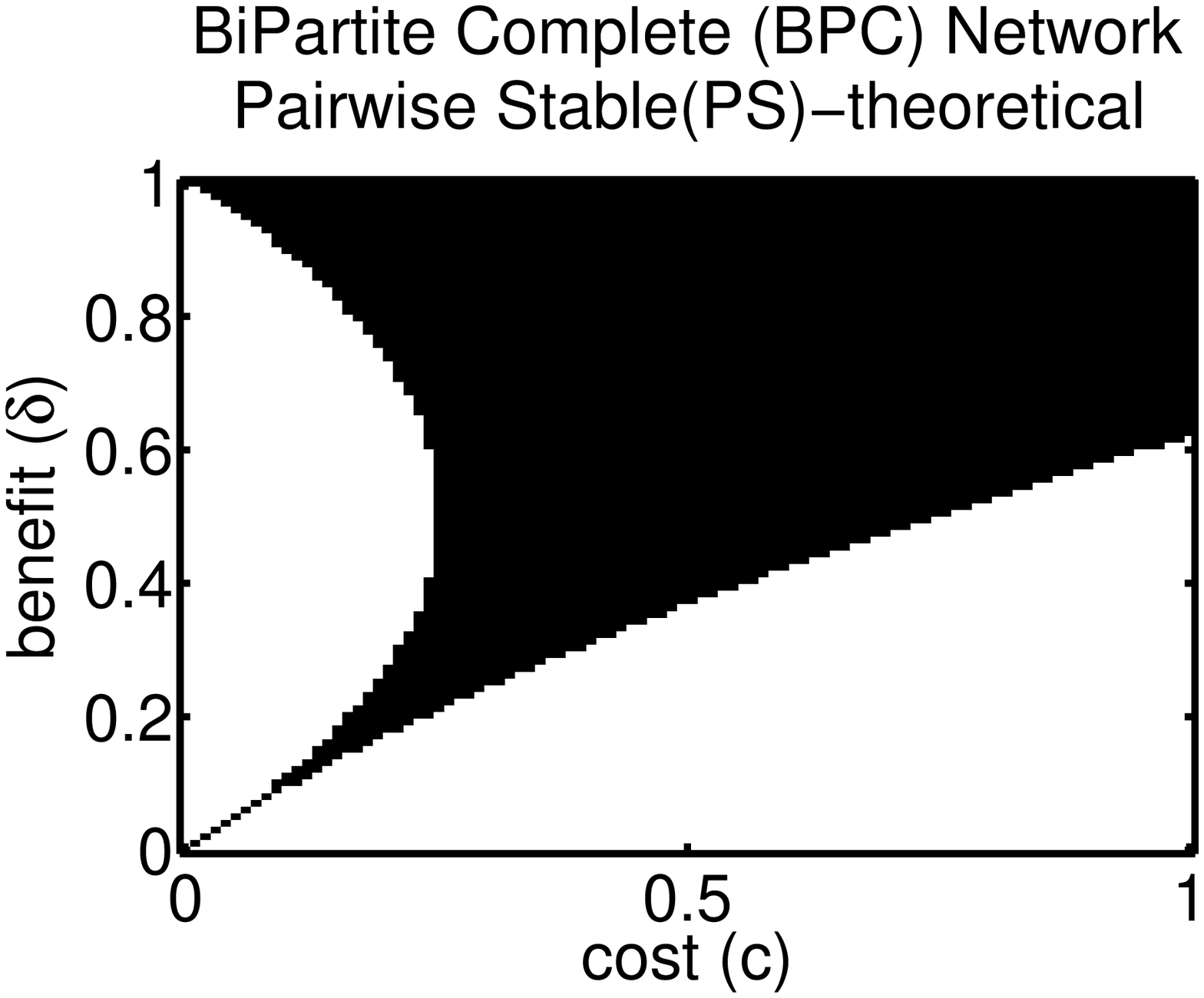, keepaspectratio}
\end{minipage}
&
\begin{minipage}{3.3cm}
\centering
\epsfig{height=3cm, width=3cm, angle=0.0,figure=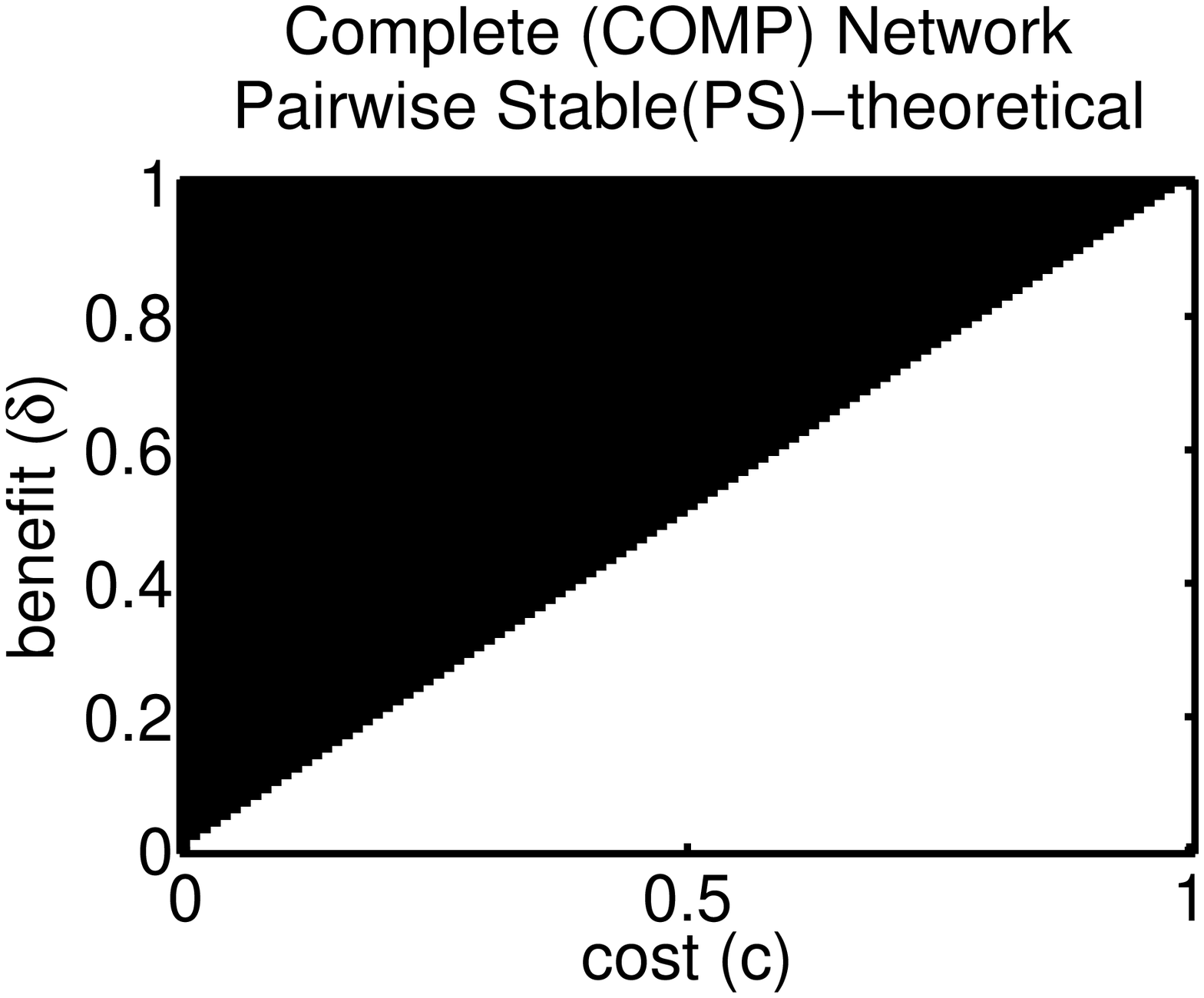}
\end{minipage}
&
\begin{minipage}{3.3cm}
\centering
\epsfig{height=3cm, width=3cm, angle=0.0,figure=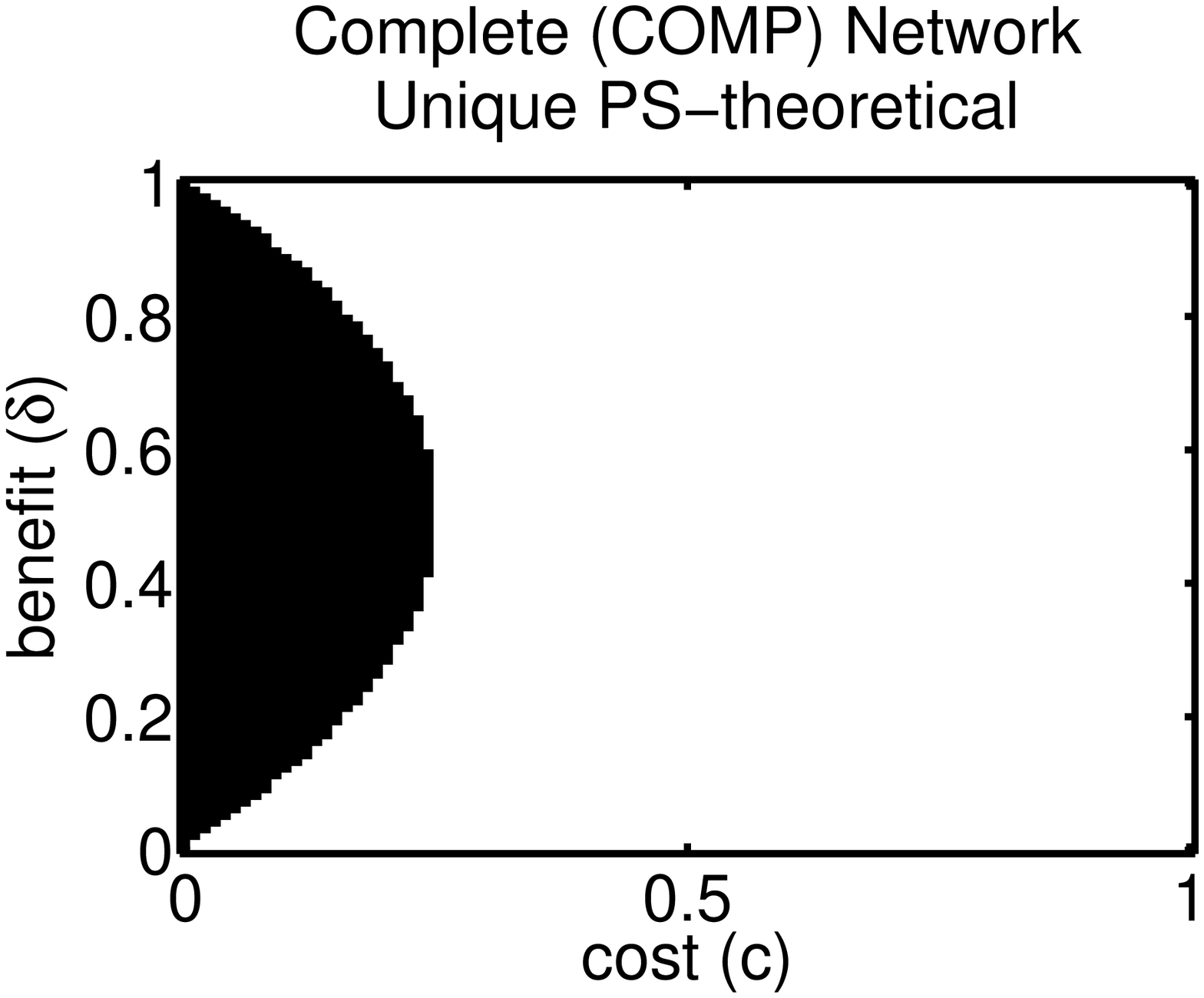}
\end{minipage}
& 
\begin{minipage}{3.3cm}
\centering
\epsfig{height=3cm, width=3.1cm, angle=0.0,figure=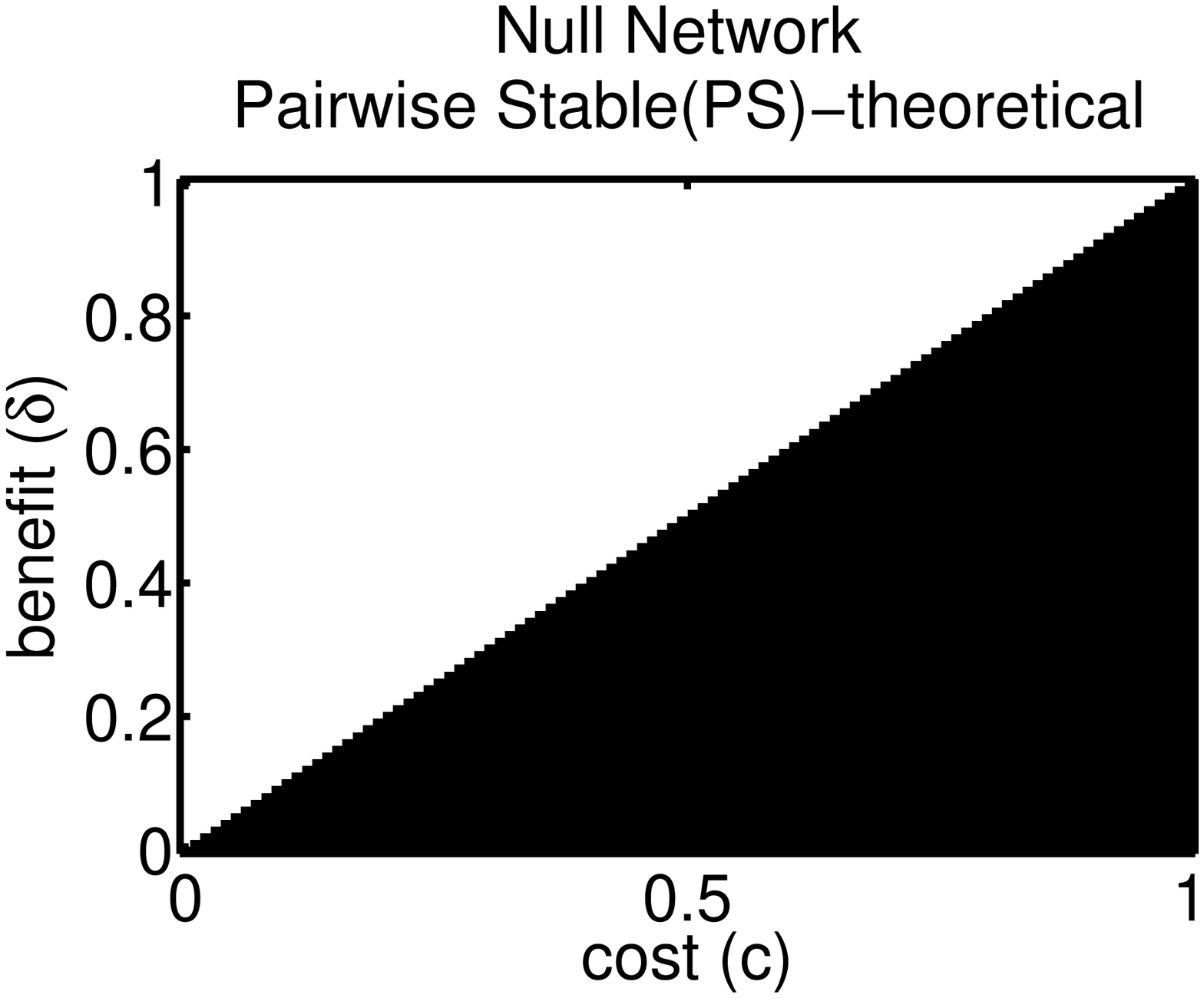}
\end{minipage}
\\
(a) & (b) & (c) & (d) \\
\begin{minipage}{3cm}
\centering
\epsfig{height=3cm, width=3cm, angle=0.0,figure=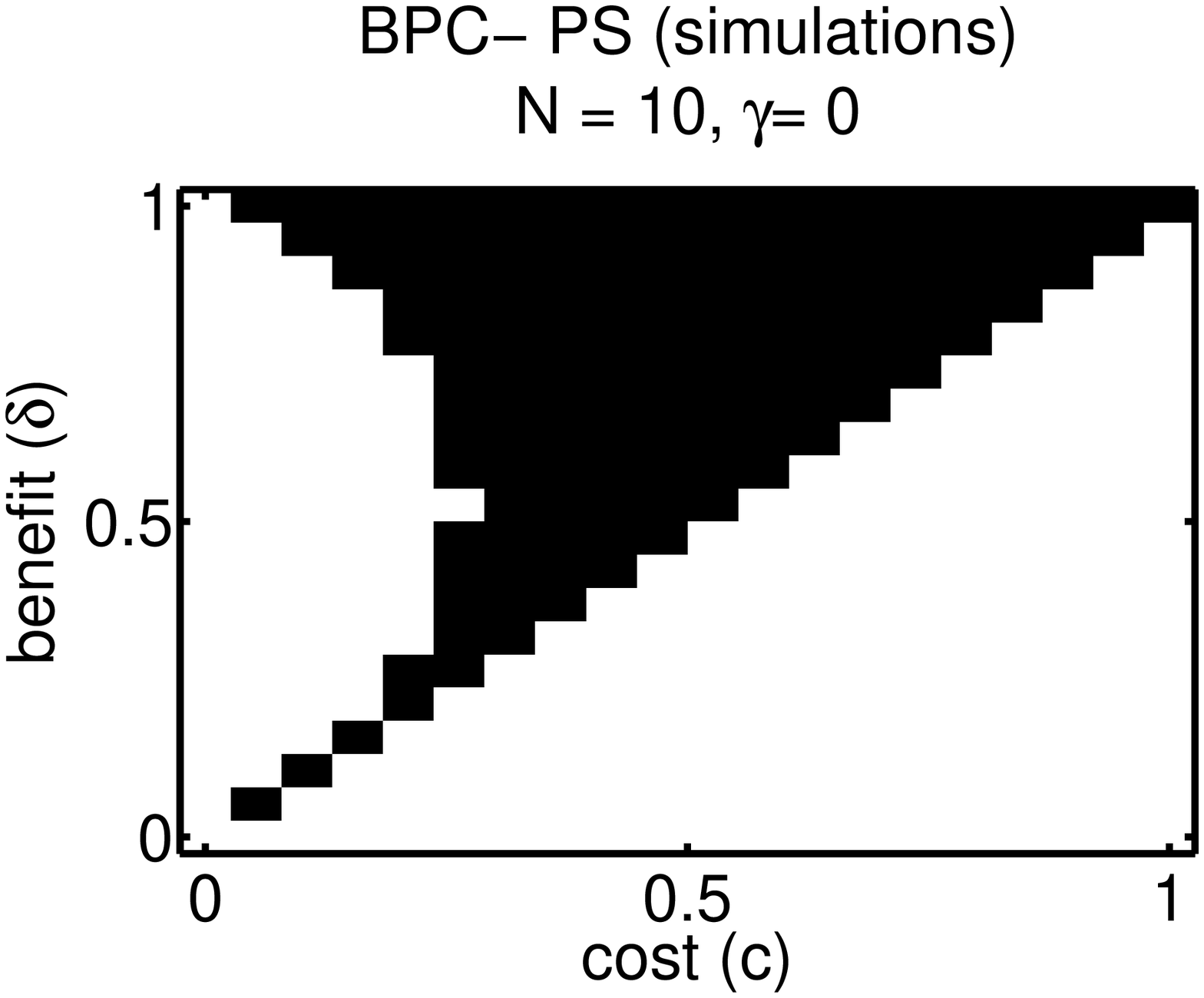}
\end{minipage}
&
\begin{minipage}{3 cm}
\centering
\epsfig{height=3cm, width=3cm, angle=0.0,figure=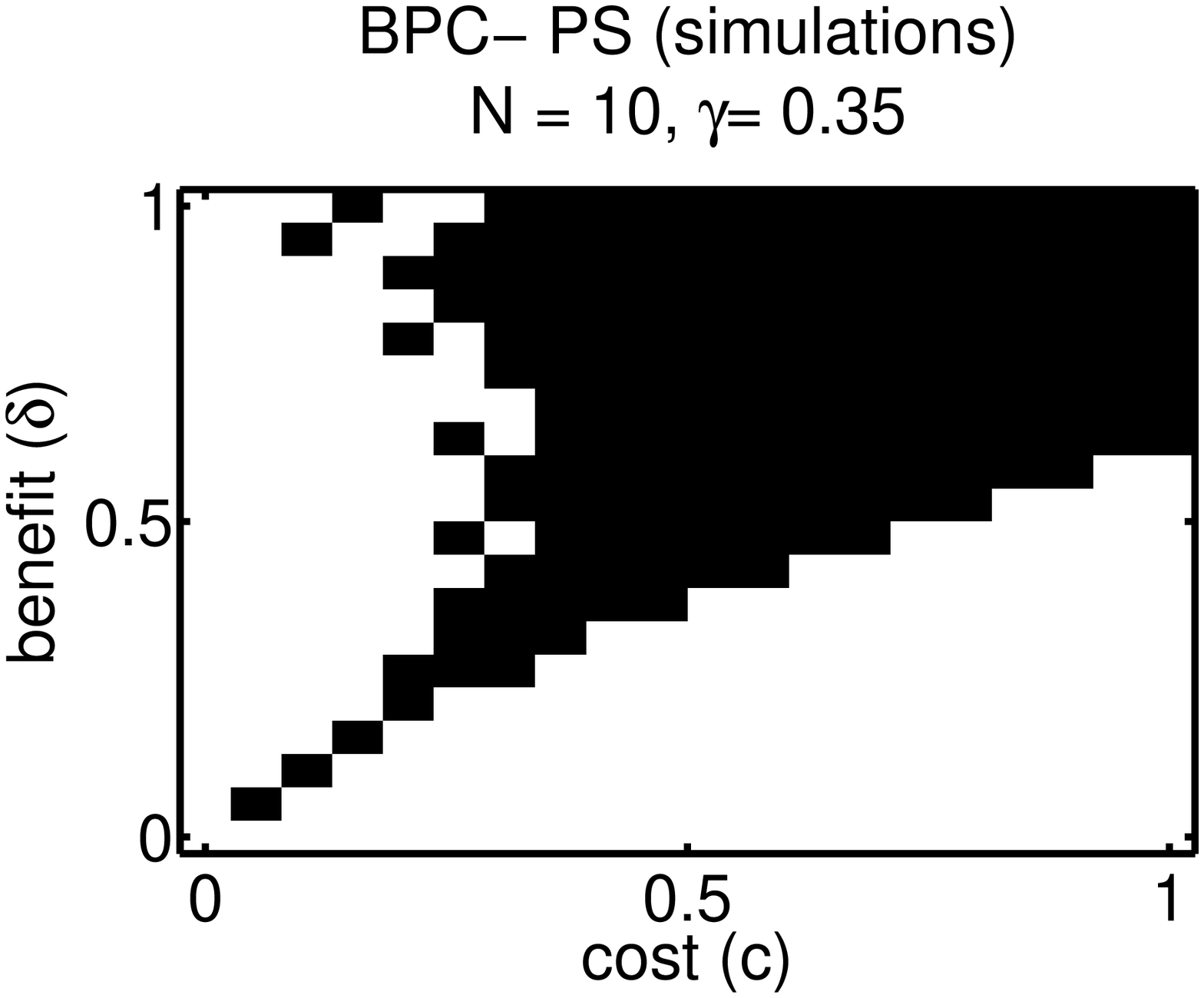}
\end{minipage}
&
\begin{minipage}{3 cm}

\centering
\epsfig{height=3cm, width=3cm, angle=0.0,figure=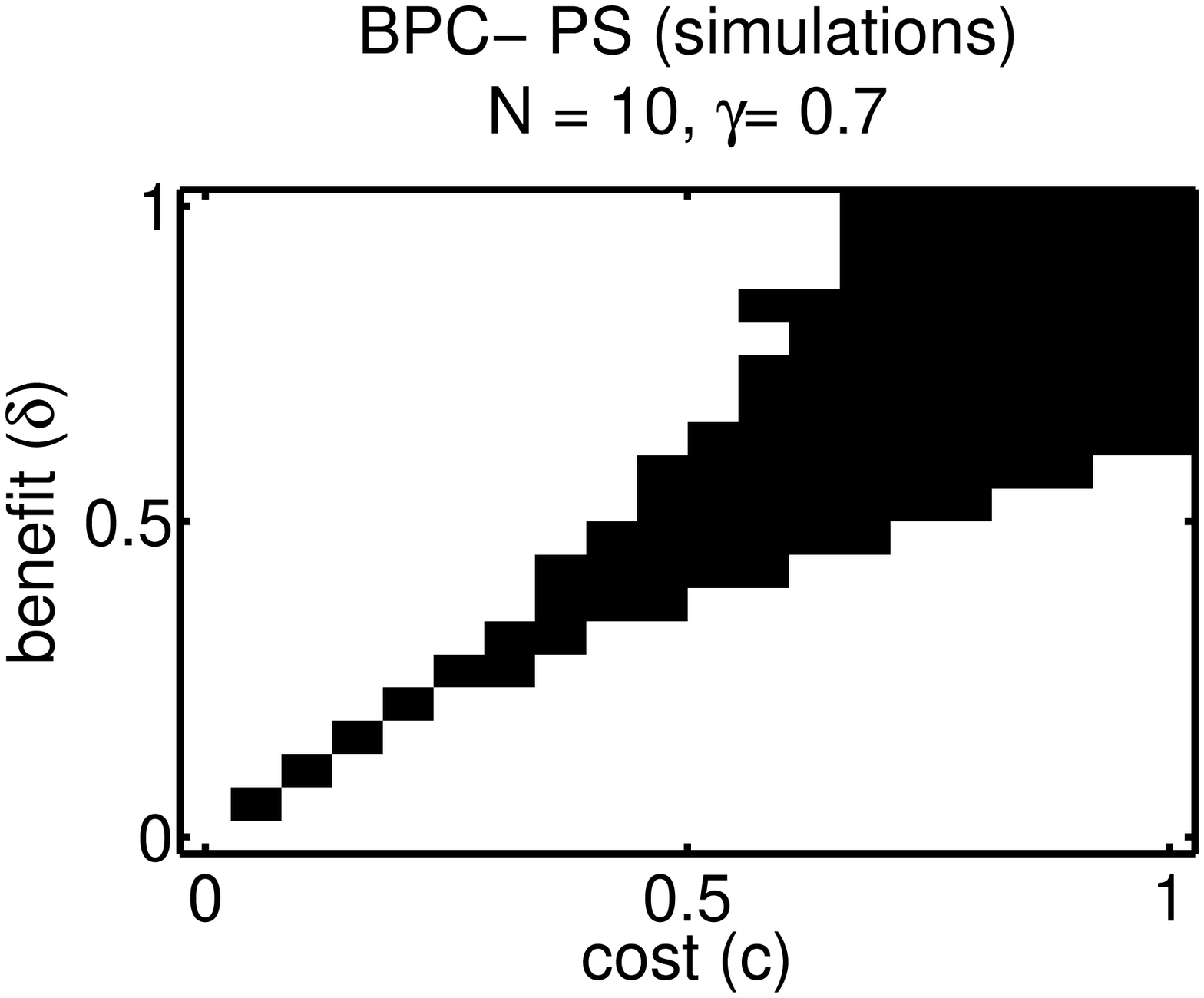}
\end{minipage}
& 
\begin{minipage}{3.5cm}

\centering
\epsfig{height=3cm, width=3.1cm, angle=0.0,figure=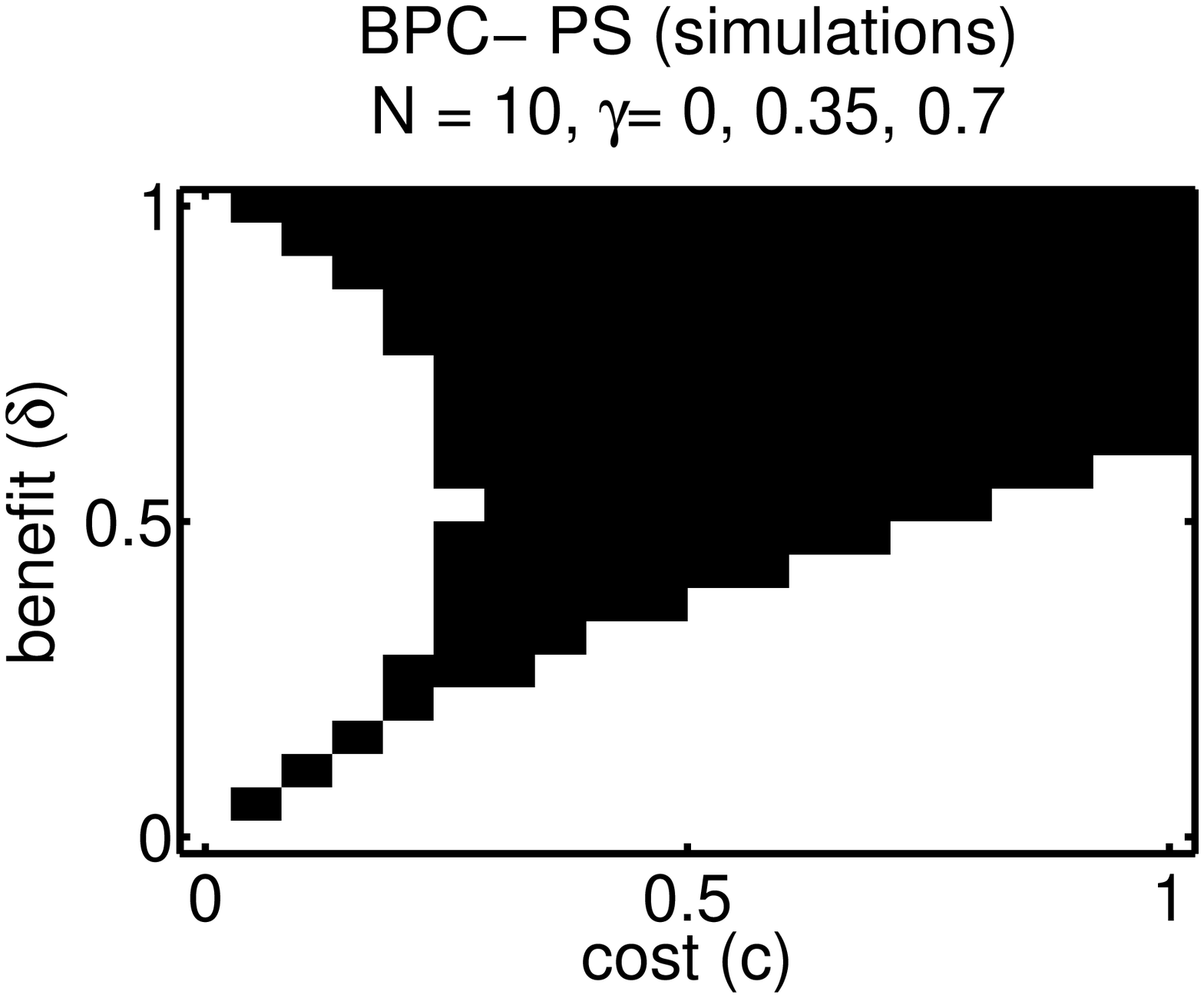}
\end{minipage}
\\
(e) & (f) & (g) & (h) \\
\begin{minipage}{3cm}

\centering
\epsfig{height=3cm, width=3cm, angle=0.0,figure=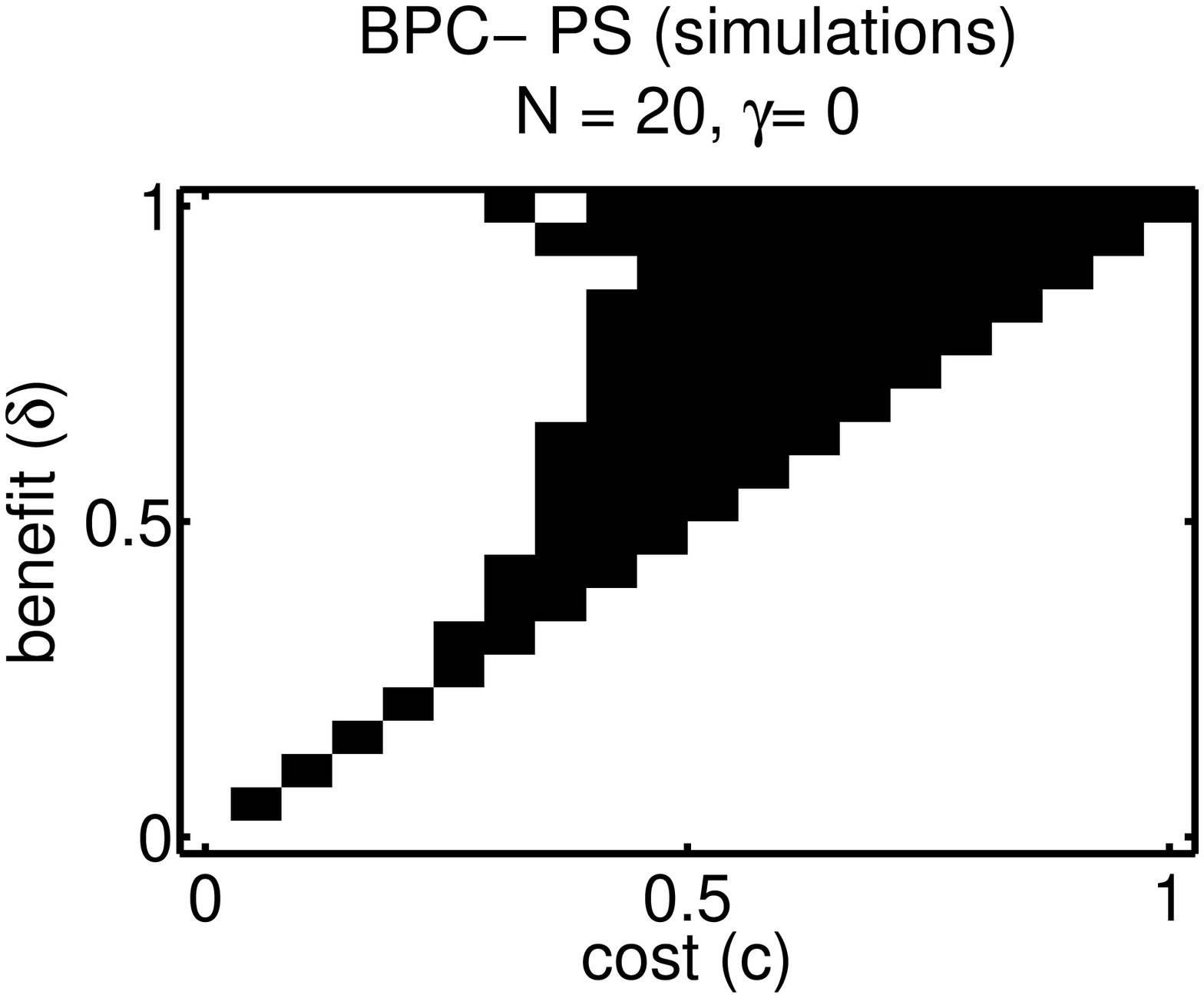}
\end{minipage}
&
\begin{minipage}{3 cm}

\centering
\epsfig{height=3cm, width=3cm, angle=0.0,figure=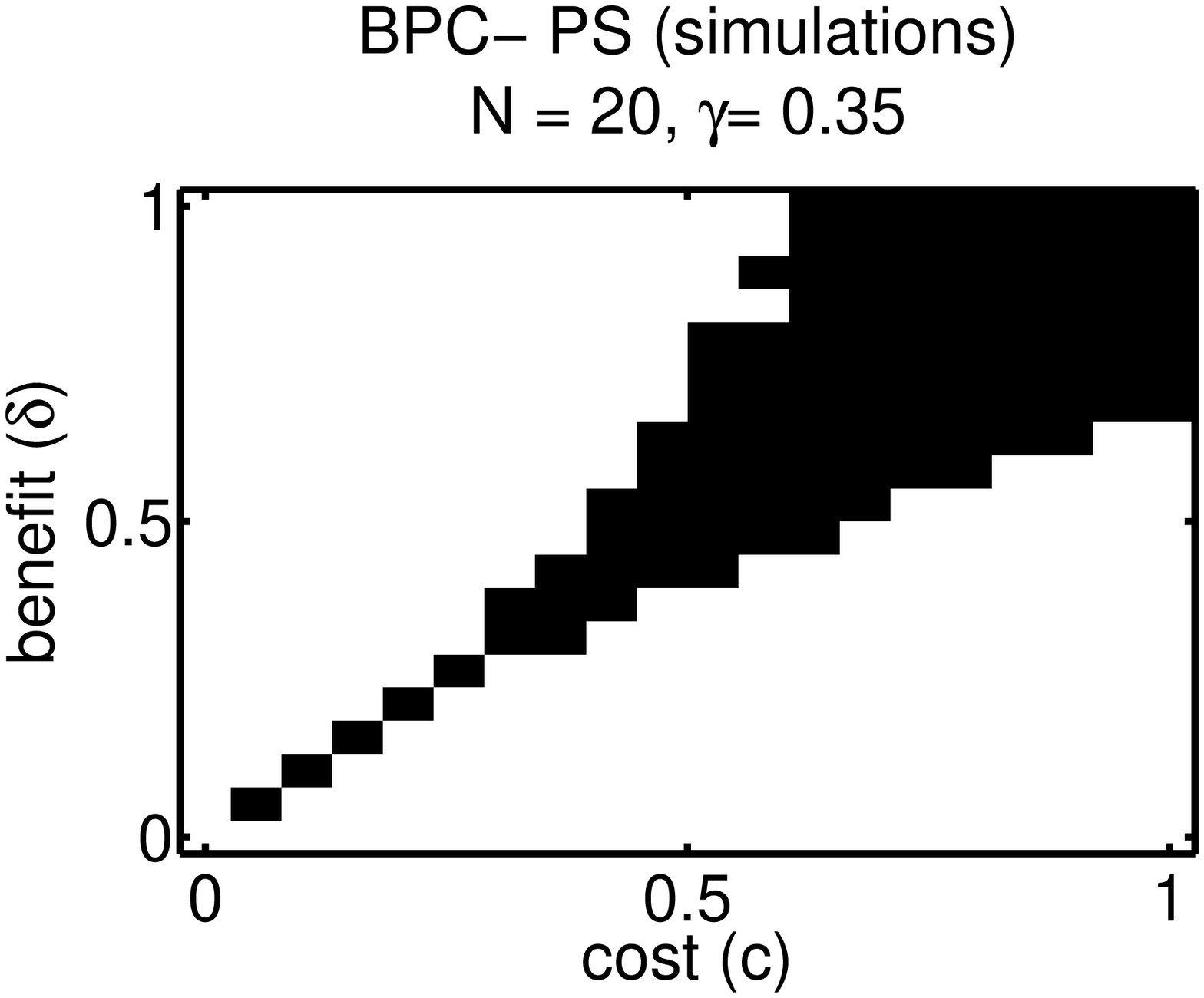}
\end{minipage}
&
\begin{minipage}{3 cm}

\centering
\epsfig{height=3cm, width=3cm, angle=0.0,figure=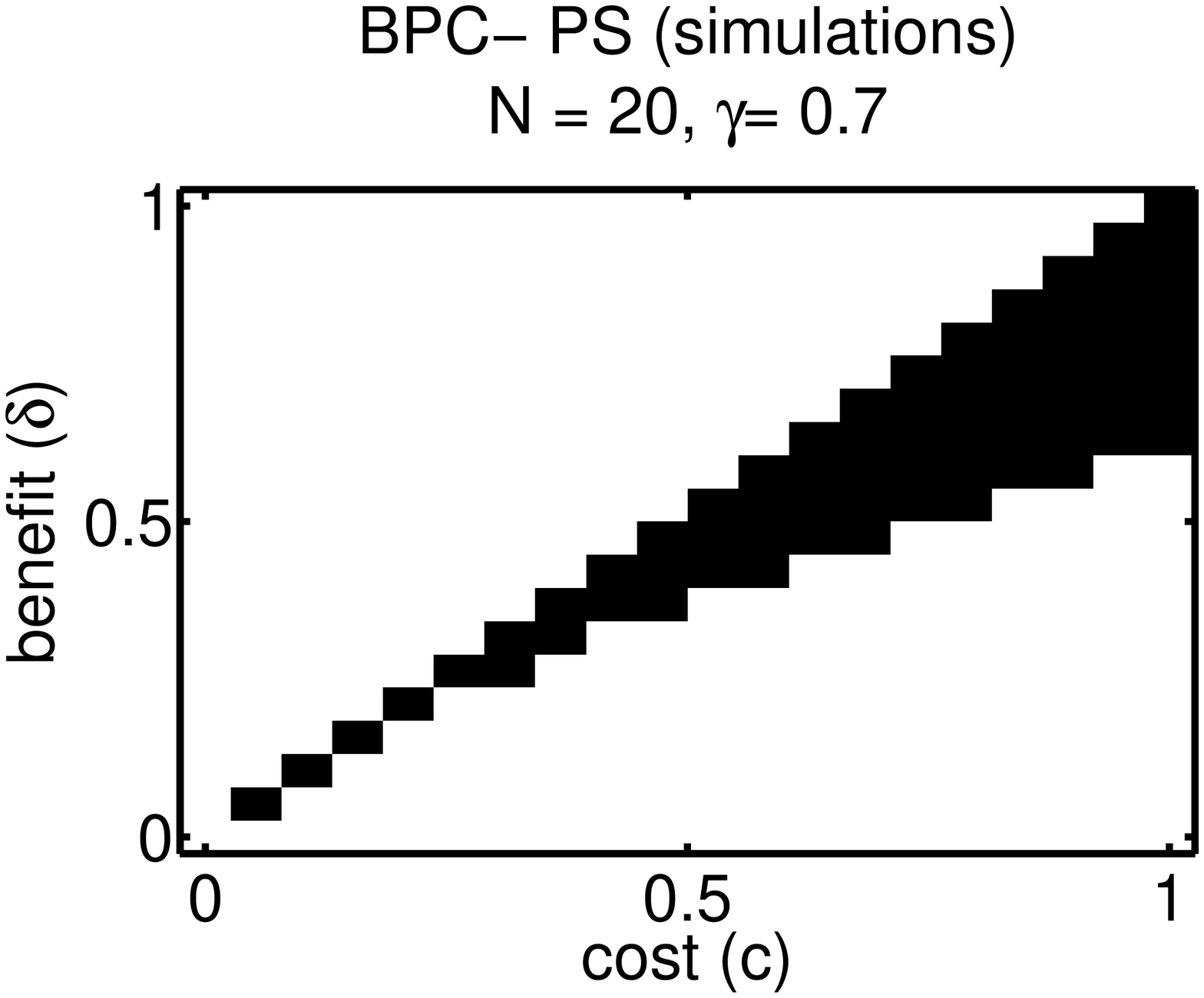}
\end{minipage}
& 
\begin{minipage}{3.5cm}

\centering
\epsfig{height=3cm, width=3.1cm, angle=0.0,figure=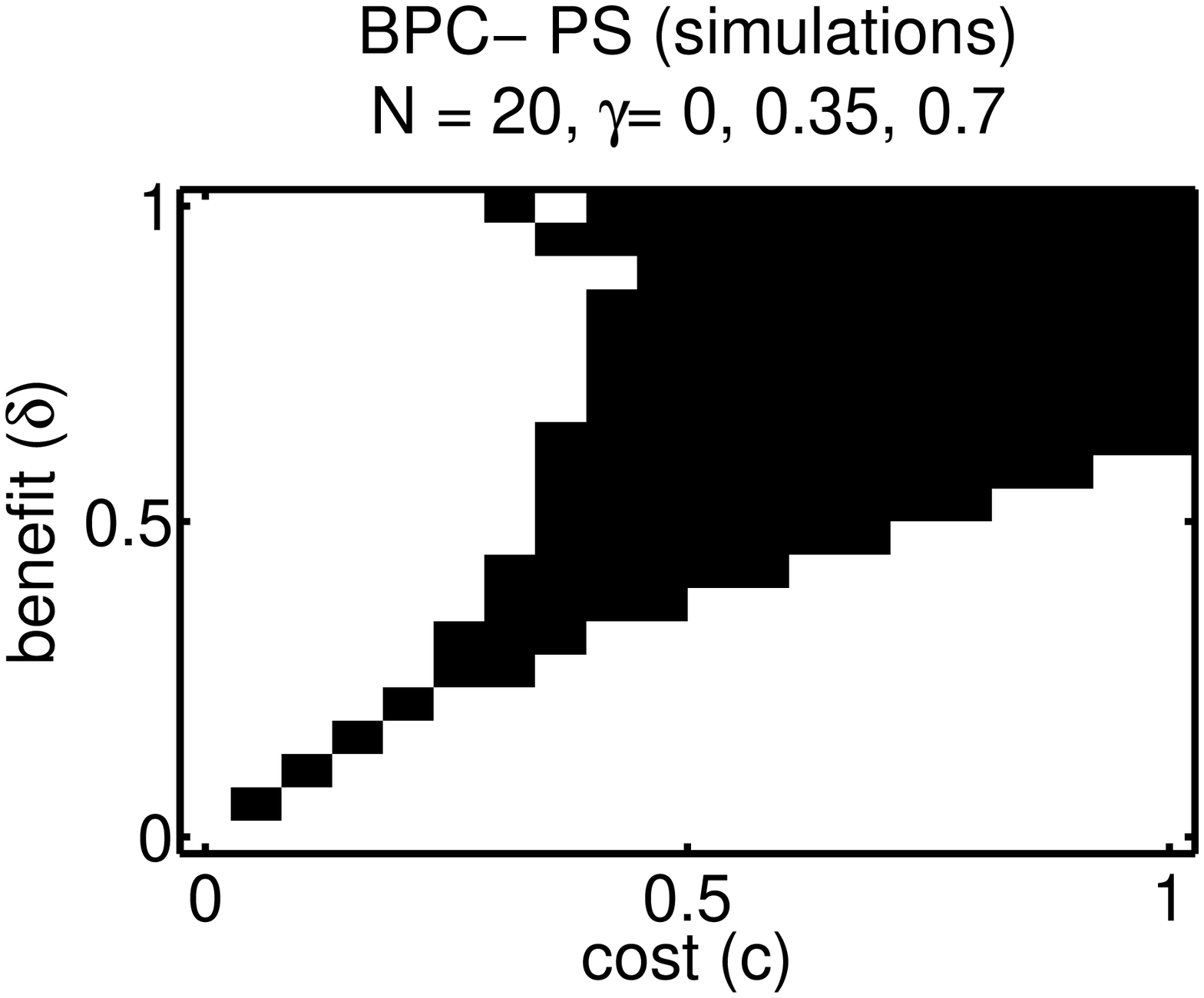}
\end{minipage}
\\
\small (i) & (j) & (k) & (l) \\ \normalsize 
\begin{minipage}{3cm}

\centering
\epsfig{height=3cm, width=3cm, angle=0.0,figure=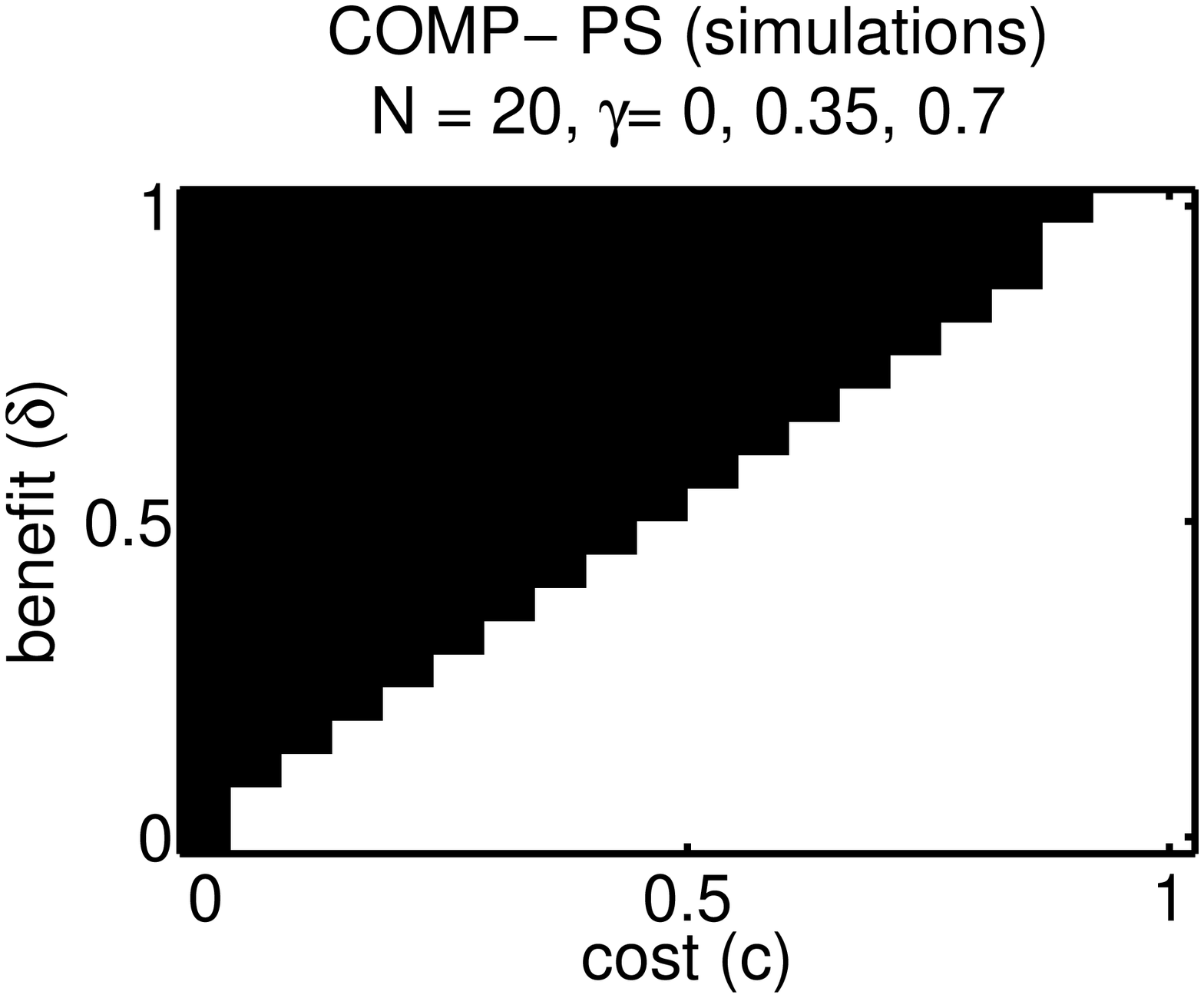}
\end{minipage}
&
\begin{minipage}{3 cm}

\centering
\epsfig{height=3cm, width=3cm, angle=0.0,figure=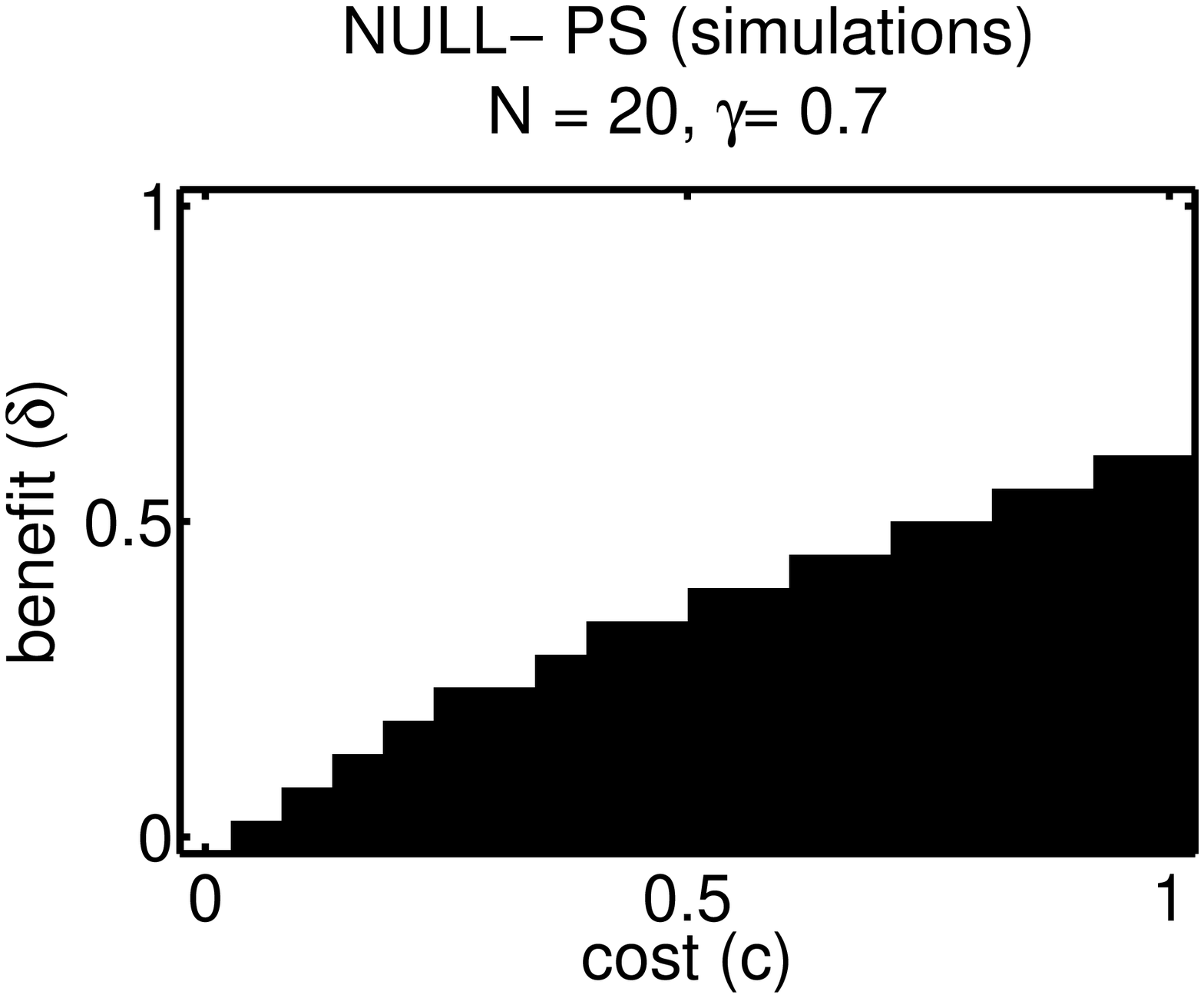}
\end{minipage}
&
\begin{minipage}{3 cm}

\centering
\epsfig{height=3cm, width=3cm, angle=0.0,figure=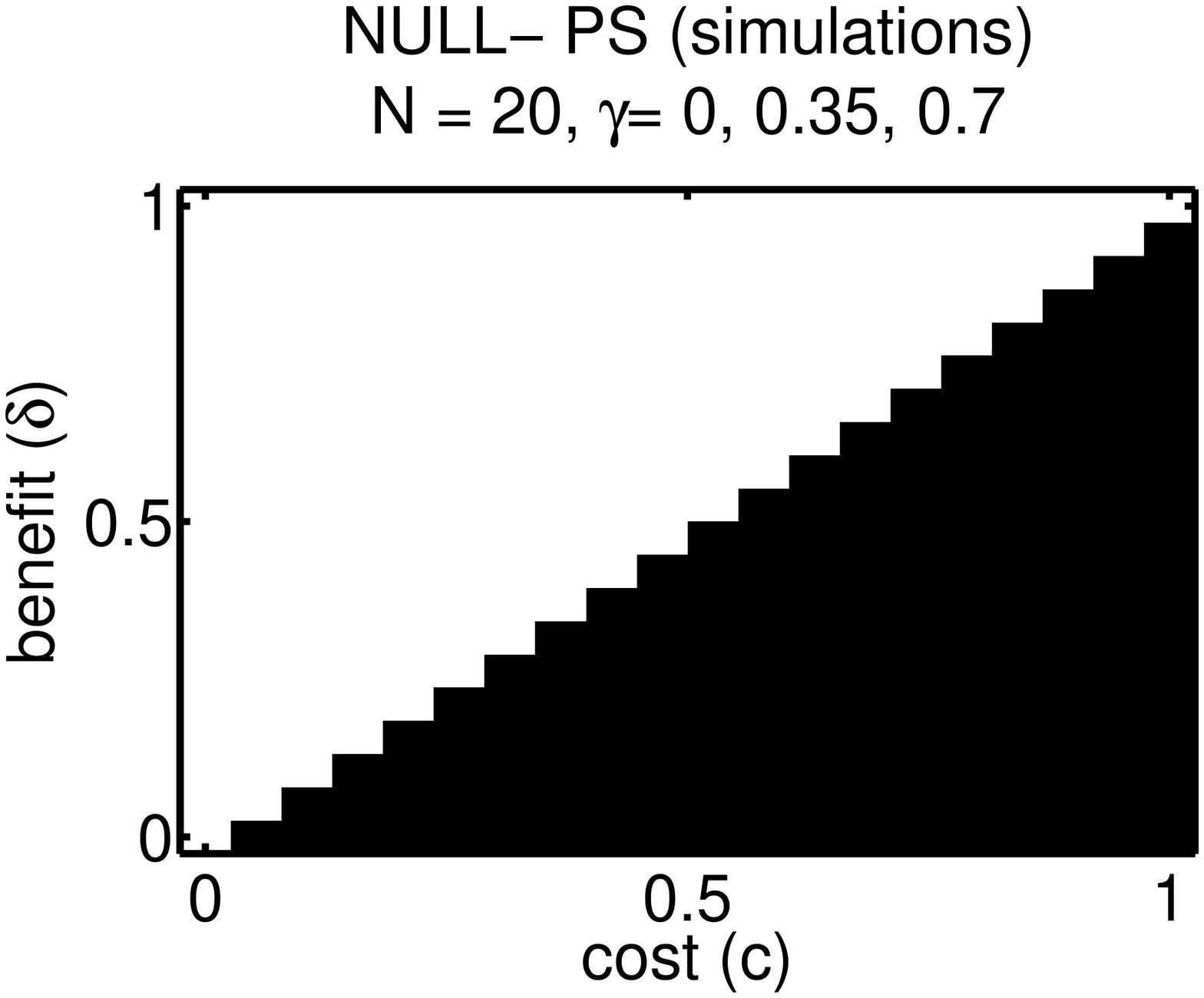}
\end{minipage}
& 
\begin{minipage}{3.5cm}

\centering
\epsfig{height=3cm, width=3.1cm, angle=0.0,figure=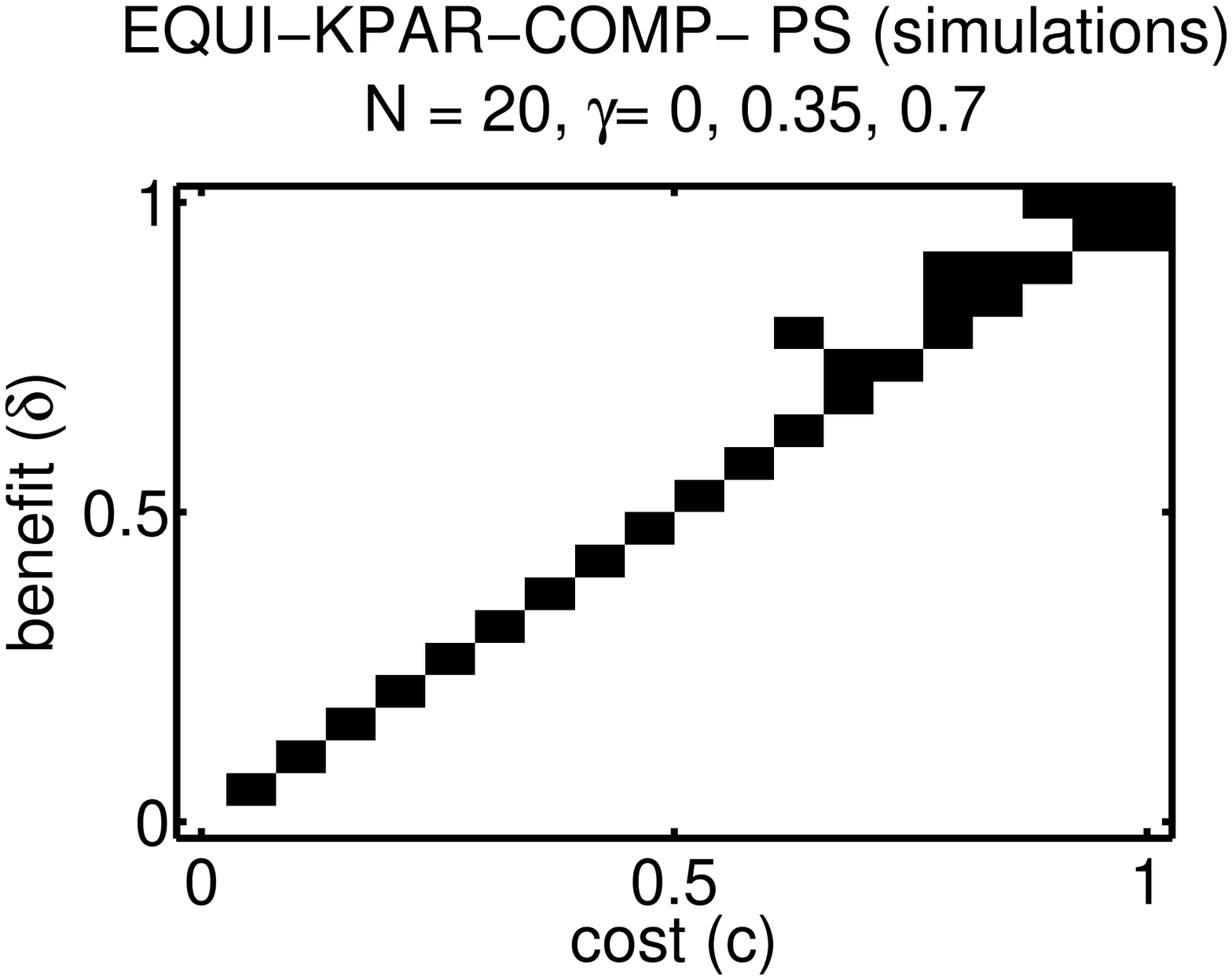}
\end{minipage}\\
(m) & (n) & (o) & (p) \\
\begin{minipage}{3cm}

\centering
\epsfig{height=3cm, width=3cm, angle=0.0,figure=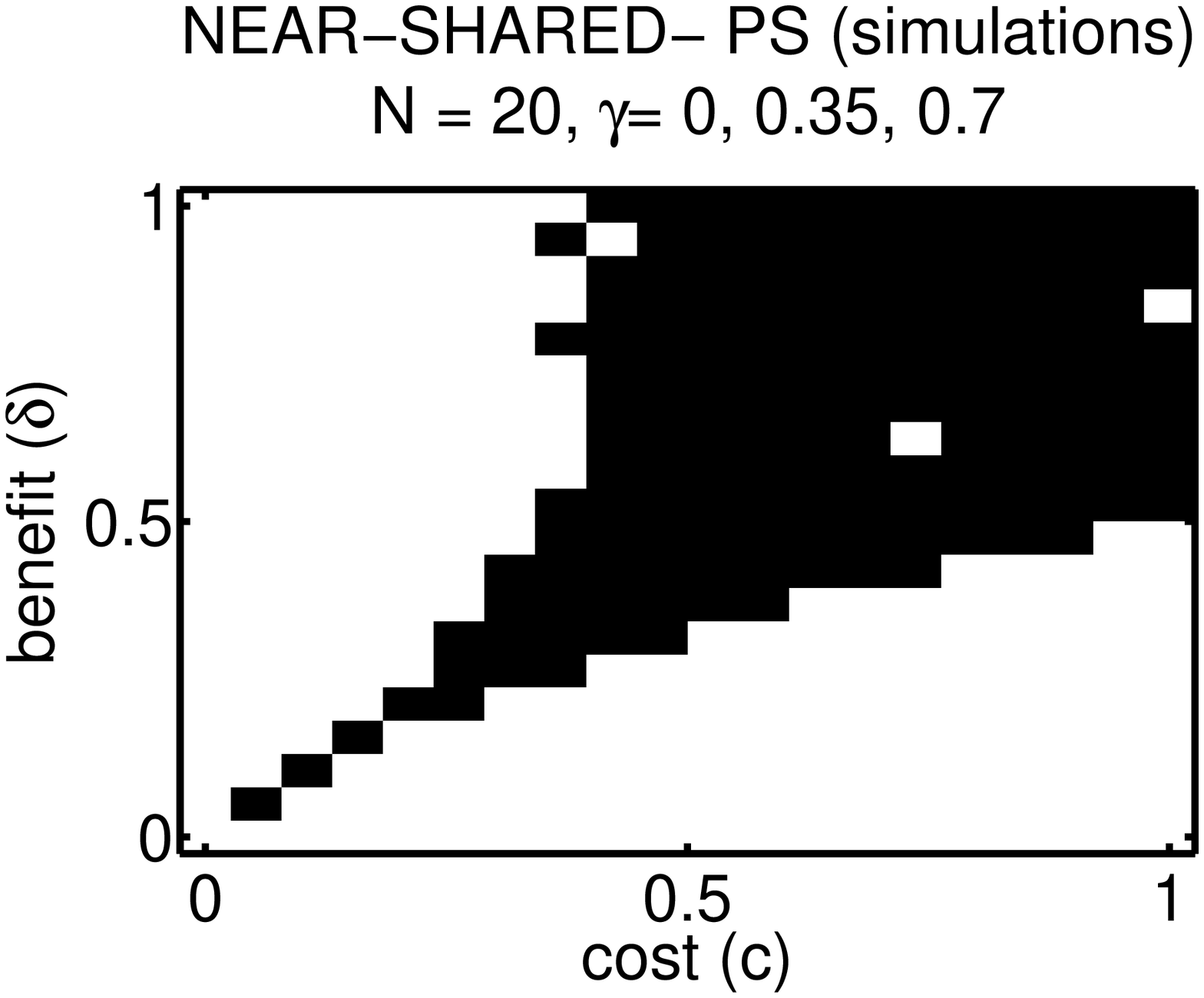}
\end{minipage}
&
\begin{minipage}{3 cm}

\centering
\epsfig{height=3cm, width=3cm, angle=0.0,figure=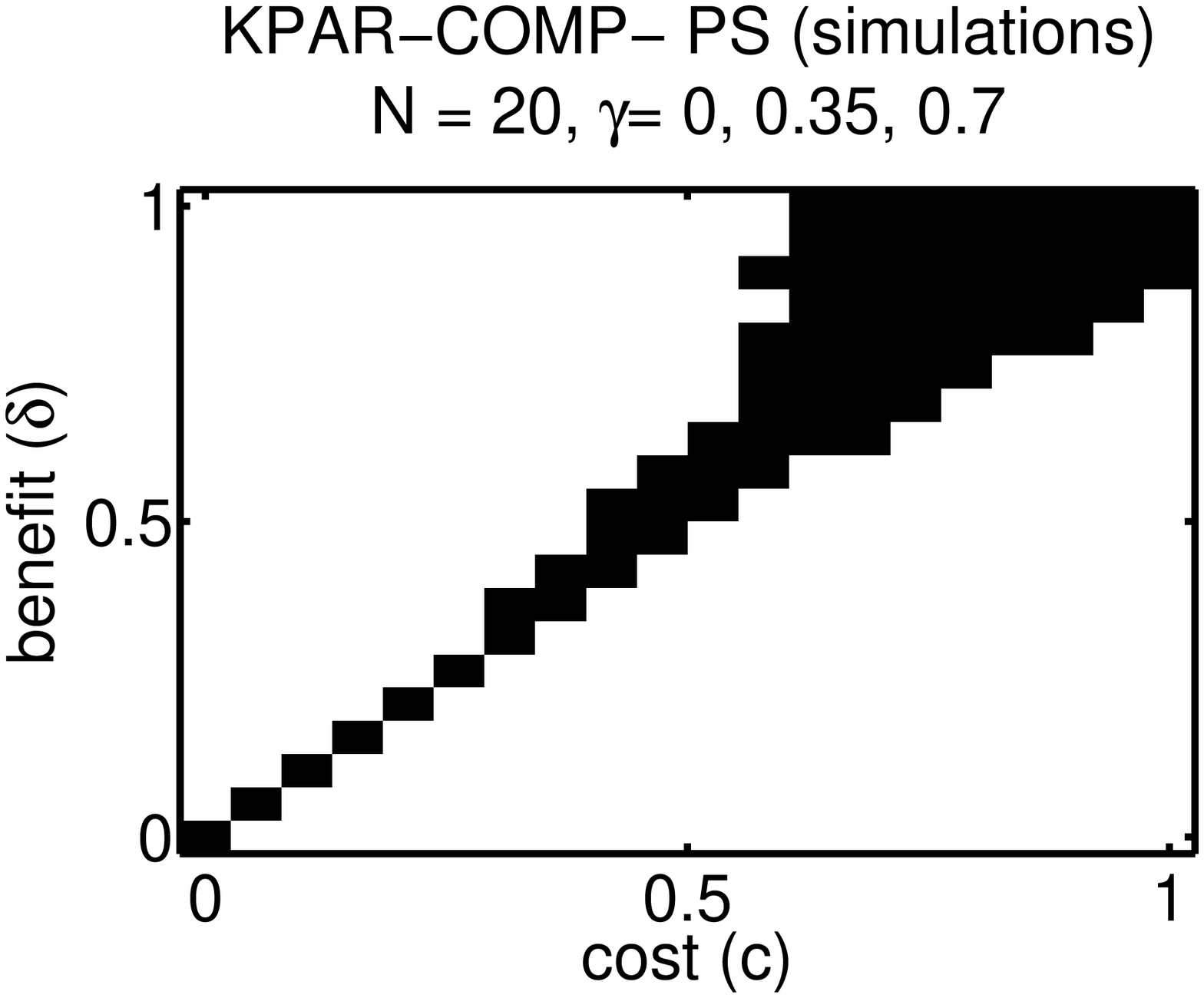}
\end{minipage}
&
\begin{minipage}{3 cm}

\centering
\end{minipage}
& 
\begin{minipage}{3.5cm}

\centering
\end{minipage}\\
(q) & (r) &  &  \\
\end{tabular}
\caption{Validation of theoretical results through simulations [$\text{Repetitions}=100$: for each ($\delta$, $c$) pair ]\label{fig:validation}}
\end{figure*}

We conducted our simulations for all combinations of $\delta$ and $c$ as explained before. Figure~\ref{fig:validation}(a)-Figure~\ref{fig:validation}(r) validate the analytical results derived in Table~\ref{summarytable2} . The vertical axis of each plot in Figure~\ref{fig:validation} is the benefit value ($\delta$), ranging from $0$ to $1$, and the horizontal axis represents the cost parameter ($c$), ranging from $0$ to $1$.
In general, given a particular value of $\delta$ and $c$, there may be multiple network structures that may be pairwise stable. The type of network structure emerging in the network formation process depends on a number of factors like the initial network, the scheduling order of the nodes along with the parameters of $\delta$ and $c$. Hence, we run each simulation run \textit{Num-Repetitions} times each time starting with random schedules and starting with different initial networks with the hope of getting all possible pairwise stable networks. In particular, we start with three different initial networks with densities $(0, 0.35, 0.7)$ respectively as shown in Table~\ref{tab:Simulation-parameters}.

We plot the pairwise stable regions for different networks namely bipartite complete network, null network, complete network, etc and compare with the theoretical predictions. Figure~\ref{fig:validation}(a)-(d) show theoretical results and Figure~\ref{fig:validation}(e)-(r) show the results from the simulations.

Figure~\ref{fig:validation}(e)  shows the regions where the Bipartite Complete (BPC) network emerged as one of the pairwise stable network when the simulation run was started with number of nodes ($N=10$) and initial network with density($\gamma=0$). Clearly, we can see that BPC does not emerge as pairwise stable in the regions where $\delta<c$ as the null network (which coincides with the initial network) is also pairwise stable and the nodes prefer not to add any links to the initial network. However, Figure~\ref{fig:validation}(f) and Figure~\ref{fig:validation}(g) show that if the starting network is already having some existing links then nodes try to form BPC network even in the regions where $\delta < c$. This shows the importance of the initial network in the network formation process.
Figure~\ref{fig:validation}(h) is obtained by merging all the regions of Figure~\ref{fig:validation}(e)-(g) and this closely corresponds to the theoretical predictions of BPC stability shown in Figure~\ref{fig:validation}(a). Figure~\ref{fig:validation}(i)-(l) similarly show results for $N=20$. In this case, however, we observe that Figure~\ref{fig:validation}(l) is not as close to Figure~\ref{fig:validation}(a) which is due to the fact that there may be many more pairwise stable topologies that may emerge as the number of nodes increase which illustrates a fundamental difficulty in characterizing \textit{all} pairwise stable networks for \textit{every} possible value of number of nodes ($N$). 

Another observation is that the complete network is theoretically proven to be the unique pairwise stable network in the region shown in Figure~\ref{fig:validation}(c). We can clearly see the simulation results in Figure~\ref{fig:validation}(h) and Figure~\ref{fig:validation}(l) that this region is clearly excluded from the BPC stable region as starting with any initial network, only the complete graph emerges as unique the pairwise stable network in the region specified by Figure~\ref{fig:validation}(c).

We similarly show the stability regions for complete and null networks in Figure~\ref{fig:validation}(m) and Figure~\ref{fig:validation}(o) respectively which corresponds to the theoretical predictions of Figure~\ref{fig:validation}(b) and Figure~\ref{fig:validation}(d) respectively. As explained earlier, Figure~\ref{fig:validation}(n) again illustrates the importance of initial network in making the null network as the pairwise stable network.

As shown in Proposition~\ref{kpartite-result}, the equi-kpartite network is stable when $\delta=c$ and Figure~\ref{fig:validation}(p) shows that indeed in this region, the equi-kpartite network does emerge as the pairwise stable network when $N=20$. Proposition~\ref{kpartite-result} was only a sufficient condition, we observe from the figure that there are other regions of $\delta$ and $c$ (which we have not analytically characterized) at which equi-kpartite network emerges as the pairwise stable network.

As explained earlier, our characterization of pairwise stable network structures as shown in Table~\ref{summarytable2} is not exhaustive and  hence, we used simulations to depict the region of stability for important types of network structures namely the near-shared network and k-partite complete network. We show the results in Figure~\ref{fig:validation}(q) and Figure~\ref{fig:validation}(r).

\subsection{Emergent Network Topologies During Simulations}
\begin{figure*}[htb!]
\begin{tabular}{cc}
\begin{minipage}{8cm}
\vspace{0.2in}
\centering
\epsfig{height=8cm, width=8cm, angle=0.0,figure=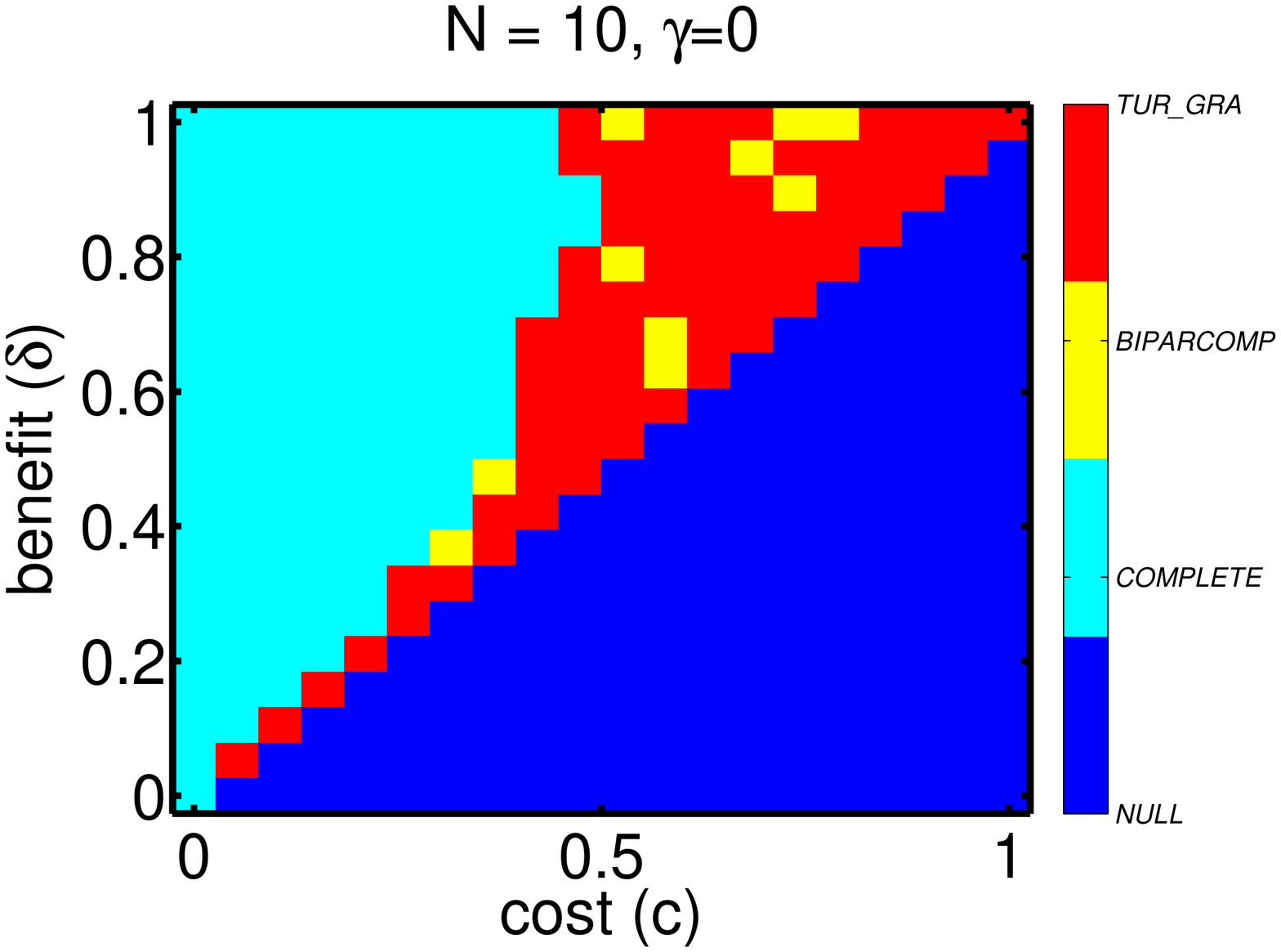, keepaspectratio}
\end{minipage}
&
\begin{minipage}{8cm}
\vspace{0.2in}
\centering
\epsfig{height=8cm, width=8cm, angle=0.0,figure=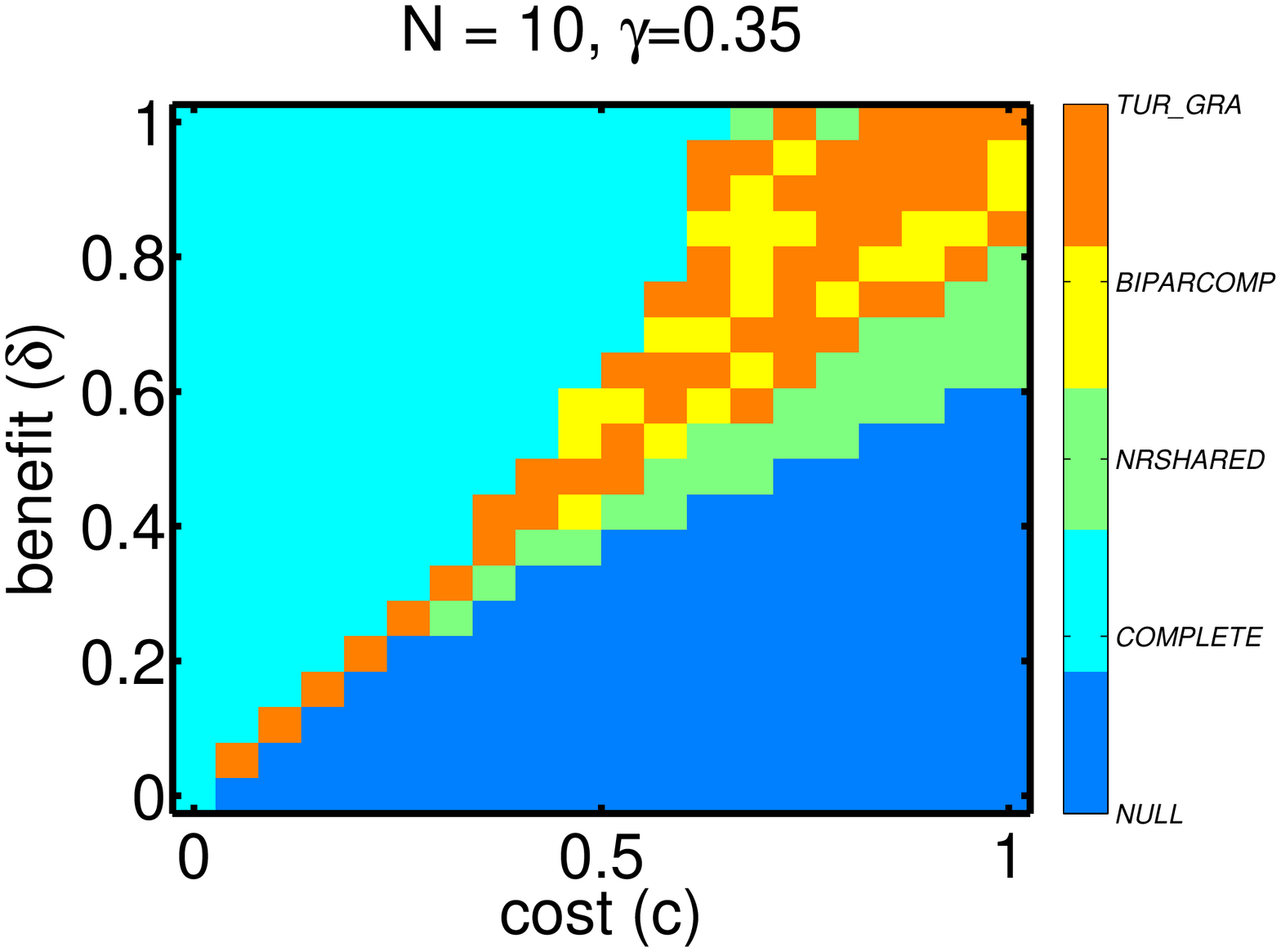, keepaspectratio}
\end{minipage}
\\
\begin{minipage}{8 cm}
\vspace{0.2in}
\centering
\epsfig{height=8cm, width=8cm, angle=0.0,figure=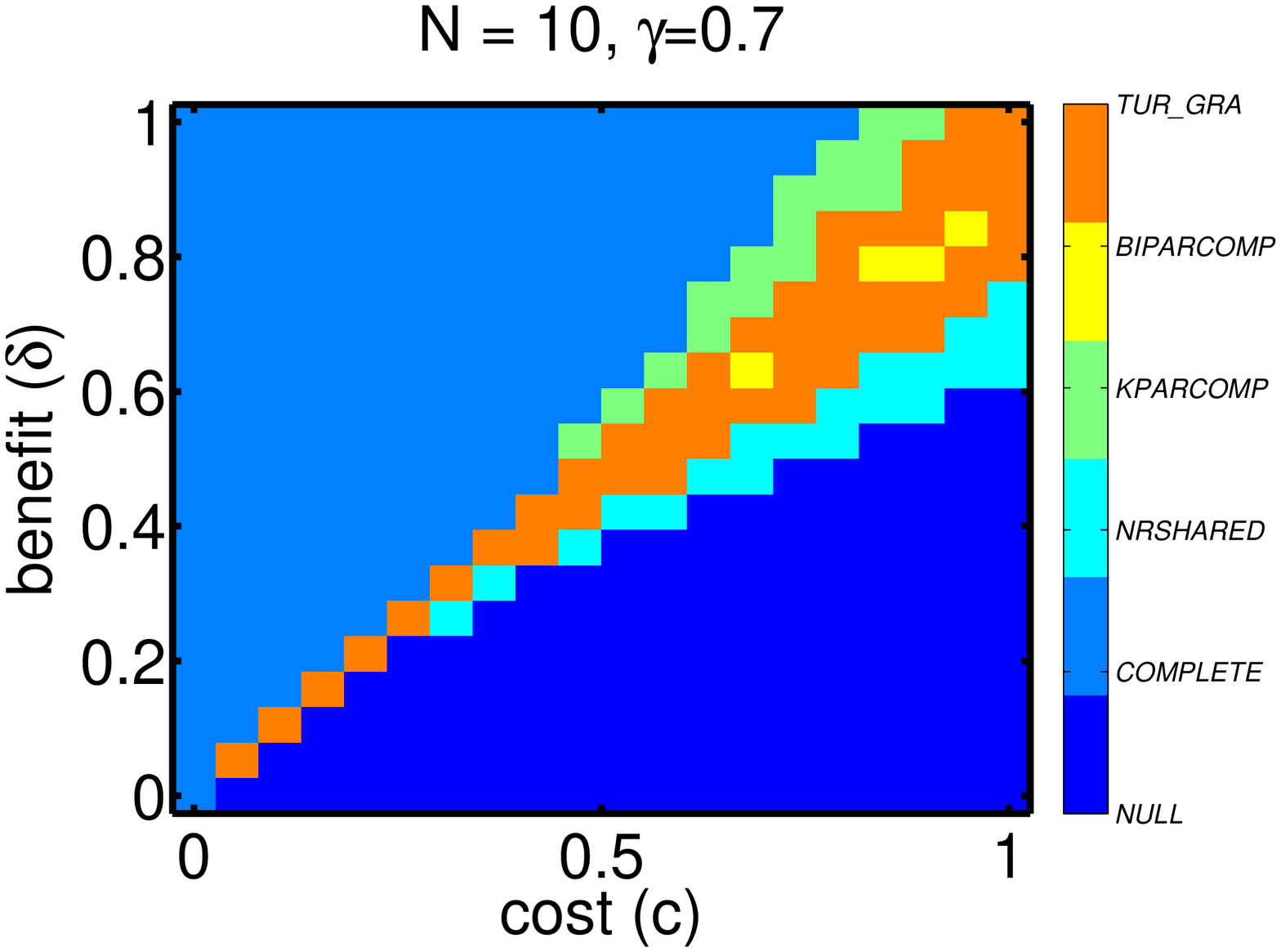, keepaspectratio}
\end{minipage}
&
\begin{minipage}{8cm}
\vspace{0.2in}
\centering
\epsfig{height=8cm, width=8cm, angle=0.0,figure=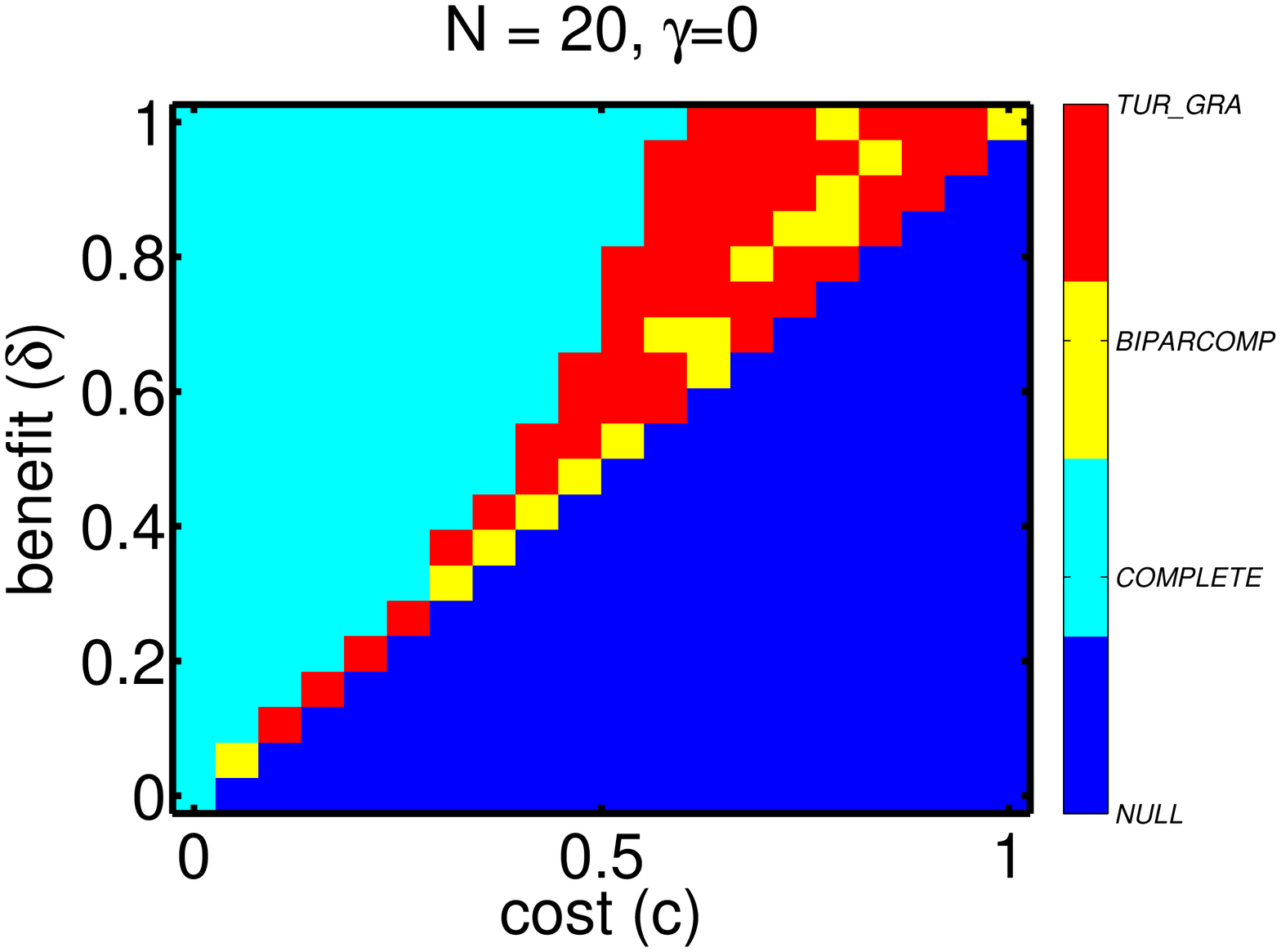, keepaspectratio}
\end{minipage}
\\
\end{tabular}
\\
\begin{tabular}{cc}
\begin{minipage}{8cm}
\vspace{0.2in}
\centering
\epsfig{height=8cm, width=8cm, angle=0.0,figure=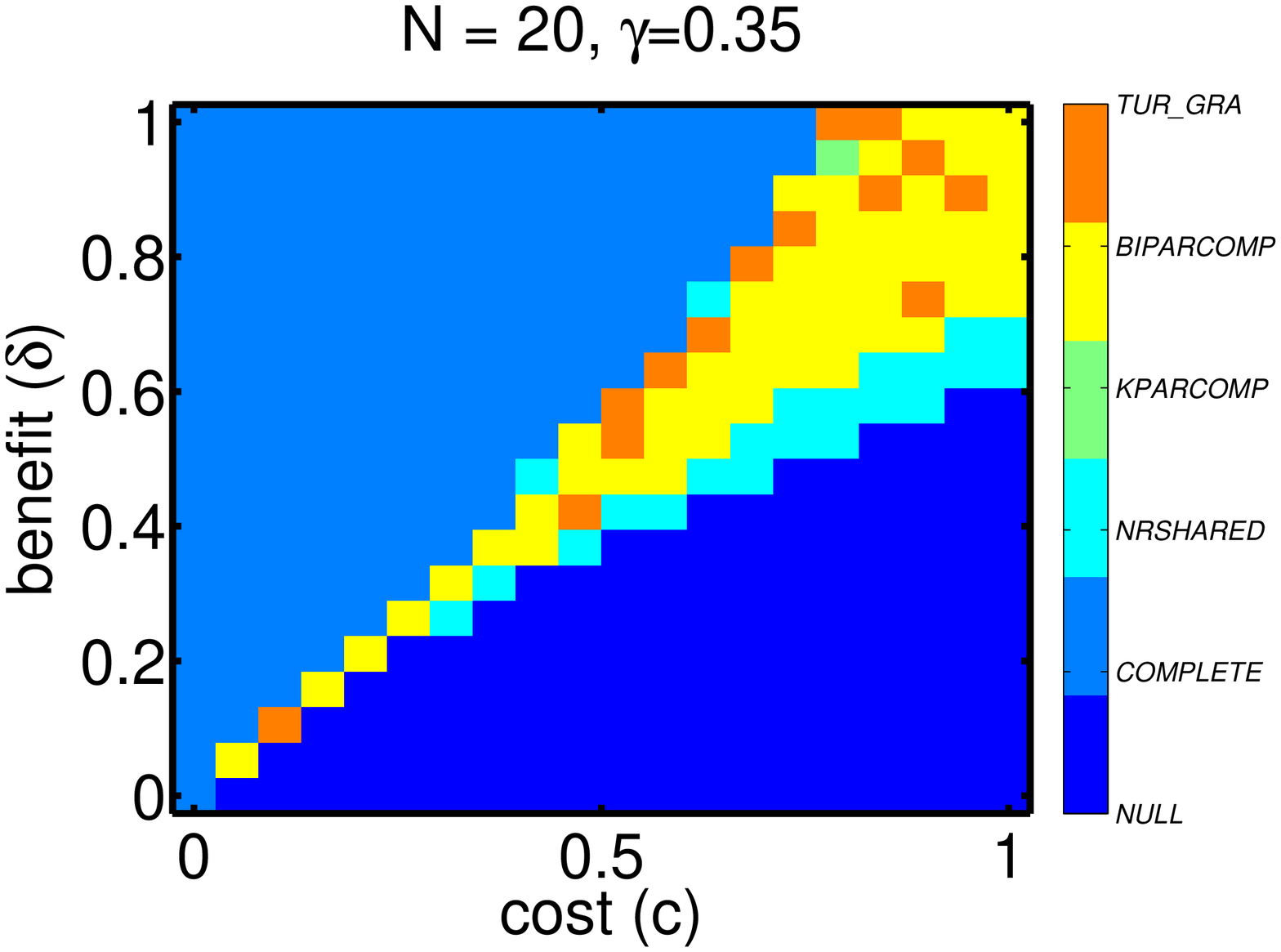, keepaspectratio}
\end{minipage}
&
\begin{minipage}{8 cm}
\vspace{0.2in}
\centering
\epsfig{height=8cm, width=8cm, angle=0.0,figure=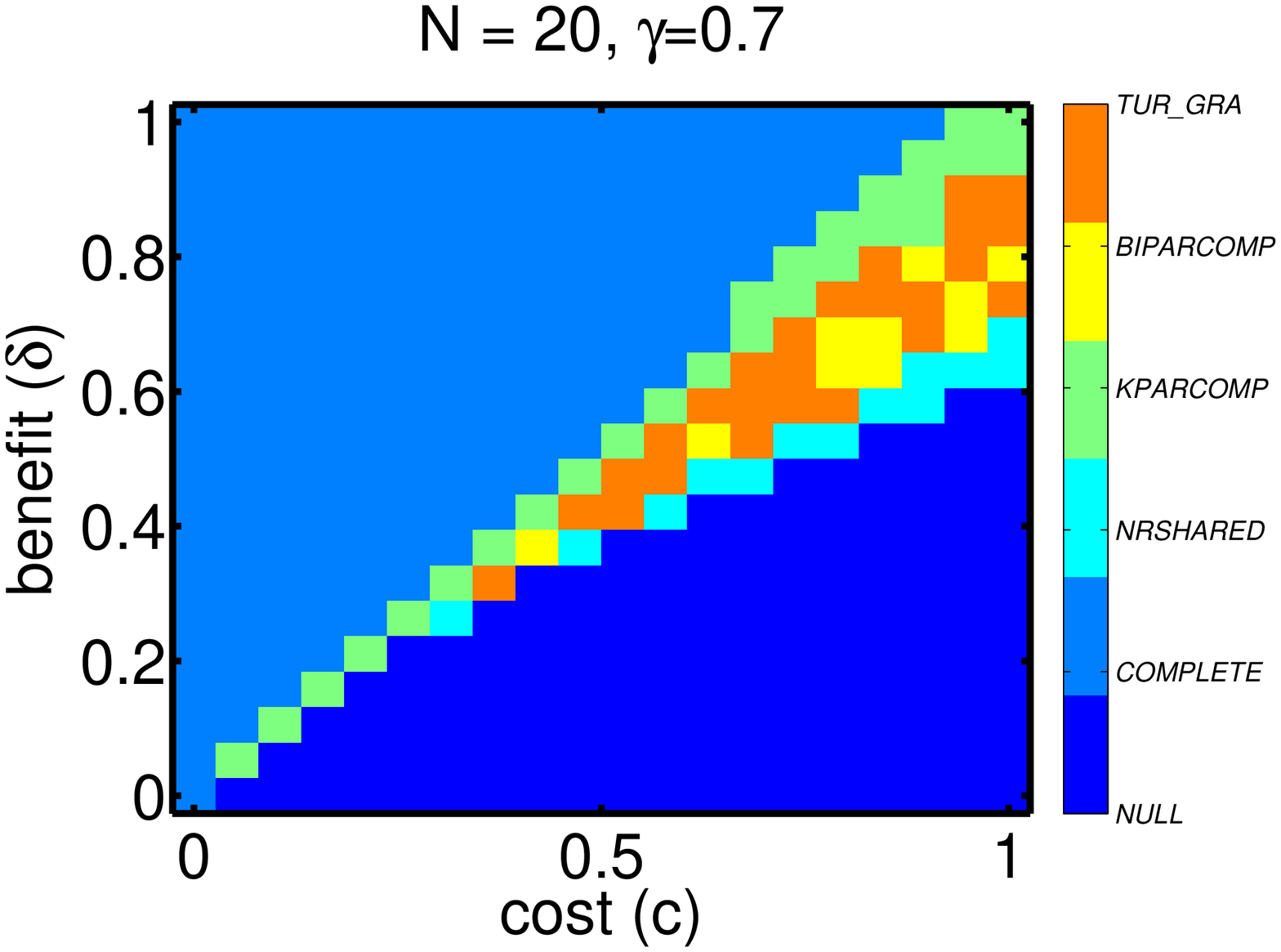, keepaspectratio}
\end{minipage}
\end{tabular}
\\
\caption{Network topologies obtained during simulations \label{fig:Simulation-Results-N}}
\end{figure*}

Figure~\ref{fig:Simulation-Results-N} shows the simulation results for $10$-node and $20$-node networks. The exact parameter configurations and the initial network densities are marked
in Figure~\ref{fig:Simulation-Results-N}. The vertical axis of each plot in Figure~\ref{fig:Simulation-Results-N} is the benefit value ($\delta$), ranging from $0$ to $1$, and the horizontal axis represents the cost parameter ($c$), ranging from $0$ to $1$. As noted earlier, for a $<c,\delta>$ pair, we repeat the simulation for \textit{Num-Repetitions}.
Each repetition for the simulation results in a network that can be classified as one of the structures mentioned in the theoretical analysis. We plot the most frequent (\textit {modal} ) network structure as determined by the frequency with which each of the network structures resulted in \textit{Num-Repetitions} simulation runs. The experiment was repeated starting with different network densities, $\gamma=0, 0.35 \text{ and } 0.7$.
We list some of the abbreviations used in the legends of the plots in Table~\ref{fig:Simulation-Results-N}. 
\begin{table}[h]
\centering
\scriptsize
\begin{tabular}{|c|c||c|c|}
\hline 
TUR\_GRA & Turan Graph & BIPARCOMP& BiPartite Complete \\
\hline
NRSHARED & Near-Shared & KPARCOMP & KPartite Complete \\
\hline 
\end{tabular}
\caption{Some abbreviations used in Figure~\ref{fig:Simulation-Results-N}}\label{abbreviations}
\end{table}

In each of the plots in Figure~\ref{fig:Simulation-Results-N}, we observe that the complete graph is the resultant pairwise stable network (when $\delta > c$, $(\delta - c) \geq \delta^2 $) which  concurs with the theoretical deductions that the complete graph is the unique pairwise stable network in this region (Table~\ref{summarytable2} and Figure~\ref{fig:validation}(c)).

We can also infer from Figure~\ref{fig:validation}(a), Figure~\ref{fig:validation}(b) and Figure~\ref{fig:validation}(d) that there is an overlap in the stability regions among complete and complete bipartite and also between null and complete bipartite networks. However, as observed through simulations (Figure~\ref{fig:Simulation-Results-N}), we see that the complete bipartite network emerges as the \textit{modal} pairwise stable network in its regions of overlap with the aforementioned networks. This can be attributed to the fact there are a large  number of possible bipartite graphs whereas there is only one null network and one complete network. Hence, the likelihood of the null and complete emerging in a region where the bipartite network is also pairwise stable, is small. 

We also observe from some of the plots in Figure~\ref{fig:Simulation-Results-N} that Near-Shared and K-Partite Complete networks emerge as pairwise stable networks under some regions of the parameters. As explained in earlier sections, this can be attributed to the fact that our analytical results (as shown in Table~\ref{summarytable2}) is not exhaustive and there exist some new topologies ( which we characterize as Near-Shared or K-Partite Complete networks) which are also pairwise stable.


\subsection{Network Evolution}

Having studied the macroscopic behaviour of our simulations, we investigate the network formation process from a microscopic viewpoint. We examine various snapshots during the network formation process of a single simulation run which is repeated just once for a fixed parameter of $\delta$ and $c$. We consider $\delta=c=0.5$ as our parameter configuration. We can observe from the our proposed utility model (Equation~\ref{proposedutilitymodel}) that for this configuration the benefits from direct links is $0$ and so, nodes try to maximize the benefits due to bridging behavior. The nodes form/delete links such that they emerge as a bridge in connecting their unconnected neighbors. Hence, we would expect the final pairwise stable network to be consisting of nodes who are filling the positions of structural holes in the network. In other words, the emergent pairwise stable graph should be a \textit{triangle-free} as nodes form links with nodes who are themselves are not connected with each other. 

We depict the snapshots of network formation process in Figure~\ref{fig:NetworkFormationSnapshot}. We can see that initially the nodes are forming links in such a way that triangles are not present but eventually triangles eventually do form due to the cumulative action of other nodes in the network. When triangles emerge in the neighbourhood of a node, it leads to deletion of links from that node (as the node will benefit strictly from deletion) and the final emergent network (Figure~\ref{fig:NetworkFormationSnapshot}(l)) is a bipartite complete network (which is triangle-free) with alternate nodes in the ring layout depiction in Figure~\ref{fig:NetworkFormationSnapshot}(l) belonging to the same partition.

\begin{figure*}[htb!]
\begin{tabular}{||c|c|c||}
\hline
\hline
\begin{minipage}{5.5 cm}
\vspace{0.2in}
\centering
\epsfig{height=4cm, width=4cm, angle=0.0,figure=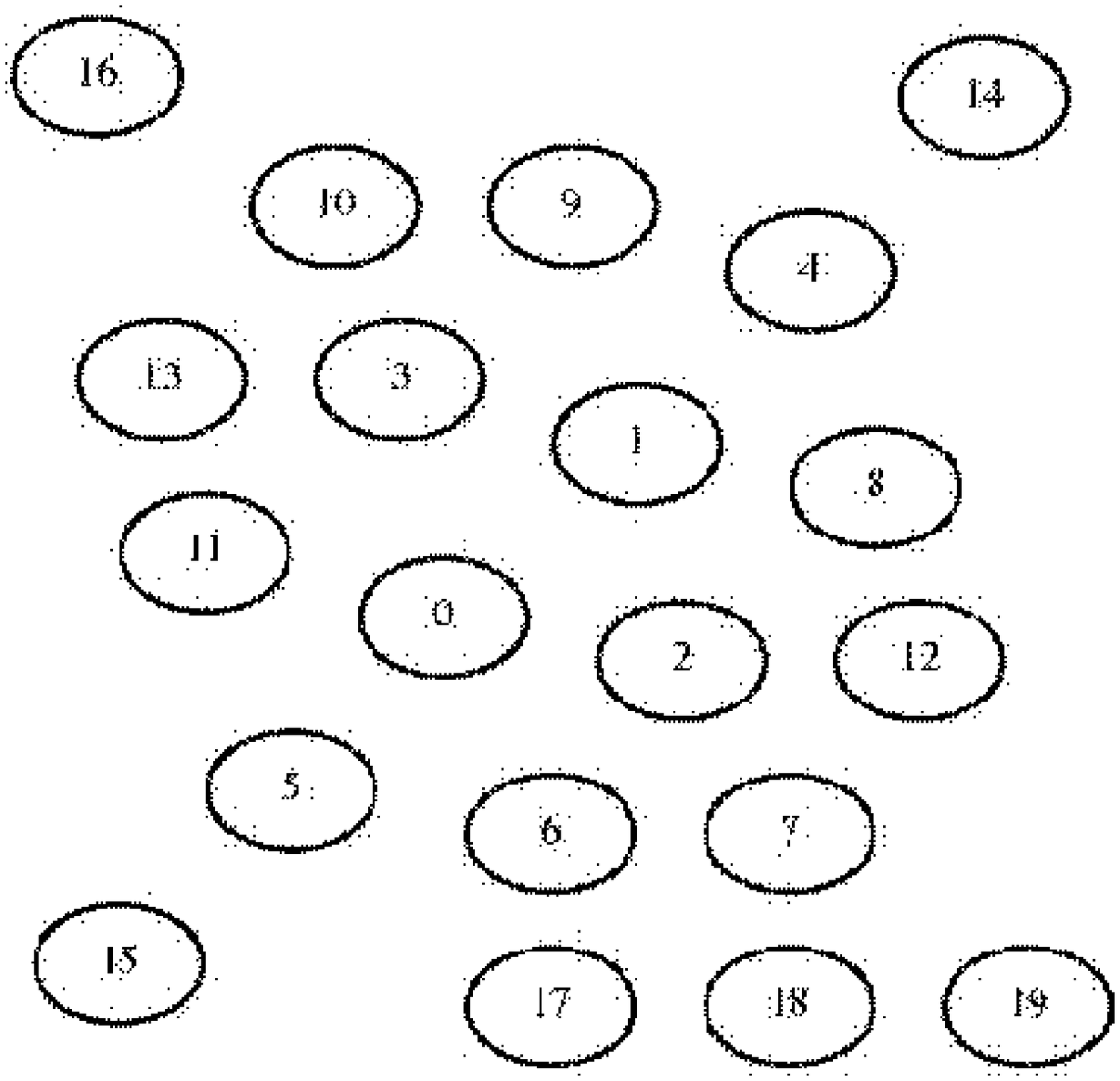, keepaspectratio}
\end{minipage}
&
\begin{minipage}{5.5 cm}
\vspace{0.2in}
\centering
\epsfig{height=4cm, width=4cm, angle=0.0,figure=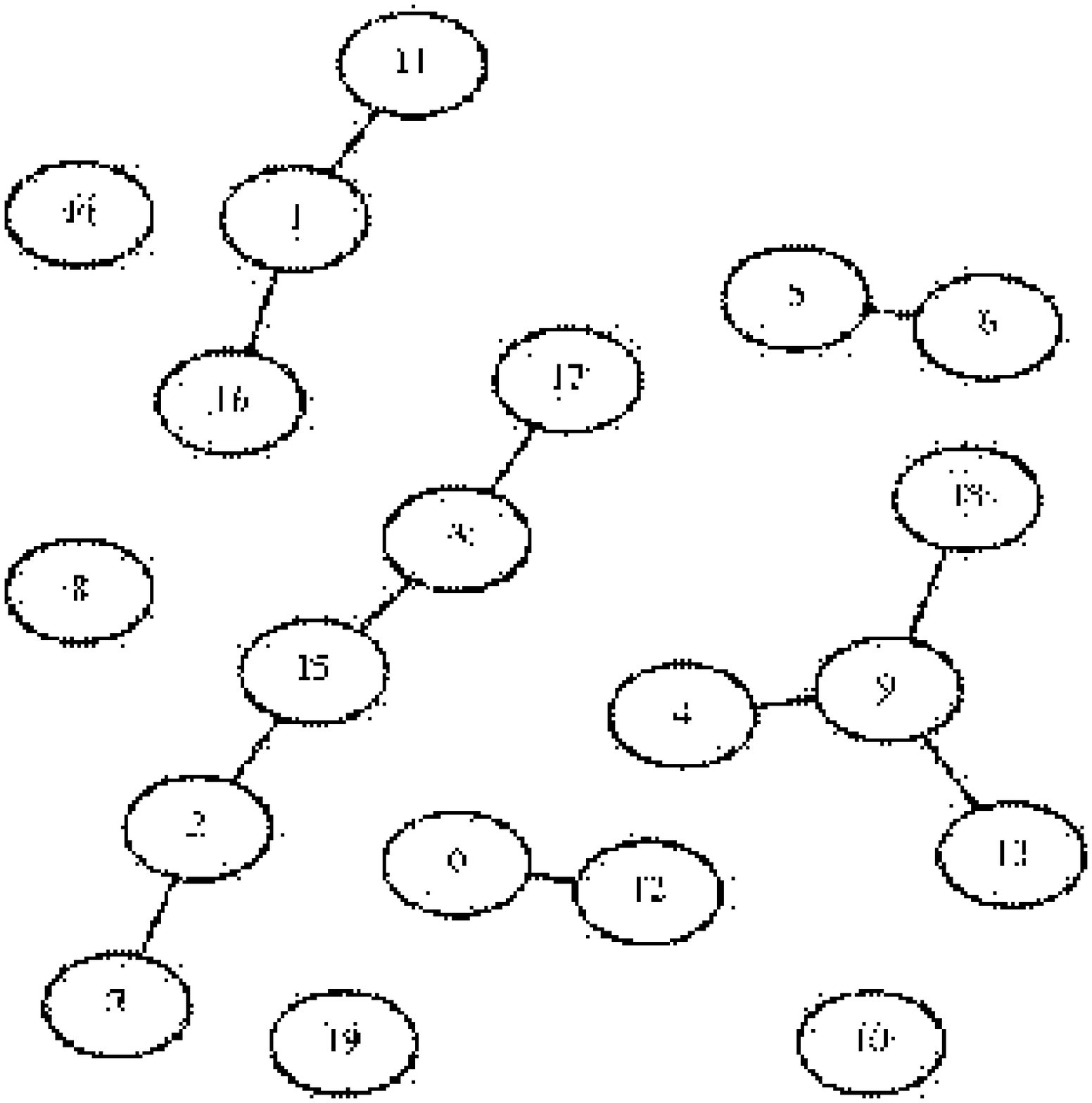, keepaspectratio}
\end{minipage}
&
\begin{minipage}{5.5 cm}
\vspace{0.2in}
\centering
\epsfig{height=4cm, width=4cm, angle=0.0,figure=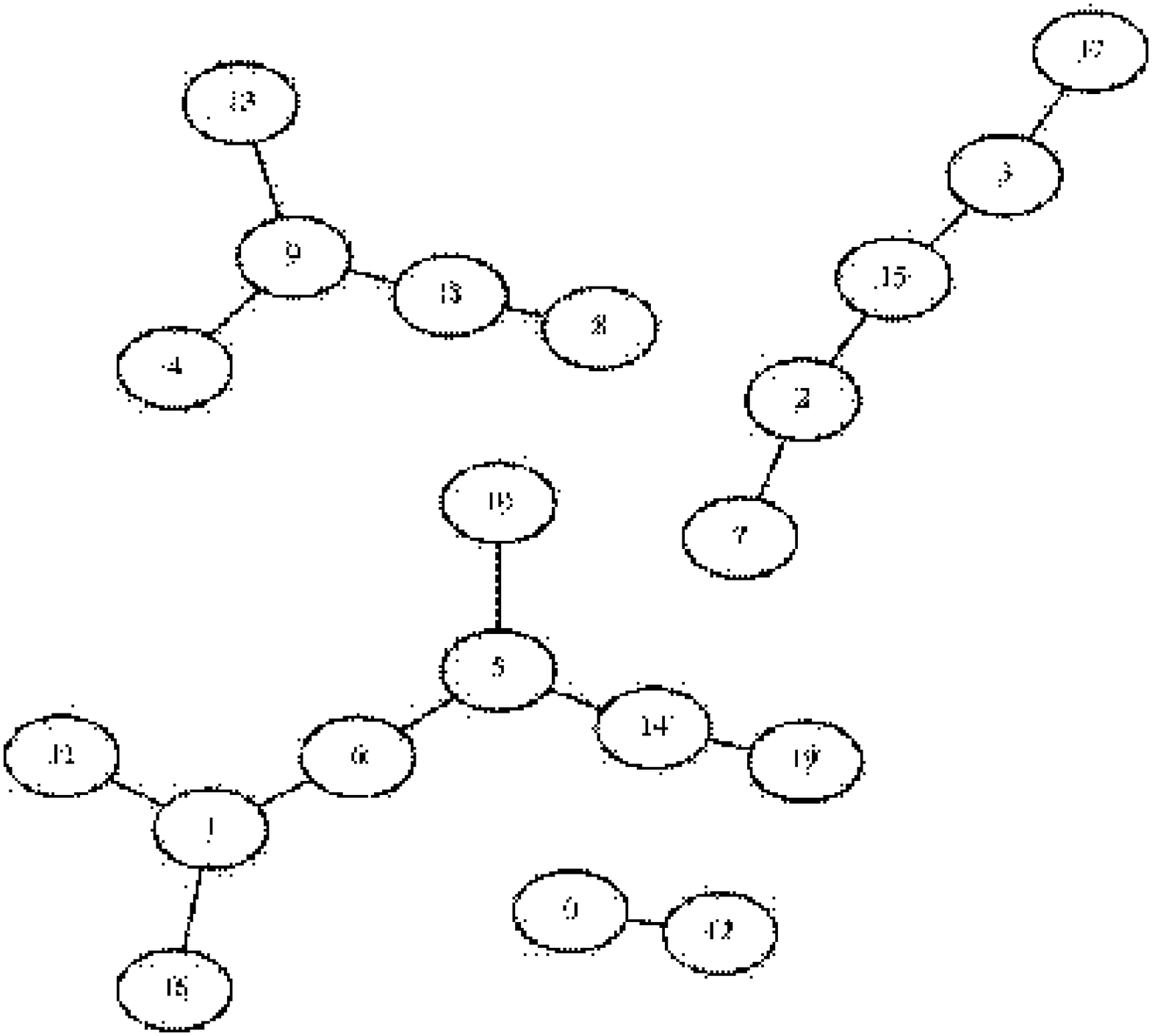, keepaspectratio}
\end{minipage}
\\
(a) & (b) & (c)
\\
\hline
\begin{minipage}{5.5 cm}
\vspace{0.2in}
\centering
\epsfig{height=4cm, width=4cm, angle=0.0,figure=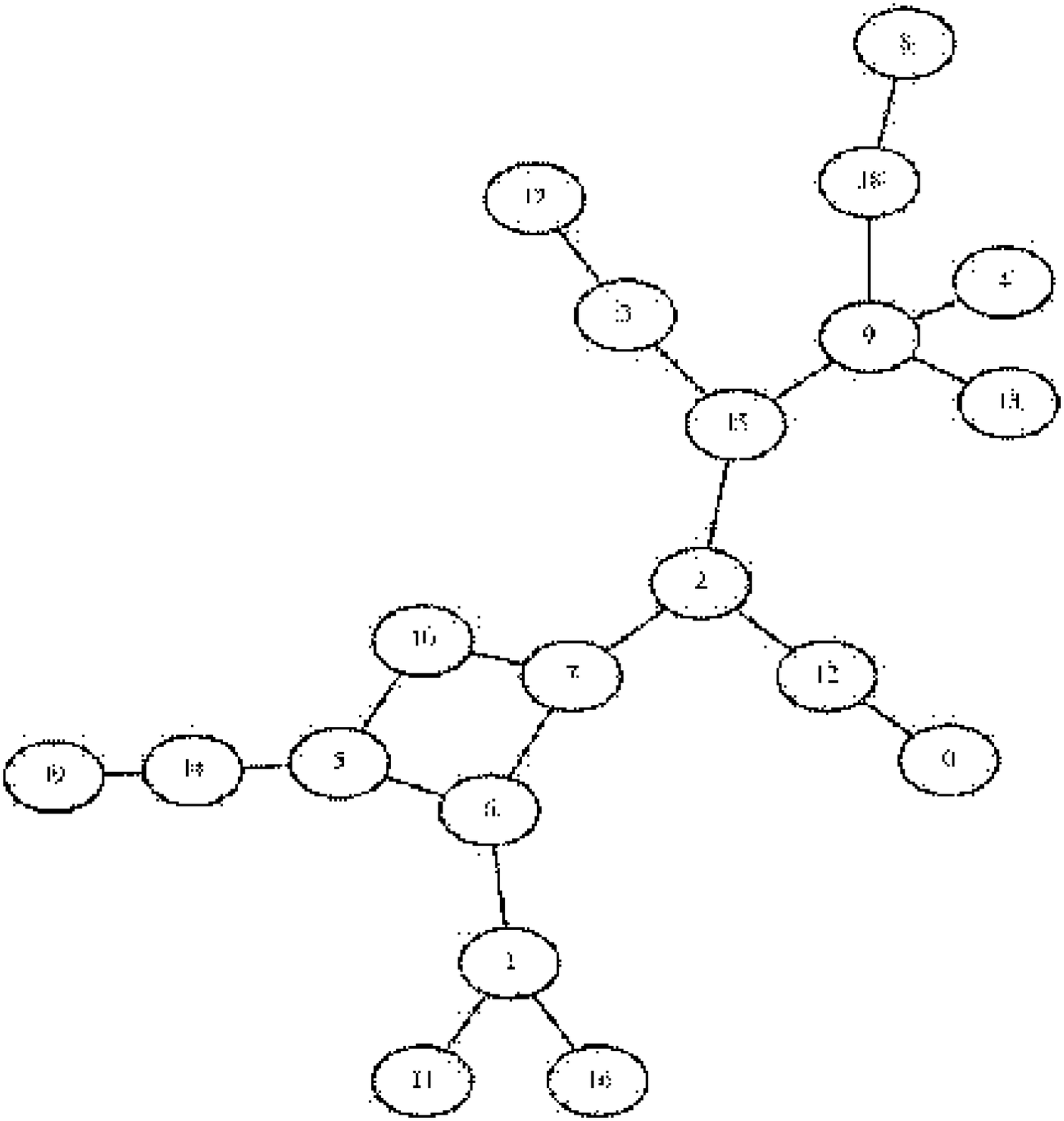, keepaspectratio}
\end{minipage}
&
\begin{minipage}{5.5 cm}
\vspace{0.2in}
\centering
\epsfig{height=4cm, width=4cm, angle=0.0,figure=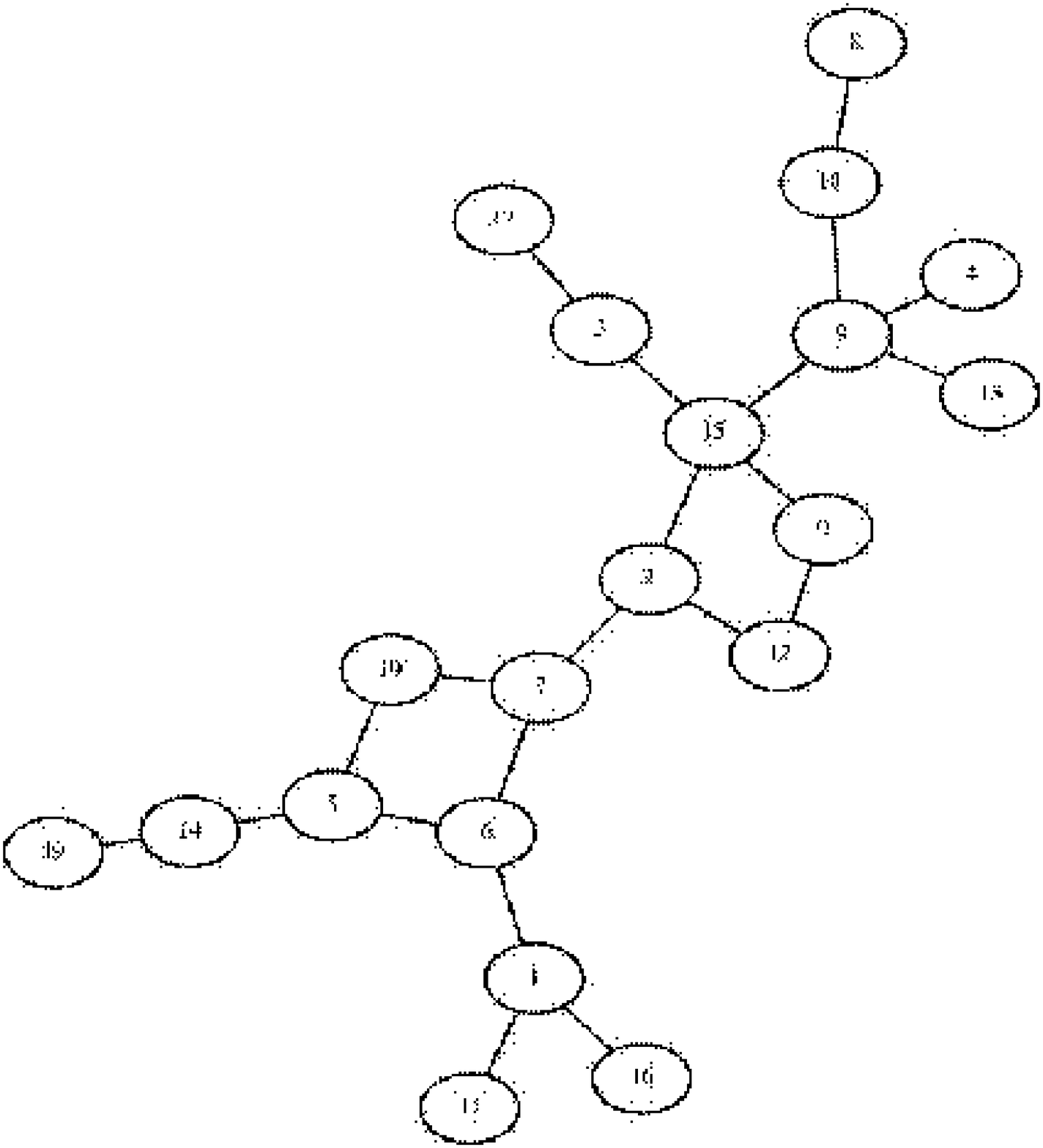, keepaspectratio}
\end{minipage}
&
\begin{minipage}{5.5 cm}
\vspace{0.2in}
\centering
\epsfig{height=4cm, width=4cm, angle=0.0,figure=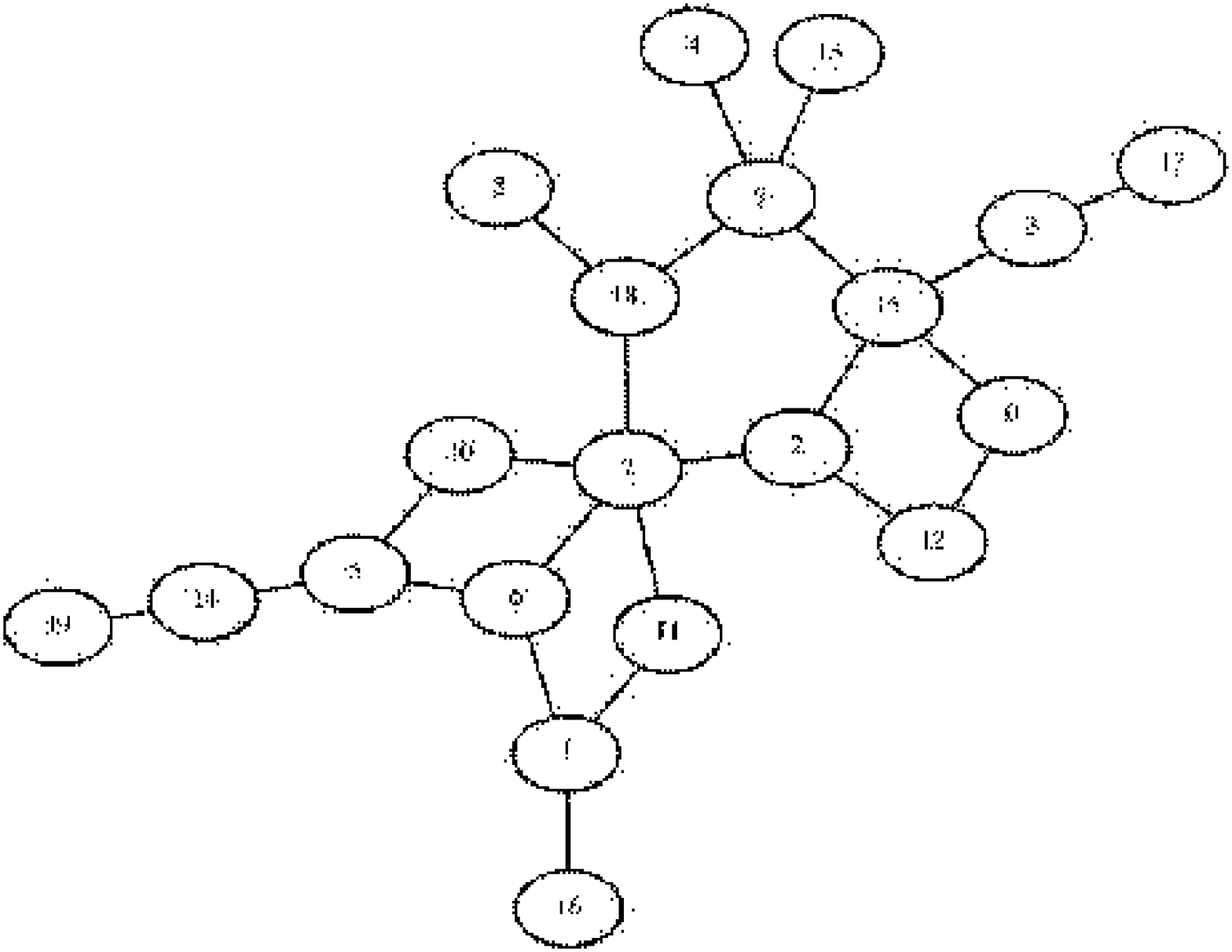, keepaspectratio}
\end{minipage}
\\
(d) & (e) & (f)
\\
\hline
\begin{minipage}{5.5 cm}
\vspace{0.2in}
\centering
\epsfig{height=4cm, width=4cm, angle=0.0,figure=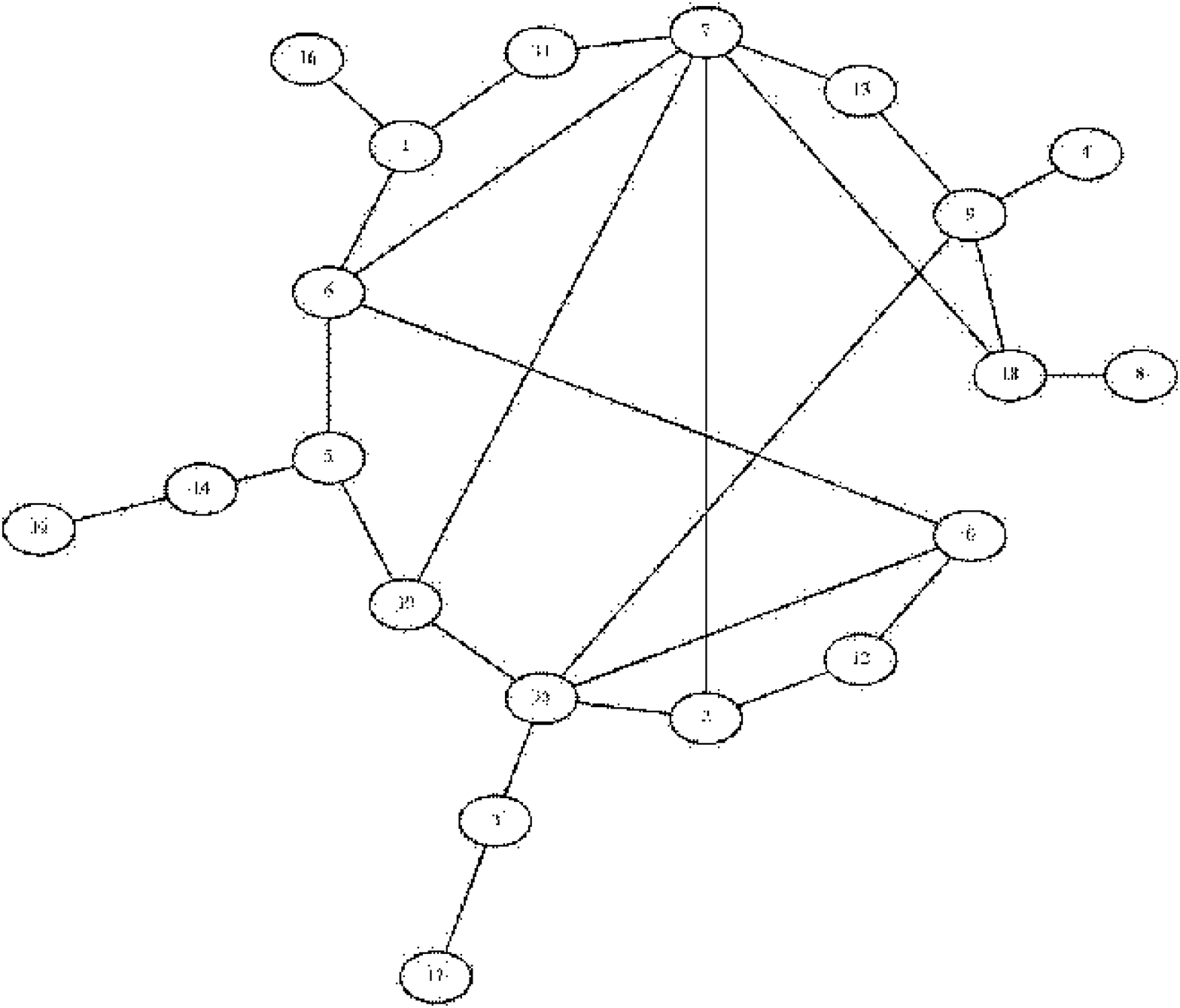, keepaspectratio}
\end{minipage}
&
\begin{minipage}{5.5 cm}
\vspace{0.2in}
\centering
\epsfig{height=4cm, width=4cm, angle=0.0,figure=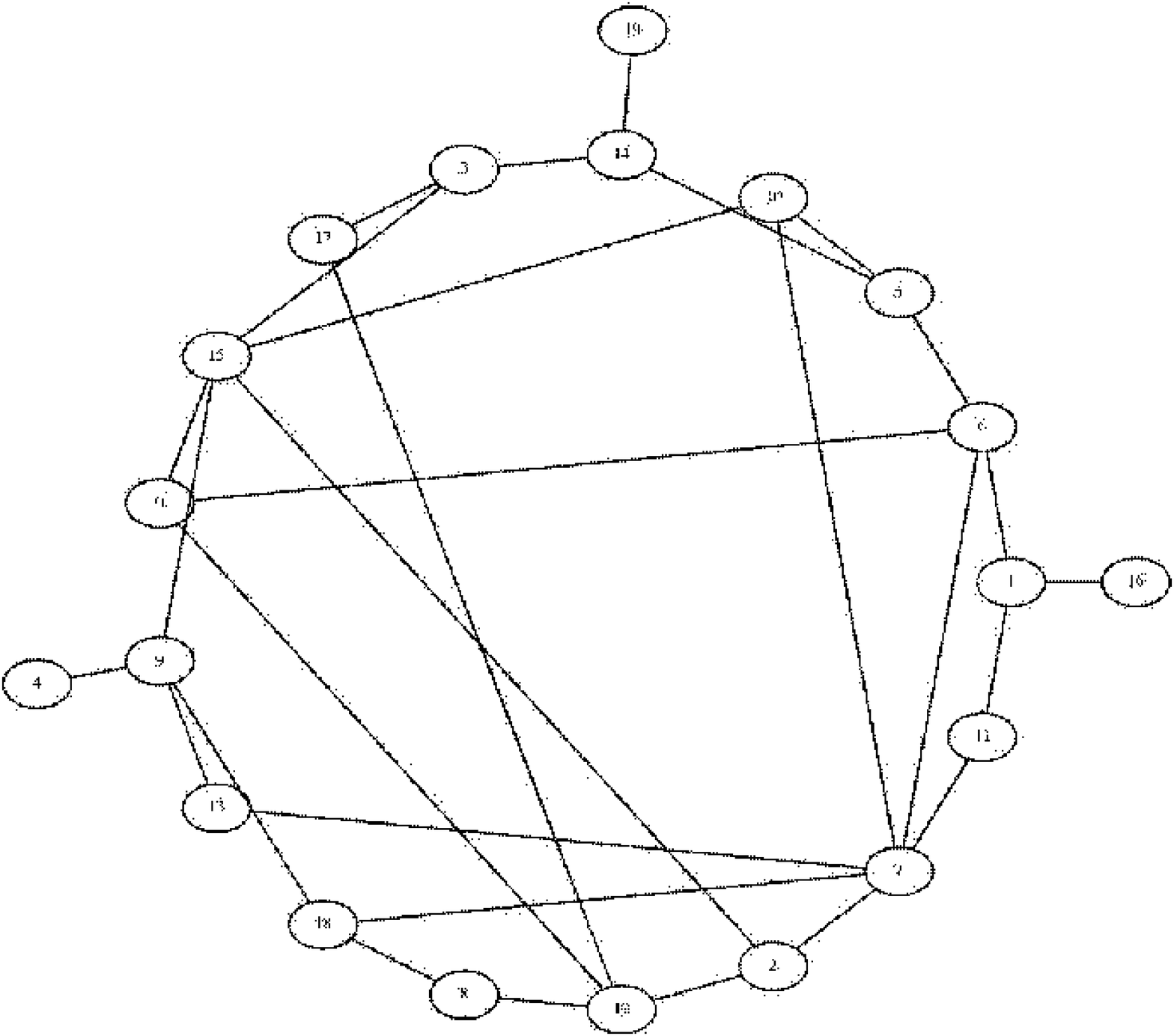, keepaspectratio}
\end{minipage}
&
\begin{minipage}{5.5 cm}
\vspace{0.2in}
\centering
\epsfig{height=4cm, width=4cm, angle=0.0,figure=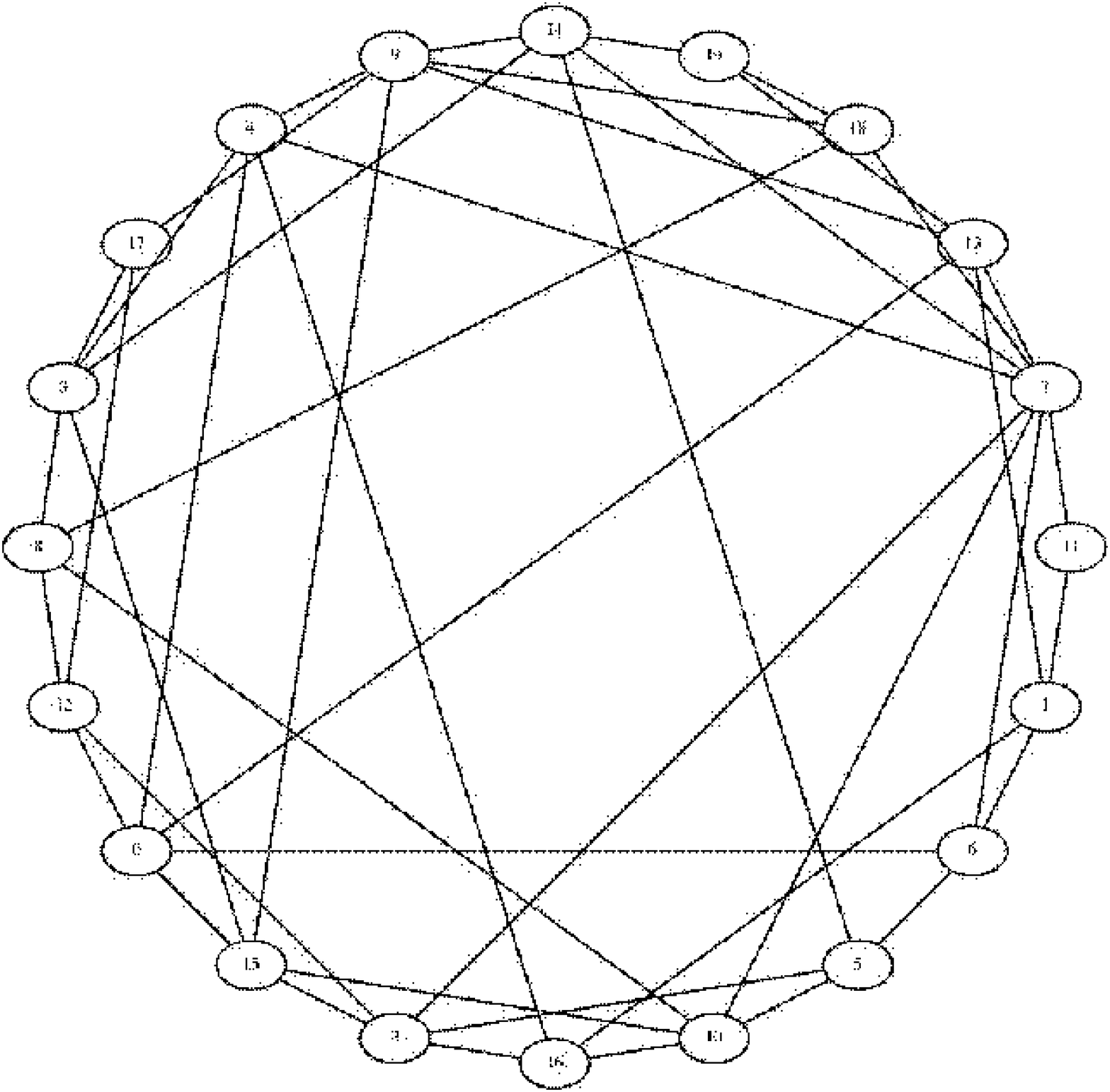, keepaspectratio}
\end{minipage}
\\
(g) & (h) & (i)
\\
\hline
\begin{minipage}{5.5 cm}
\vspace{0.2in}
\centering
\epsfig{height=4cm, width=4cm, angle=0.0,figure=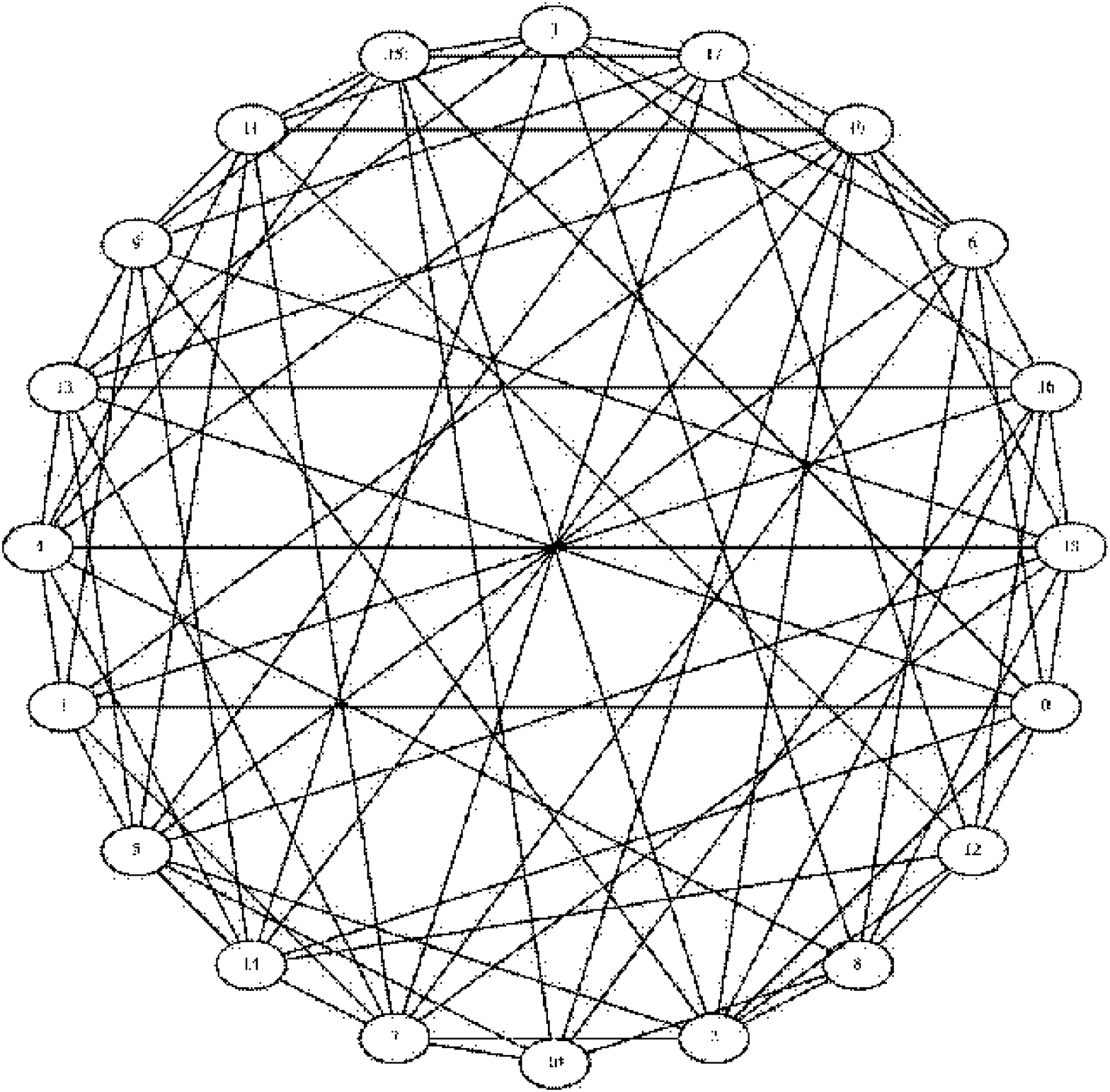, keepaspectratio}
\end{minipage}
&
\begin{minipage}{5.5 cm}
\vspace{0.2in}
\centering
\epsfig{height=4cm, width=4cm, angle=0.0,figure=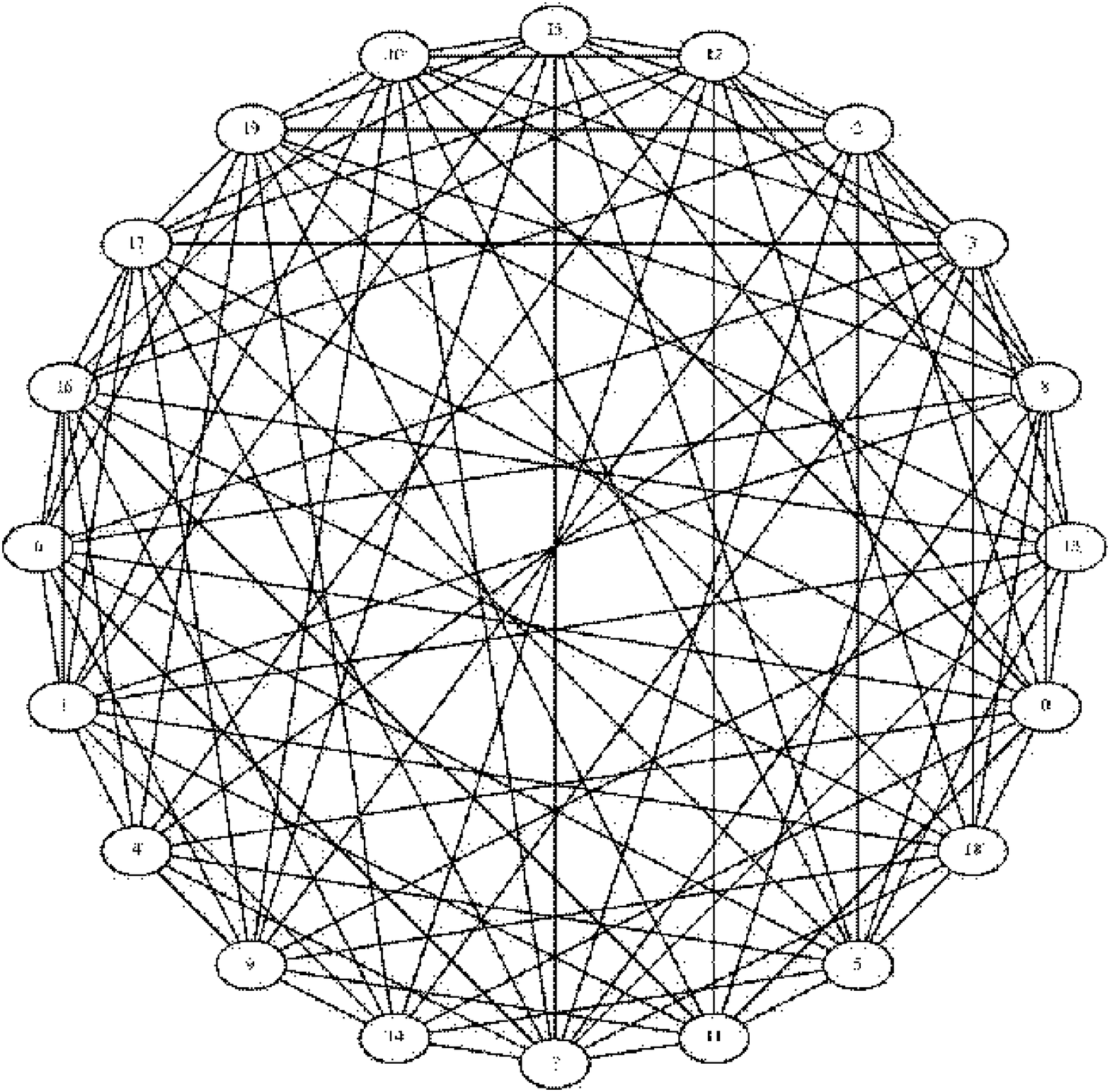, keepaspectratio}
\end{minipage}
&
\begin{minipage}{5.5 cm}
\vspace{0.2in}
\centering
\epsfig{height=4cm, width=4cm, angle=0.0,figure=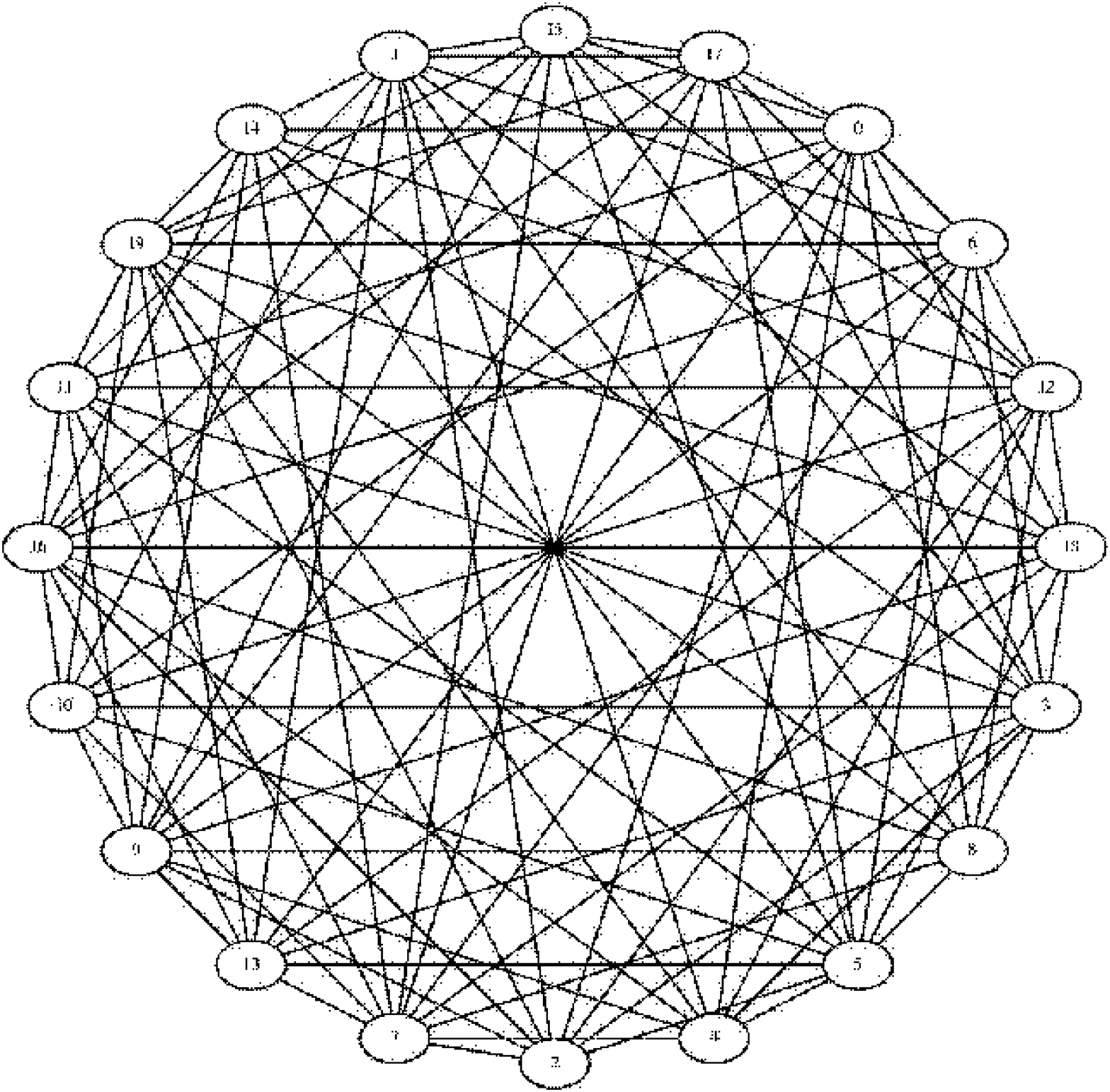, keepaspectratio}
\end{minipage}
\\
(j) & (k) & (l)
\\
\hline
\hline
\end{tabular}
\caption{Evolution of the network formation process ($N=20, \delta=0.5, c=0.5$)\label{fig:NetworkFormationSnapshot}}
\end{figure*}

In complex network literature, the number of triangles in the network is a important parameter which was first studied by Watts and Strogatz~\cite{watts:98} by definition the notion of \textit{clustering}, sometimes also known as network transitivity. Clustering refers to the increased propensity of pairs of people to be acquainted with one another if they have another acquaintance in common. Watts and Strogatz~\cite{watts:98} define a \textit{clustering coefficient} (denoted by $C$) that measures the degree of clustering in a undirected unweighted graph.
\begin{align}
\nonumber C &= \displaystyle \frac{3 \times \text{Number of triangles on the graph}}{\text{Number of connected triples of vertices}} 
\end{align}
The factor three accounts for the fact that each triangle can be seen as consisting of three different connected triples, one with each of the vertices as central vertex, and assures that $0 \leq C \leq  1$. A triangle is a set of three vertices with edges between each pair of vertices; a connected triple is a set of three vertices where each vertex can be reached from each other (directly or indirectly), i.e. two vertices must be adjacent to another vertex (the central vertex). 


It can be observed from the utility model proposed in equation (\ref{proposedutilitymodel}) in Section~\ref{utilitymodel} that $\Biggl(\displaystyle\frac{\sigma_i}{{d_i \choose 2}}\Biggr)$ component in the utility model corresponds to the clustering coefficient of node $i$. Thus, in our utility model, nodes benefit from having lesser clustering coefficient as this will lead to the formation of structural holes, which in turn leads to increase in the payoff for the node. We elaborate more on this when we discuss efficient network topologies in Section~\ref{sec:Efficiency}.

\begin{figure*}[h]
\begin{tabular}{cc}
\begin{minipage}{8 cm}
\vspace{0.2in}
\centering
\epsfig{height=7.5cm, width=7.5cm, angle=0.0,figure=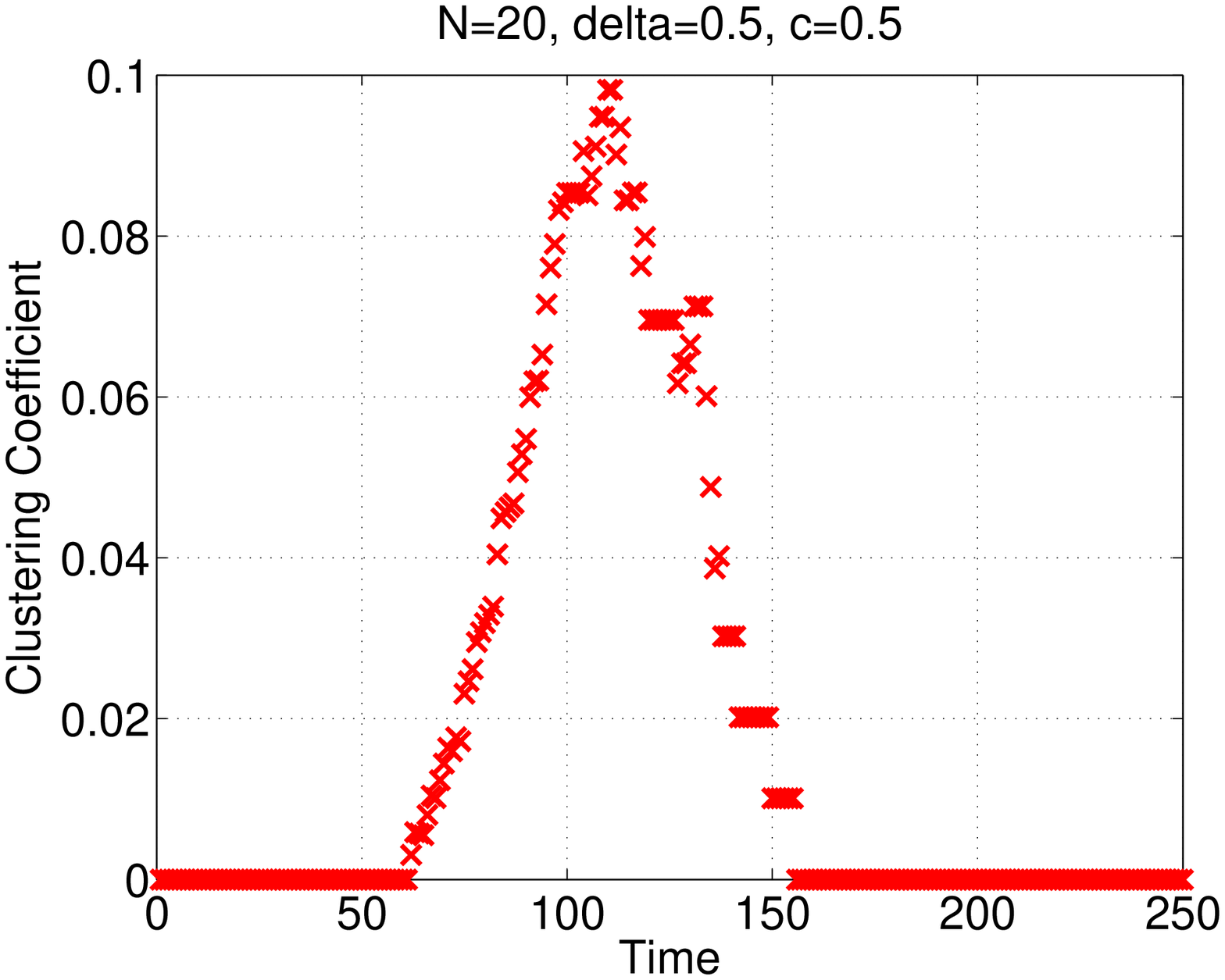}
\end{minipage}
&
\begin{minipage}{8 cm}
\vspace{0.2in}
\centering
\epsfig{height=7.5cm, width=7.5cm, angle=0.0,figure=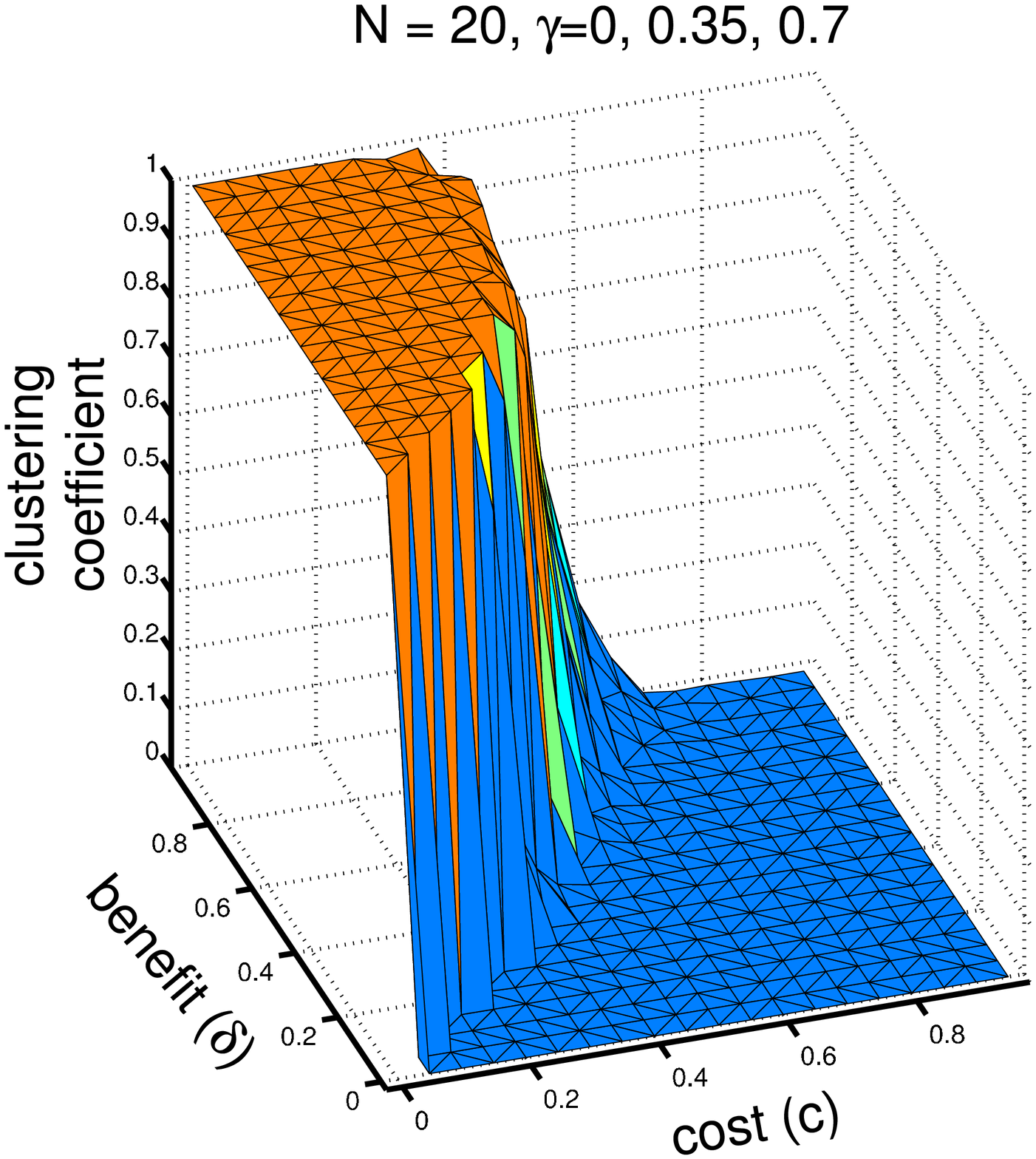}
\end{minipage}
\\
(a) & (b) 
\end{tabular}
\caption{Study of Clustering Coefficient ($N=20$)\label{fig:clusteringcoefficient}}
\end{figure*}

We now study how the clustering coefficient changes as the network evolves through the different phases shown in Figure~\ref{fig:NetworkFormationSnapshot}. We plot this result in Figure~\ref{fig:clusteringcoefficient}(a). We see that upto time epoch $50$ clustering coefficient is $0$. Later there is a increase in the value which is followed by the reduction in the clustering coefficient back to $0$ (at time epoch $150$) when the pairwise stable network emerges. As explained before, this is indeed the expected behaviour during the network formation process for the parameters $\delta=c=0.5$. 

We also study the average clustering co-efficient in all the pairwise stable networks that emerge for different values of $\delta$ and $c$. We take the average over running \textit{Num-repetitions} number of times. The result is shown in the 3d plot in Figure~\ref{fig:clusteringcoefficient} (b). We can see that the clustering coefficient assumes value of $1$ in the regions where the complete network is stable and $0$ when the null network is stable. In other regions, the clustering coefficient value is between $0$ and $1$ which indicates a tradeoff between the benefits from direct links and the benefits from bridging benefits to the nodes in the network.

\begin{figure*}[h]
\begin{tabular}{ccc}
\begin{minipage}{5.5 cm}
\vspace{0.2in}
\centering
\epsfig{height=5.5cm, width=5.5cm, angle=0.0,figure=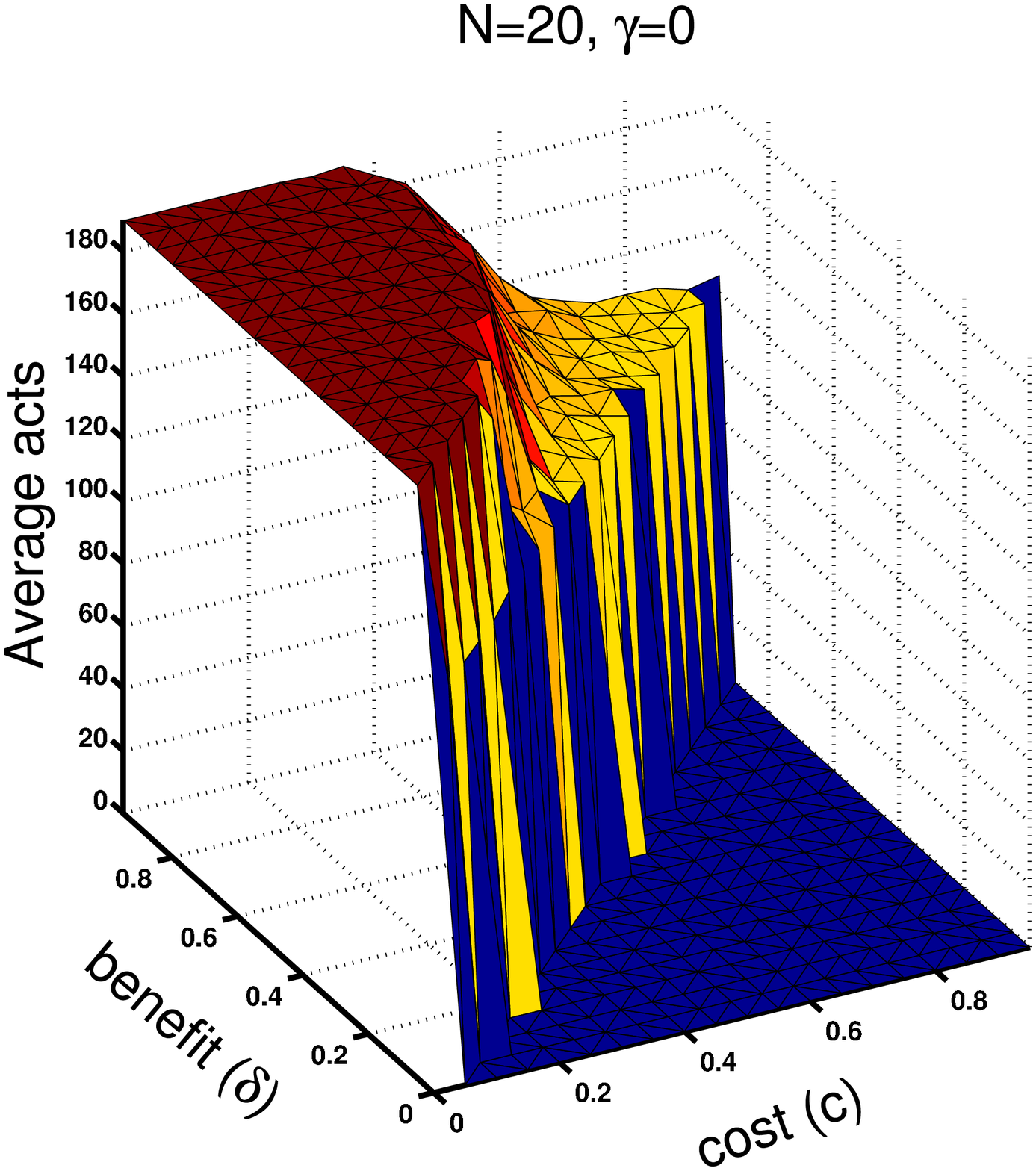, keepaspectratio}
\end{minipage}
&
\begin{minipage}{5.5 cm}
\vspace{0.2in}
\centering
\epsfig{height=5.5cm, width=5.5cm, angle=0.0,figure=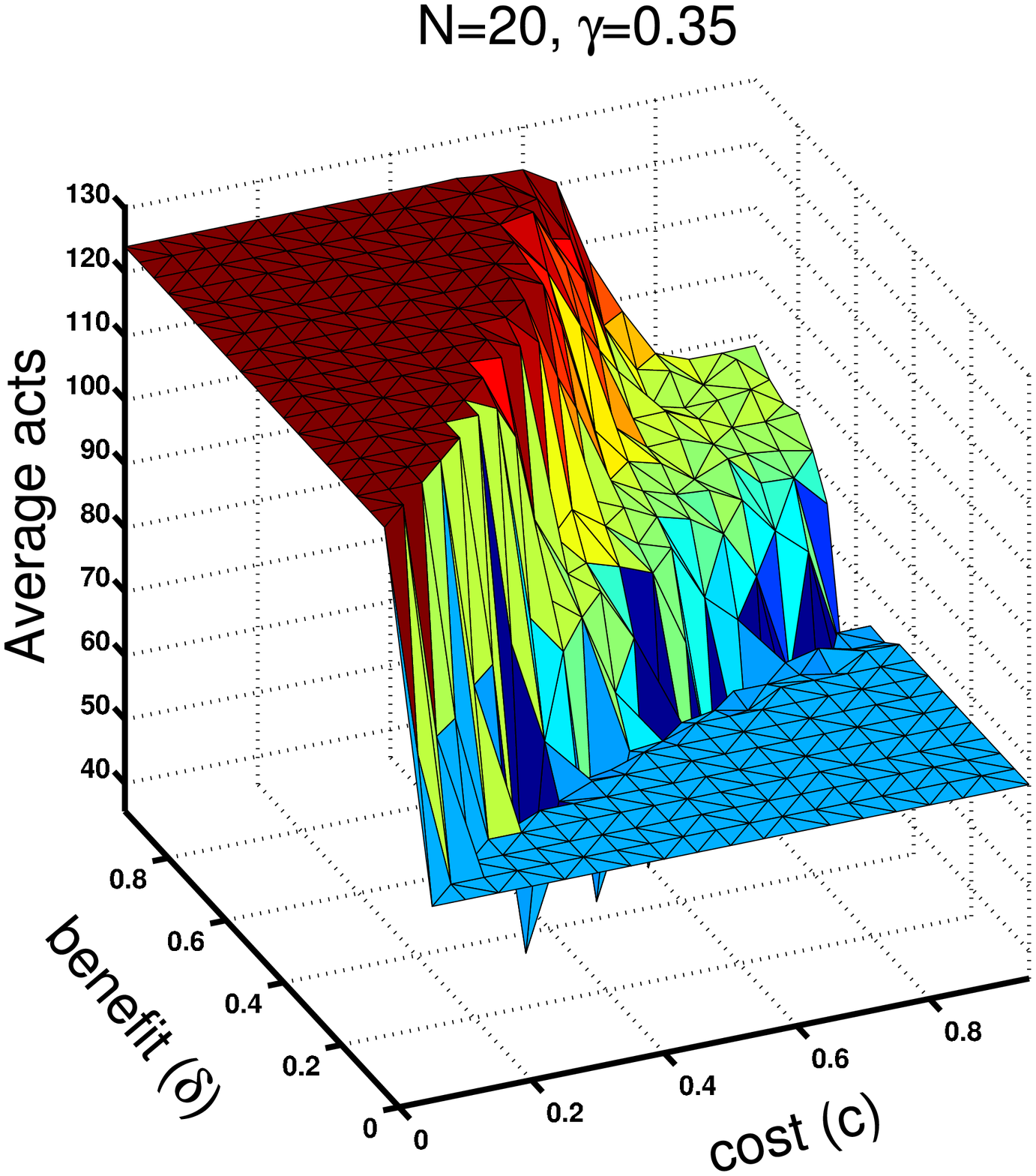, keepaspectratio}
\end{minipage}
&
\begin{minipage}{5.5 cm}
\vspace{0.2in}
\centering
\epsfig{height=5.5cm, width=5.5cm, angle=0.0,figure=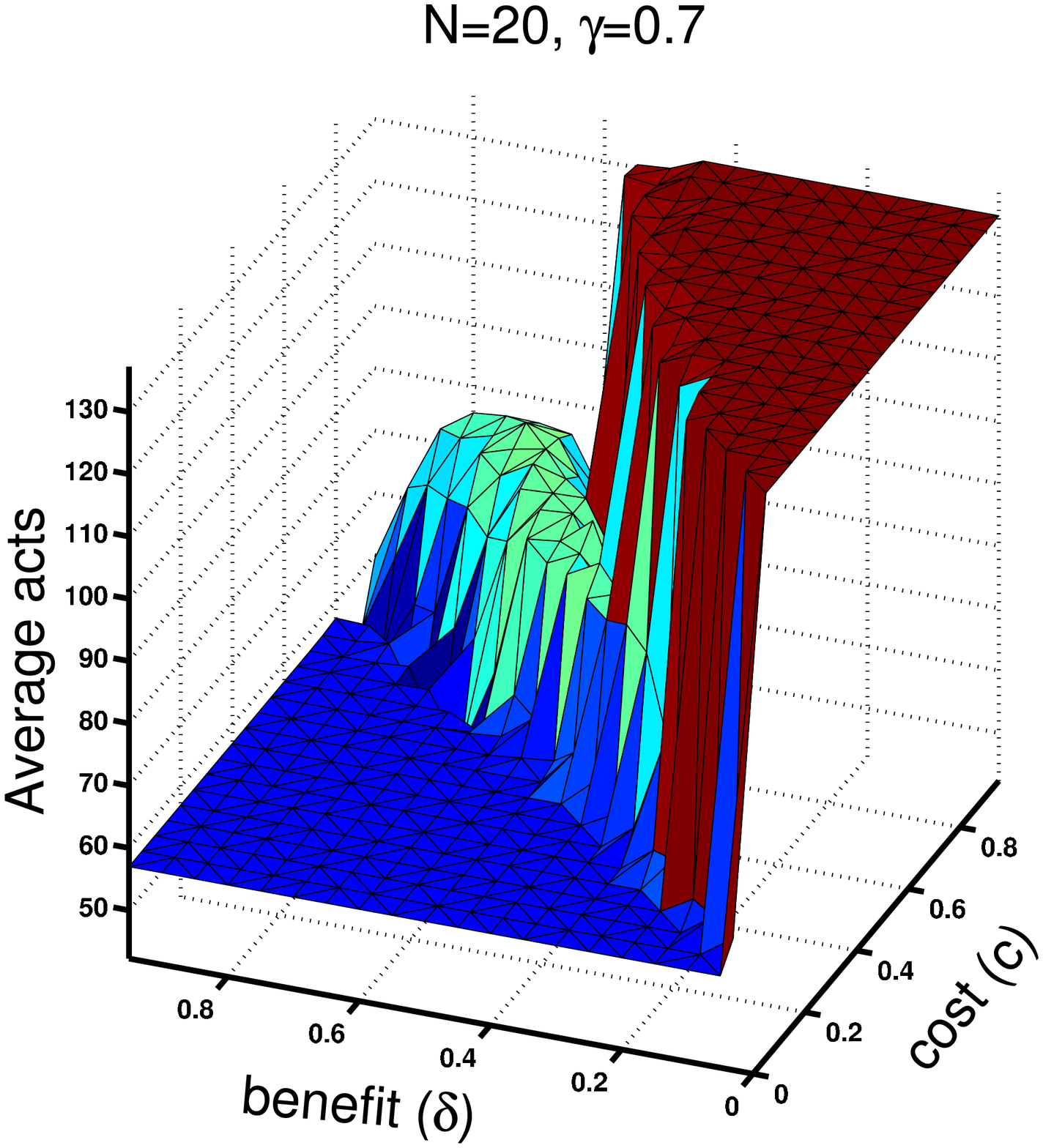, keepaspectratio}
\end{minipage}
\\
(a) & (b) & (c)
\end{tabular}
\caption{Study of number of acts before converging to pairwise stable network ($N=20$ )\label{fig:meanacts}}
\end{figure*}

\subsection{Average Number of Actions before Convergence}

In this section, we will study the effect the initial network density has on the effort needed by the nodes to achieve convergence to a pairwise stable network. A single addition of an edge or a single deletion of an edge by a node is considered to be a single `act' by that player. We now study the mean number of acts performed by the players to converge to a pairwise stable network starting from various initial random networks. We can see from Figure~\ref{fig:meanacts}(a) that the number of changes to the network is more when the $\delta > c$ region and this is because the initial network is a null network and the players need to perform a lot more additions/deletions to the network before reaching the final stable network which is the complete network. When $\delta <c$, the players need not perform any change to the network as the initial null network is already pairwise stable. In fact, we can observe from the Figure~\ref{fig:meanacts} that the number of acts needed to reach the complete network is maximum (about $180$)  when starting with null network than when compared to other scenarios of $\gamma=0.35$ and $\gamma=0.7$ (mean acts is about $130$).

We observe a reversal of the work needed to reach null network in Figure~\ref{fig:meanacts}(c) where more number of changes is needed to reach null network than reaching the complete network. This can be attributed to the fact that the initial network is already a dense network to start with and it takes relatively less effort to reach the complete network than the null network under appropriate configurations of $\delta$ and $c$.

Initial network density of $0.35$ corresponds to a medium-dense network (Figure~\ref{fig:meanacts}(b)) and hence there is a non-zero effort to reach any of the pairwise stable network under any parameter configuration. However, as in Figure~\ref{fig:meanacts}(a), it takes more effort for players to reach the complete network than the null network.

\section{Analytical Characterization of Topologies of Efficient Networks}\label{sec:Efficiency}

In this section, we study the structure of efficient networks, i.e., networks that maximize the
overall utility, under various conditions of $\delta$ and $c$. First, we begin by introducing a few useful classical results in extremal graph theory and we use these results later in our analysis.

\subsection{Triangles in a Graph}

If three nodes $i$, $j$, and $k$ in $G(V, E)$ are such that $i$
and $j$, $j$ and $k$, $k$ and $i$ are connected by edges, then we
say that nodes $i,j,k$ form a triangle in $G$. The number of
triangles in a simple graph $G$ plays a crucial role in the
computation of payoffs to the nodes and we state here some
classical results. We know from Turan's theorem ~\cite{turan},
that it is possible to have a triangle free graph if the following
holds:
\begin{equation}
\label{turantheorem} e \le \bigg\lfloor
\frac{n^2}{4} \bigg\rfloor
\end{equation}
Here $e$ denotes the number of edges and $n$ the number of vertices of
the graph.  Moreover, from ~\cite{nor:ste}, we know that the number of triangles,
$T$, can be lower bounded, if the number of edges exceed the
above value $\lfloor \frac{n^2}{4} \rfloor $, by
\begin{equation} \label{theorem2}
T \geq \frac{n(4e-n^2)}{9}
\end{equation}

In what follows, we refer to the graph having maximum number of
edges with no triangles as the {\em Turan Graph} and we represent it
by $G_{Turan}$. It is easy to verify that such a graph is a
complete bipartite graph, and the the number of vertices in each
partition differs at most by $1$.



\subsection{Finding the Efficient Graph}
\begin{definition}[Efficient Graph]
The utility ($u(G))$ of a given network $G$ is defined as the sum
of payoffs of all the nodes in that network. That is,
\begin{align}
u(G) &= \sum_{i=1}^{n} u_{i}(G). 
\end{align}
A graph that maximizes the above expression (i.e. sum of payoffs
of nodes) is called an efficient graph.
\end{definition}
We now present a series of results on the topologies of efficient
networks using the proposed framework. These results are based on
different ranges for the values of $\delta$ and $c$.


\begin{proposition}
When $ \delta <  c $ and $\delta^2 < (c - \delta)$, the null graph
is the unique efficient graph.
\end{proposition}
\begin{proof}
For any node $i$, $d_{i} > 0 $ implies that the utility of that
node is negative thus reducing the overall network utility.  This
follows from $(\delta - c + \delta^{2}) $ being negative.
\end{proof}

\begin{proposition}
\label{eff-thm} When $ \delta = c $, the Turan graph is the
unique efficient graph.
\end{proposition}

\begin{proof}
We will analyze the efficiency of an arbitrary graph (denoted by $G$) as follows.
\small{
\begin{align}
\nonumber u(G) &= \sum_{i=1}^{n} u_{i}(G) =
\sum_{i=1}^{n} d_{i} \delta^{2} \left( 1 -
\displaystyle\frac{\sigma_{i}}{\binom{d_{i}}{2}} \right )
\end{align}
\begin{align}
\nonumber &= \delta^{2} \sum_{i=1}^{n} d_{i} - \delta^{2} \sum_{i=1}^{n}
                 \displaystyle\frac{2 \sigma_{i}}{(d_{i} - 1)}
\end{align}
\begin{align}
\nonumber & \le \delta^{2} \sum_{i=1}^{n} d_{i} - \frac{\delta^{2}}{(n-2)}
            \sum_{i=1}^{n} 2 \sigma_{i}
\end{align}
\begin{align}\label{efficiency_equation}
&=\delta^{2} \sum_{i=1}^{n} d_{i} - \frac{\delta^{2}}{(n-2)} (2 \times 3 \times T_{3}(G))
\end{align}}\normalsize
where, $T_{3}(G) $ is the number of triangles in the graph $G$. The last step
of the above simplification is due to the fact that the number of links
among the neighbours of a node $i$ is the number of triangles in the graph in
which node $i$ is one of the vertices of the triangle. The factor $3$ in the
last step is due to the fact that every triangle contributes to the
$\sigma_{i}$ of $3$ nodes. We know that, for an efficient graph, Equation~(\ref{efficiency_equation}) should be maximized and that happens when the
number of triangles in a graph is minimized while simultaneously  maximizing
the number of edges in the graph.

The Turan graph (refer Equation~(\ref{turantheorem})) is a graph with
maximum edges that has no triangles.  So an efficient graph must
have an efficiency greater than or equal to that of a Turan graph.
Thus, it is clear that there is no need to consider graphs with
edges lesser than that of a Turan graph.  Let us consider the case
when a graph (denoted by $\overline{G}$) has more edges than the
Turan graph. Let $\overline{G}$ have $\lfloor \frac{n^2}{4}
\rfloor + x$ edges where $x>0$. From
Equation~(\ref{efficiency_equation}), we know that \small{
\begin{align}\label{eqn5.5}
\nonumber u(\overline{G}) & = \sum_{i = 1}^{n}u_{i}(G) = \delta^{2} \displaystyle\sum_{i=1}^{n} d_{i} - \delta^{2} \displaystyle\sum_{i=1}^{n} \frac{2 \sigma_{i}}{(d_{i} - 1)} \\
&\le \delta^{2} \left(2 \left(\bigg\lfloor \frac{n^2}{4}
\bigg\rfloor +x\right)\right) - \frac{\delta^{2}}{(n-2)} (6
T_{3}(\overline{G}))
\end{align}}\normalsize
where $T_{3}(\overline{G})$ is the number of triangles in
$\overline{G}$. From Equation~(\ref{theorem2}), we have
\small{\begin{align}\label{eff_g_dash} u(\overline{G}) & \le
\delta^{2} \left(2\left( \bigg\lfloor \frac{n^2}{4} \bigg\rfloor
+x\right)\right) - \frac{\delta^{2}}{(n-2)} \left(6 n
\left(\frac{4e-n^2}{9}\right) \right)
\end{align}}\normalsize
Since $T_3(G_{Turan}) = 0$, the efficiency of the Turan graph is:
\small{\begin{align} u(G_{Turan}) = \sum_{i}
u_{i}(G_{Turan}) &=\delta^{2} \left(2\times \bigg\lfloor
\frac{n^2}{4} \bigg\rfloor \right) 
\end{align}}\normalsize

The change in efficiency ($\Delta{u}$) between the two graphs is
\begin{equation}
\label{deltagain}
\Delta{u} = u(\overline{G}) - u(G_{Turan}) \le 2 \delta^2 \left (x  - \frac{n}{(n-2)} \frac{4x}{3} \right)
\end{equation}
which is clearly negative for any $x>0$. This implies that the
Turan graph is the unique efficient graph.
\end{proof}

\begin{proposition}
When $ \delta < c $ and $ \delta^2 > (c - \delta) $, the Turan
graph is the unique efficient graph.
\end{proposition}
\begin{proof}
We prove this by contradiction. Assume that $\overline{G}$ is any
graph other than the Turan graph and $\overline{G}$ is efficient.
We show below that $\overline{G}$ cannot have lesser number of edges than $G_{turan}$,

\begin{align}
\nonumber u(\overline{G}) &= \sum_{i=1}^{n} u_{i}(\overline{G}) =
(\delta-c)\sum_{i=1}^{n} d_i + \sum_{i=1}^{n} d_{i} \delta^{2}
\left( 1 - \displaystyle\frac{\sigma_{i}}{\binom{d_{i}}{2}} \right)  \\
\nonumber & \le \left( \delta - c + \delta^2 \right)   \sum_{i=1}^{n} d_i \\
\nonumber & <  u(G_{turan}) \mbox{~whenever,~} \sum_{i=1}^{n} d_i < 2 \bigg\lfloor \frac{n^2}{4} \bigg\rfloor
\end{align}

And observe, if  $\overline{G}$ has same number of edges as
$G_{turan}$ and is different from it, it can contain triangles and
will have an utility less than that of $G_{turan}$, as the benefit
from bridging would go down and the benefit from direct links would
remain unchanged.

Thus $\overline{G}$ contains more edges than $G_{turan}$.
Observe, that the benefit from direct links is negative
$(\delta - c) \sum_{i=0}^{n} d_i < 0 $, and
$\overline{G}$ has an higher utility compared to that of $G_{turan}$. It has to be that the
bridging benefits in $\overline{G}$ has to be greater than that of the
Turan graph, as the utility due to direct links term has become more
negative compared to its value in $G_{turan}$

\begin{align}
\nonumber u(\overline{G}) &= \sum_{i=1}^{n} u_{i}(\overline{G}) =
\underbrace{(\delta-c)\sum_{i=1}^{n} d_i}_{\text{negative}} +
\underbrace{\sum_{i=1}^{n} d_{i} \delta^{2} \left( 1 -
\displaystyle\frac{\sigma_{i}}{\binom{d_{i}}{2}} \right
)}_{\text{utility more than } G_{Turan}}
\end{align}

This implies that this graph would give a higher utility for the
$\delta = c$ case, as the first term is $0$ there. This
contradicts Theorem \ref{eff-thm} and so our assumption must be
wrong. Hence the Turan graph is efficient.
\end{proof}

\begin{table}[h]
\vspace{0.1in} \centering
\begin {tabular} {||l||l||}
\hline
\hline
{\textbf{Parameter Range}} & {\textbf{Efficient Topologies}} \\
\hline
$ \delta <  c $ and $\delta^2 < (c - \delta)$ & Null network \\
\hline
$ \delta < c $ and $ \delta^2 > (c - \delta) $  & Turan network \\
\hline
$ \delta = c $ & Turan network \\
\hline
$ \delta > c $ and $\delta^2 > 3(\delta - c) $ & Turan network \\
\hline
$ \delta > c $ and $ (\delta - c ) > 2\delta^2$ & Complete network  \\
\hline
\hline
\end {tabular}
\caption{Characterization of Topologies of Efficient Networks in NFLP}\label{summarytable3}
\end{table}

\begin{proposition}
When $ \delta > c $ and $\delta^2 \geq 3(\delta - c) $, the Turan
graph is the unique efficient graph.
\end{proposition}
\begin{proof}
Let $\overline{G}$ be the efficient graph. Using a similar
analysis that lead to Equation~(\ref{eff_g_dash}), we can see that
\small{
\begin{align}\label{eff_g_dash1}
\nonumber u(\overline{G}) &\le (\delta+c+\delta^{2}) \left(2\left(\displaystyle \bigg\lfloor \frac{n^2}{4} \bigg\rfloor +x\right)\right) - \frac{\delta^{2}}{(n-2)} \left(6 n \left( \frac{4e-n^2}{9}\right)\right) \\
&= (\delta+c+\delta^{2}) \left(2 \left(\displaystyle \bigg\lfloor
\frac{n^2}{4} \bigg\rfloor +x\right)\right) -
\frac{\delta^{2}n}{(n-2)} \left(\frac{8x}{3}\right)
\end{align}
} \normalsize{\noindent For the Turan graph, it can also be seen
by simple analysis that } \small{
\begin{align}
\nonumber u(G_{Turan}) &= \displaystyle 2\bigg\lfloor \frac{n^2}{4} \bigg\rfloor \left( \delta-c+\delta^2 \right) \\
\nonumber \Rightarrow u(\overline{G})-u(G_{Turan}) & \le 2x\left( (\delta-c+\delta^2) - \displaystyle \frac{4n\delta^2} {3(n-2)} \right) \\
 & < 2x\left( (\delta-c+ \delta^2) - \displaystyle \frac{4\delta^2}{3} \right)
\end{align}
} \normalsize{\noindent Thus, when $\delta^2 \geq 3(\delta-c)$, the Turan graph is
    the unique efficient graph.}
\end{proof}





\begin{proposition}
When $ \delta > c $ and $ (\delta - c ) > 2\delta^2$ , the
complete graph is the efficient graph.
\end{proposition}
\begin{proof}
It can be shown that starting with an arbitrary graph
$\overline{G}$ (which is not a complete graph), adding an edge
between two nodes $i$ and $j$ (with smallest degree) increases the
cumulative utility of these two nodes by at least $2\delta^2$. At
the same time, there is a decrease in utility of a \textit{common}
neighbour of nodes $i$ and $j$, say node $k$, as there is a
decrease in the bridging benefits of node $k$. It can be shown
that the cumulative decrease in utility of all such common
neighbours formed is $\displaystyle\frac{2\delta^2}{d_k-1} min
(d_i,d_j)$ which is less than equal to $2\delta^2$. Repeating the
above process, we get the complete network.
\end{proof}

\begin{conjecture}\label{conj1}
When $ \delta > c $ and $ (\delta - c) \leq \delta^2 < 3(\delta -
c) $, the Turan graph is the efficient graph.
\end{conjecture}

\begin{conjecture}\label{conj2}
When $ \delta > c $ and $ (\delta - c) \le 2\delta^2 $:

    (i) if $(\delta - c) > \frac{n}{n-2} \delta^2$, then the complete graph is the efficient
    graph.

    (ii) if $(\delta - c) < \frac{n}{n-2} \delta^2$, then the Turan graph is the efficient
    graph.
\end{conjecture}

We summarize the above results on efficiency in
Table~\ref{summarytable3}.


\section{Price of Stability \textbf{(P}\small{\textbf{o}}\normalsize\textbf{S)} of the Proposed Model }\label{POS}
Recall that PoS \cite{anshelevich:08} is the ratio of the sum of payoffs of the players in a best pairwise stable network
to that of an efficient network. In NFLP, a best pairwise stable network means a pairwise stable network with a maximum value of the sum of payoffs of the players. By invoking the results derived in the previous sections, we now  present our results on PoS for the proposed model.
\begin{theorem}\label{pos-thm1} The price of stability (PoS) is $1$ in each of the following
scenarios: \\
\noindent(i) $ \delta > c $ and $ (\delta - c ) > 2\delta^2$, \\
(ii) $ \delta > c $, $\delta^2 > (\delta - c) $ and $\delta^2 \geq 3(\delta - c) $, \\
(iii)$ \delta = c $, \\
(iv) $ \delta < c $ and $ \delta^2 > (c - \delta) $.
\end{theorem}

This theorem can be proved easily using the results summarized
in Table~\ref{summarytable2} and Table~\ref{summarytable3}.

\noindent Note: Since the null network is the only efficient network
when $\delta < c $ and $\delta^2 < (c - \delta)$, PoS is not
defined in this region.

In view of Conjecture~\ref{conj1}, the following result presents
bounds on PoS.

\normalsize
\begin{proposition}
\label{pos-thm2} When $ \delta > c $ and $ (\delta - c) \leq
\delta^2 < 3(\delta - c) $, PoS $ > \frac{1}{2}$.
\end{proposition}
\begin{proof}
We know that, under the conditions $ \delta > c $ and $ (\delta - c) <\delta^2 < 3(\delta - c) $, the pairwise stable graph with the highest utility is the Turan graph (as seen from Table~\ref{summarytable2}).
Let Conjecture~\ref{conj1} be false.
In this scenario, let us denote the efficient graph by
$\overline{G}$. We will now evaluate an upper bound on the maximum
efficiency of $\overline{G}$. $\overline{G}$ has to have more
direct links than the Turan graph (as $\delta > c $) to be a
candidate for efficient graph. Let $\overline{G}$ have
$\left(\displaystyle \bigg\lfloor \frac{n^2}{4} \bigg\rfloor +
x\right)$ edges where $x>0$. \vspace{-0.1in} \small{
\begin{align}\label{efficiency_equation2}
\nonumber u(\overline{G}) = \sum_{i=1}^{n} u_{i}(\overline{G}) = (\delta-c)\sum_{i=1}^{n} d_i+ \sum_{i=1}^{n} d_{i} \delta^{2} \left( 1 - \displaystyle\frac{\sigma_{i}}{\binom{d_{i}}{2}} \right ) \\
\nonumber = (\delta-c+\delta^{2})\sum_{i=1}^{n} d_i - \delta^{2}\left(\displaystyle\frac{2\sigma_{i}}{d_i-1} \right )
\end{align}
} \normalsize Since  $d_i$ can be at most $(n-1)$, \small{
\begin{align}
\nonumber u(\overline{G}) \leq (\delta-c+\delta^{2}) n(n-1) - \left(\frac{2\delta^{2}}{n-2}\right) \sum_{i=1}^{n}  \sigma_i \\
\nonumber u(\overline{G}) \leq (\delta-c+\delta^{2}) n(n-1) -
\left(\frac{2\delta^{2}}{n-2}\right) T_3(\overline{G})
\end{align}
} \normalsize By Equation~(\ref{theorem2}), we have \small{
\begin{align}
\nonumber u(\overline{G}) &\leq  (\delta-c+\delta^{2}) n(n-1) - \left(\frac{2\delta^{2}}{n-2}\right) \left(\displaystyle \frac {n(4e-n^2)}{9}\right) \\
\nonumber &= (\delta-c+\delta^{2}) n(n-1) -
\left(\frac{\delta^{2}n}{n-2}\right) \left(\displaystyle \frac
{8x}{9}\right) \vspace{-0.1in}
\end{align}
} \normalsize Since
$\displaystyle\left(\frac{\delta^{2}n}{n-2}\right)
\left(\displaystyle \frac {8x}{9}\right) > 0$, we have \small{
\begin{align}
\nonumber u(\overline{G}) \leq  (\delta-c+\delta^{2}) n(n-1)
\end{align}
} \normalsize
The Turan graph is pairwise stable under these conditions (refer
Table~\ref{summarytable2}). Hence we get the following:
\small{
\begin{align}
\nonumber u(G_{Turan}) & = (\delta-c+\delta^2) \left(\displaystyle 2 \bigg\lfloor \frac{n^2}{4} \bigg\rfloor \right) \\
\nonumber PoS & \ge \frac{u(G_{Turan})}{u(\overline{G})} \geq
\displaystyle \frac{(\delta-c+\delta^2)
\left(\displaystyle\frac{n^2 - 1}{2}\right)}{(\delta-c+\delta^{2})
n(n-1) } = \displaystyle \frac{1}{2} + \frac{1}{2n}
\end{align}
} \normalsize This implies that $ PoS > \frac{1}{2}$.
\end{proof}
\textit{Remark:} In view of Conjecture~\ref{conj2}, it can be
noted that a similar bound can be obtained in the region
\text{$\delta
> c$ and $(\delta-c) \le 2\delta^2$}. The details are not provided
here due to space constraints.

From Theorem \ref{pos-thm1} and Theorem \ref{pos-thm2} along with the simulation results, we conclude that, under mild conditions, the proposed NFLP produces efficient networks that are pairwise stable. This is desirable from the view of
system design.




%

\section{Conclusions and Future Work}\label{conclusion}

In this paper, we proposed a network formation game with localized payoffs (NFLP) and studied the topologies of pairwise stable and
efficient networks. We gained additional insights about the network formation process through detailed simulations.
We also studied the tradeoff between pairwise stability and efficiency using the notion of PoS. In particular, we computed the PoS of the proposed NFLP. Except for a few configurations of $\delta$ and $c$, we have shown that PoS is $1$. This means that, under mild conditions, that NFLP produces efficient networks that are pairwise stable.

In the utility function we defined in Section~\ref{utilitymodel}, the payoff of any node had two components - benefit from direct links and benefit from bridging. The pairwise stable network topologies of our model (Section \ref{sec:Stability}) shows that there are no bridges in the equilibrium networks. Bridges can also be considered as bottlenecks of information flow. Since every node is striving to obtain a bridging position there are no bridges in the equilibrium networks, this suggests that the proposed utility model avoids bottlenecks in decentralized network formation. Here are a few pointers for future work. First, the framework in this paper can be extended to the case of directed graphs and weighed graphs. This involves certain challenges such as defining the utility model appropriately. Second, the setting in this paper can be extended by varying the notions of stability and efficiency. We note that there are several possible notions of stability and efficiency that exist in the literature. The choice of an appropriate notion of stability as well as efficiency is a topic of debate.

Further, our model gives us some valuable hints at the networks formed in real world as well. Some noted work in complex network literature has observed the emergence of bipartite graphs in real world scenarios \cite{albert:02, newmanstrogatzwatts}. An important example has been the class of collaboration networks. It has been observed that the network of actors basically is a uni-mode bipartite graph \cite{newmanstrogatzwatts}. Other important examples of real world bipartite networks include boards of directors of companies, co-ownership networks of companies and collaboration networks of scientists and movie actors. In the analysis of our proposed model in this paper, we have seen the emergence of important graph structures like the Turan graph and in general, bipartite graphs and $k$-partite graphs during the network formation process under many configurations. Though our model does not precisely solve the difficult problem of identification of \textit{all} parameters affecting network formation, it nevertheless offers valuable hints about some of the important  parameters affecting real world network formation. The studies on our utility model of network formation also offers strong evidence that incorporation of important game theoretic concepts like pairwise stability is vital to the understanding of complex network formation behaviour.

It is the goal of our future work to expand the horizon of our understanding of other class of real world networks namely the Internet (or the world wide web), epidemic networks, friendship networks, power grid networks, etc, and propose suitable strategic complex network formation models that, at least, approximately imitate the formation behaviour of some of these important real world networks.

\small{
\bibliographystyle{IEEEtran}
\bibliography{SNF-finalversion-SocialNetworks}}

\end{document}